\documentclass[reqno,12pt,letterpaper]{amsart}
\usepackage[usenames,dvipsnames]{color}
\usepackage[colorlinks=true,linkcolor=blue,citecolor=blue]{hyperref}
\usepackage{amsxtra}
\usepackage[alphabetic,nobysame]{amsrefs}
\usepackage{autonum}
\usepackage{subcaption}
\usepackage{soul,ulem}
\usepackage{amsmath,amssymb,amsthm,graphicx,mathrsfs,url}
\usepackage{slashed}
\usepackage{tikz}
\usepackage{comment}
\usepackage{placeins}
\usetikzlibrary{shadings}
\usepackage{floatrow}

\newtheorem{prop}{Proposition}[section]
\newtheorem{theo}[prop]{Theorem}\usepackage[top=1in,left=1in,right=1in,bottom=1in]{geometry}
\numberwithin{equation}{section}

\usepackage[top=1in,left=1in,right=1in,bottom=1in]{geometry}
\newtheorem{example}{Example}
\newtheorem{lemm}[prop]{Lemma}

\theoremstyle{definition}

\newtheorem{rem}[prop]{Remark}

\newcommand{\de}{\ensuremath{\partial}}

\newcommand{\R}{{\mathbb R}}
\newcommand{\Di}{{\slashed{D}}}
\newcommand{\tDi}{{\tilde{\slashed{D}}}}

\newcommand{\hf}{{\hat{f}}}
\newcommand{\mCS}{\mathbb S}
\newcommand{\anut}{\aver{\nu}_{\rT}}
\newcommand{\vsa}{\varsigma}
\newcommand{\TT}{}
\newcommand{\m}{}

\date{\today}
\newcommand{\aw}[1]{{{\color{blue}{AW: #1}}}}

\newcommand{\slb}[1]{{{\color{orange}{SLB: #1}}}}
\newcommand{\alexis}[1]{{{\color{purple}{AD: #1}}}}
\newcommand{\jl}[1]{{{\color{brown}{[JL: #1]}}}}
\newcommand{\gb}[1]{{{\color{blue}{GB: #1}}}}

\def\smallsection#1{\smallskip\noindent\textbf{#1}.}

\title[Magnetic slowdown of topological edge states]{Magnetic slowdown of topological edge states}

\author{G. Bal}
\address[Guillaume Bal]{University of Chicago, USA.}
\email{guillaumebal@uchicago.edu}

\author{S. Becker}
\address[Simon Becker]{Courant Institute for Mathematical Sciences, New York City, USA.}
\email{simon.becker@cims.nyu.edu}

\author{A. Drouot}
\address[Alexis Drouot]{University of Washington, USA.} 
\email{adrouot@uw.edu}

\begin{document}

\maketitle

\newcommand{\eps}{\varepsilon}
\newcommand{\Rm}{\mathbb{R}}
\newcommand{\Cm}{\mathbb{C}}
\newcommand{\Zm}{\mathbb{Z}}
\newcommand{\Nm}{\mathbb{N}}
\newcommand{\tGamma}{\widetilde{\Gamma}}
\newcommand{\Mm}{\mathbb{M}}
\newcommand{\Sm}{\mathbb{S}}
\newcommand{\dsum}{\displaystyle\sum}
\newcommand{\dint}{\displaystyle\int}
\newcommand{\fa}{\mathfrak a}
\newcommand{\Ub}{{\mathbf U}}
\newcommand{\aver}[1]{\langle #1 \rangle}
\newcommand{\mA}{\mathcal A}
\newcommand{\mC}{\mathcal C}
\newcommand{\mL}{\mathcal L}
\newcommand{\mR}{\mathcal R}
\newcommand{\mS}{\mathcal S}
\newcommand{\mT}{\mathcal T}
\newcommand{\mV}{\mathcal V}
\newcommand{\fU}{\mathfrak U}
\newcommand{\bU}{\operatorname U}
\newcommand{\rT}{{\rm T}}
\newcommand{\fb}{\mathfral{b}} 
\newcommand{\zetatwo}{\zeta}
\newcommand{\xione}{\xi}
\newcommand{\xinorm}{\xi}
\newcommand{\ynorm}{\zeta}
\newcommand{\Xitot}{\xi,\zeta}
\newcommand{\fhat}{\hat f}
\newcommand{\C}{\mathbb{C}}
\newcommand{\Z}{\mathbb{Z}}
\newcommand{\epsi}{\varepsilon}
\newcommand{\N}{\mathbb{N}}
\newcommand{\p}{\partial}
\newcommand{\te}{\theta}
\newcommand{\RR}{\mathcal{R}}
\newcommand{\trace}{\operatorname{tr}}
\newcommand{\matrice}[1]{\left[  \begin{matrix} #1 \end{matrix}\right]}
\newcommand{\systeme}[1]{\left\{  \begin{matrix} #1 \end{matrix}\right.}
\newcommand{\lr}[1]{\left\langle #1 \right\rangle}
\newcommand{\loc}{{\operatorname{loc}}}
\newcommand{\HH}{{\mathcal{H}}}
\newcommand{\vp}{{\varphi}}
\newcommand{\hh}{{\mathfrak{h}}}
\newcommand{\aaa}{{\mathfrak{a}}}
\newcommand{\SSS}{{\mathcal{S}}}
\newcommand{\VV}{{\mathcal{V}}}
\newcommand{\tL}{{\tilde{L}}}
\newcommand{\tf}{{\tilde{f}}}
\newcommand{\tSSS}{{\tilde{\SSS}}}
\newcommand{\sgn}{\operatorname{sgn}}

\begin{abstract}
 We study the propagation of wavepackets along curved interfaces between topological, magnetic materials. Our Hamiltonian is a massive Dirac operator with a magnetic potential. We construct semiclassical wavepackets propagating along the curved interface as adiabatic modulations of straight edge states under constant magnetic fields. While in the magnetic-free case, the wavepackets propagate coherently at speed one, here they experience slowdown, dispersion, and Aharonov--Bohm effects. Several numerical simulations illustrate our results.
\end{abstract}



\section{Introduction}\label{sec:intro}

This paper analyzes wavepackets propagating along an interface between two topologically distinct materials, in the presence of an external magnetic field. It extends constructions carried out in \cite{bal2021edge} for magnetic-free models. We represent here the electron dynamics via a two-dimensional Dirac equation:
\begin{equation}\label{eq:D1}
    (\eps D_t + \slashed D)\Psi(t,x) = 0, \quad (t,x) \in \mathbb R \times \mathbb R^2.
\end{equation}
In \eqref{eq:D1},  $\slashed D$ denotes a Dirac operator with sign-changing mass and magnetic field:
\begin{equation}\label{eq:Dirac}
  \slashed D = (\eps D_1-A_1(x))\sigma_1+(\eps D_2-A_2(x))\sigma_2 + \kappa(x)\sigma_3, \qquad \text{where:}
\end{equation}
\begin{itemize}
\item  $\epsi > 0$ is a small semiclassical parameter and $\epsi D_t=-i \epsi \partial_t$, $\epsi D_j=-i \epsi\partial_j$ denote the self-adjoint semiclassical derivatives;
\item $A=(A_1,A_2)^t \in C^\infty(\R^2,\R^2)$ is a magnetic potential with $\nabla A \in C^\infty_b$ (i.e.\ smooth with all derivatives uniformly bounded), inducing the magnetic field $B=\partial_1A_2-\partial_2A_1$; 
\item $\kappa \in C^\infty(\R^2,\R)$ has varying sign and satisfies $\nabla\kappa\in C^\infty_b$;
\item $\sigma_1, \sigma_2, \sigma_3$ are the standard $2\times2$ Pauli matrices.
\end{itemize}

The sign of the domain wall $\kappa$ characterizes the topological phase of the material. Specifically, under the transversality condition
\begin{equation}\label{eq-3a}
    \kappa(y) = 0 \qquad \Rightarrow \qquad \nabla \kappa(y) \neq 0,
\end{equation}
the interface $\Gamma = \kappa^{-1}(0)$ separates regions of distinct local topology \cites{B19b,bal2021edge}. For analytic reasons, we consider here a uniform version of \eqref{eq-3a}:
\begin{equation}
    \inf\{ \big|\nabla \kappa(x)\big| : x \in \Gamma\} >0.
\end{equation}

The bulk-edge correspondence for Dirac operators \cite{B19b} predicts that a tubular neighborhood of $\Gamma = \kappa^{-1}(0)$ supports asymmetric currents, hence (some analogue of) edge states for small $\epsi$. For vanishing magnetic potentials, we constructed in \cite{bal2021edge} long-lived solutions to \eqref{eq:D1}. These were confined and propagating at speed one along $\Gamma$. We referred to them as dynamical edge states, since we derived them as time-dependent adiabatic modulations of straight edge states. This paper extends the construction to the magnetic case. A non-zero magnetic field induces several new phenomena for dynamical edge states:
\begin{itemize}
    \item[(i)] A systematic slowdown, see Figure \ref{fig:my_label};
     \item[(ii)] A large phase-shift, generating a Aharonov--Bohm effect when $\Gamma$ is a loop;
     \item[(iii)] In general, a mesoscopic dispersion along $\Gamma$.
\end{itemize}
We set a few notations in \S\ref{sec-1.1}, state a simplified main result in \S\ref{sec-1.2}, and detail the effects (i)--(iii) in \S\ref{sec-1.3}.

\vspace{-6mm}

\begin{figure}[ht]
\floatbox[{\capbeside\thisfloatsetup{capbesideposition={right,center},capbesidewidth=10cm}}]{figure}[\FBwidth]
{\hspace{-1cm}\caption{\label{fig:my_label} 
Snapshots of the evolution of Gaussian wavepackets propagating along a straight interface $\kappa(x)=x_2$, under a constant magnetic field $B=0,0.5,1,1.5,2$ (from right to left), computed numerically. The packets slow down in stronger fields.}}
{\begin{tikzpicture}
   \node at (0,0) {\includegraphics[width=8cm]{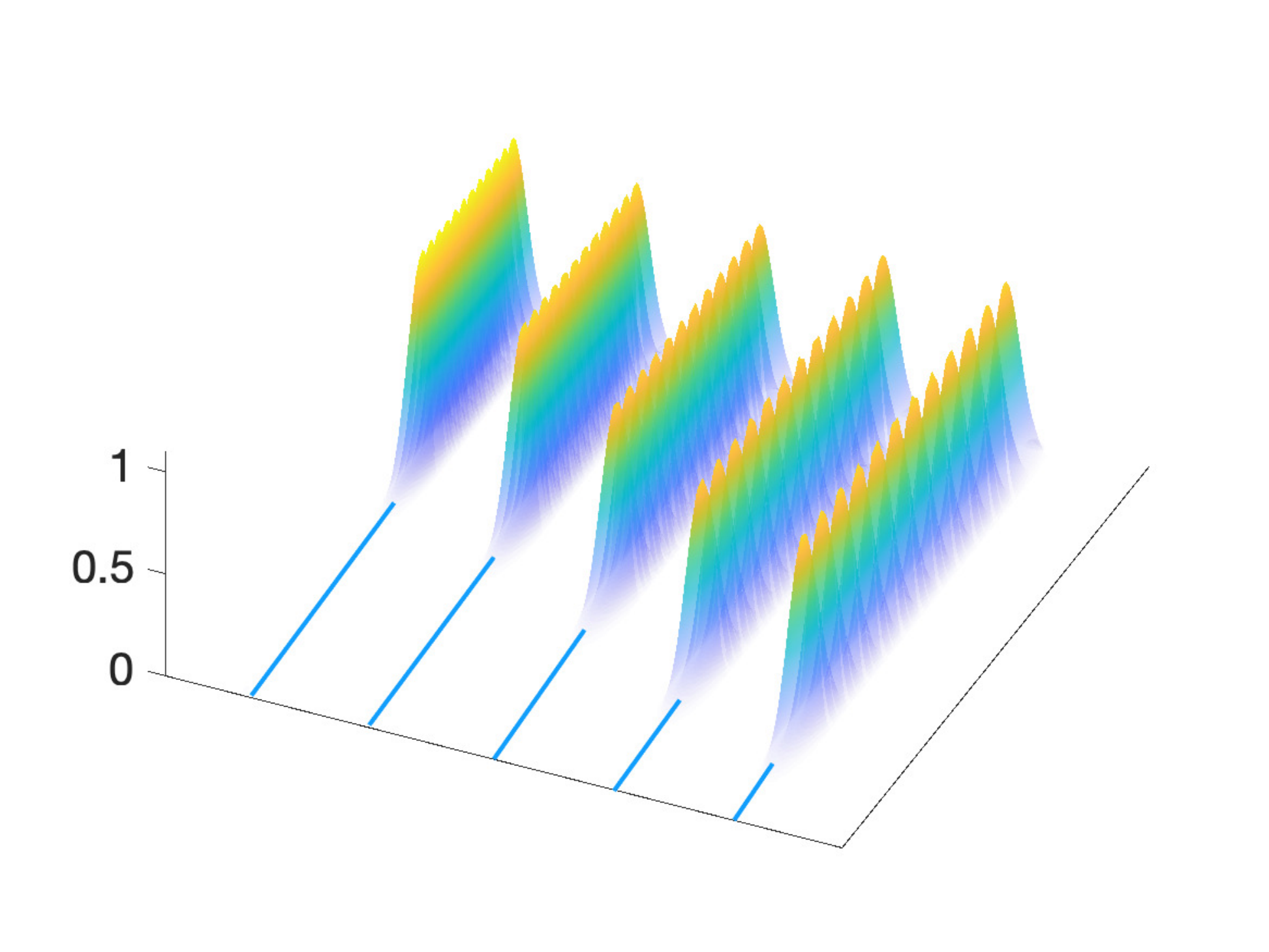}};
 \node[blue!70] at (.7,-2.5) {$B=0$};
 \node[blue!70] at (-2.5,-1.8) {$B=2$};
  \end{tikzpicture}}
\end{figure}

\vspace{-.5cm}

\subsection{Notations}\label{sec-1.1} We first define the normal and tangent vector fields to the level sets of $\kappa$:
\begin{equation}
\label{eq:tang_vec}
    n(x) = \frac{\nabla \kappa(x)}{|\nabla \kappa(x)|}, \qquad \tau(x) = J n(x), \qquad J=\matrice{ 0&-1\\ 1 & 0}.
\end{equation}

Given $y_0 \in \Gamma$, we let $y_t$ be the solution of the ODE
\begin{equation}\label{eq:yt}
    \dot y_t = c(y_t) \tau(y_t), \qquad c(y) := \dfrac{|\nabla \kappa(y)|}{\sqrt{|\nabla \kappa(y)|^2 + B(y)^2}},
\end{equation}
where $B = \nabla \times A$ is the magnetic field. Note that $\p_t \kappa(y_t) = 0$ since $\nabla \kappa(x) \cdot \tau(x) = 0$, hence $y_t \in \Gamma$ for any $t \in \R$. We then define the quantities
\begin{equation}
   B_t=B(y_t), \quad n_t = \frac{\nabla \kappa(y_t)}{|\nabla \kappa(y_t)|}, \quad r_t = |\nabla \kappa(y_t)|, \qquad \rho_t = \sqrt{r_t^2+B_t^2}, \quad \gamma_t = \frac{B_t}{\rho_t^2}.
\end{equation}

To capture the local geometry of the interface near $y_t$, we define two smoothly varying angles $\varphi_t$ and $\te_t$ (and the corresponding clockwise rotation $R_{\te_t}$) by 
\begin{equation}\label{eq:coefst}
     \cos \varphi_t= \frac{r_t}{\rho_t} := c_t, \quad \sin \varphi_t= \frac{B_t}{\rho_t}:= s_t, 
     \quad R_{\theta_t} = \matrice{ \cos\theta_t & \sin\theta_t \\ -\sin\theta_t & \cos\theta_t },
\end{equation}
see Figure \ref{fig:geom}; we assume $\varphi_0$ and $\theta_0$ belong to $[0,2\pi)$ for concreteness.
With these notations, $R_{\theta_t} n_t=e_2$ and $\dot{y_t} = c_t \tau_t$ with $\tau_t=Jn_t$. 
Differentiating the equation defining $\theta_t$, we observe that $\dot n_t+\dot\theta_t\tau_t=0$. Since $\dot y_t=c_t\tau_t$, we deduce that $\dot\theta_t=-c_t (\tau\cdot\partial_\tau n)(y_t)$. In particular, $|\dot\theta_t|=c_t K_t$ where $K_t$ is the curvature of the curve $\Gamma$ at $y_t$. 

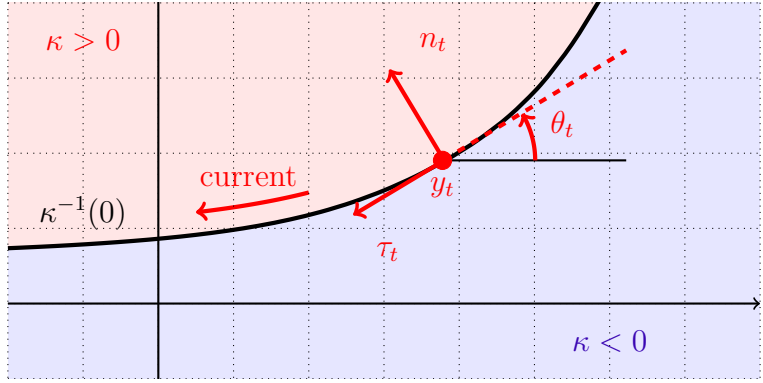
\begin{figure}[htbp]
\newcommand{\fun}{{\x/2.5+(1+\x*\x/6.25)^(1/2)-1.5 +  tanh(\x-3) +1 }}

\definecolor{dblue}{RGB}{57, 5, 179}

\begin{tikzpicture}
\clip (-5,-3) rectangle (5,2);

\fill [red!10, domain=-5:5, variable=\x]
  (-5, 3)
  -- plot ({\x}, {\fun})
  -- (5, 3)
  -- cycle;

\fill [blue!10, domain=-5:5, variable=\x]
  (-5, -4)
  -- plot ({\x}, {\fun})
  -- (5, -4)
  -- cycle;

\draw[domain=-5:5 , smooth, variable=\x, ultra thick] plot (\x, {\fun});

\draw[domain=-1:-2.5 , smooth, variable=\x, ultra thick,red,->] plot (\x, {.3+\x/2.5+(1+\x*\x/6.25)^(1/2)-1.5 +  tanh(\x-3) +1 });

\node[red] at (-1.8, -.3) {current};

\draw[dotted, thin] (-5,3) grid (5,-4);
\draw[thick,->] (-3,-4) -- (-3,3);
\draw[thick,->] (-5,-2) -- (5,-2);
\node[red] at (-4, 1.5) {$\kappa > 0$};
\node[dblue] at (3, -2.5) {$\kappa < 0$};
\node at (-4, -.8) {$\kappa^{-1}(0)$};

\begin{scope}[shift={(-1.4,-.7)},scale=1.1]

\begin{scope}[shift={(2,0.55)},scale=1.1]
\draw[domain=0:2 , smooth, variable=\x, ultra thick, red, dashed] plot (\x, {1.21+0.601*(\x-2)});
\draw[red, ultra thick, ->] (0,0) -- (-1,-0.601); 
\draw[red, ultra thick, ->] (0,0) -- (-0.601,1); 
\draw[thick] (0,0) -- (2,0);
\draw[red, domain=0:30.9 , smooth, variable=\x, ultra thick, ->] plot ({cos(\x)}, {sin(\x)});
 \node[red] at (1.3, .4) {$\theta_t$};
\end{scope}

\begin{scope}[shift={(0,0.55)},scale=1.1]
 \node[red] at (1.8, -.3) {$y_t$};
 \draw[red,fill=red] (1.8, 0) circle (1mm);
\node[red] at (1.7, 1.3) {$n_t$};
\node[red] at (1.2, -1) {$\tau_t$};
\end{scope}

\end{scope}
\end{tikzpicture}
\caption{Geometry of curved interface.} \label{fig:geom} 
\end{figure}

We also introduce the unitary pullback operator $\RR_{\te_t}$ by $R_{\te_t}$ as $\RR_{\te_t} g(z) = g(R_{\te_t} z)$. Associated to $\te_t$ and $\varphi_t$ are two spinorial rotations,
\begin{equation}
  U_{2,\varphi_t} = e^{-i\frac{\varphi_t}2\sigma_2}, \qquad U_{3,\theta_t} = e^{-i\frac{\theta_t}2\sigma_3}.
\end{equation}

\subsection{Simplified main result}\label{sec-1.2} For the sake of simplicity, we consider a restricted setup: the domain wall satisfies a geometric condition and the magnetic field is constant. Our main result, \ref{thm:main}, takes then a clearer form.

The aforementioned assumptions are:
\begin{itemize}
    \item $\kappa \in C^\infty(\R^2,\R)$ with $\nabla \kappa \in C^\infty_b$
    satisfies
\begin{equation}\label{eq:00g}
    y \in \kappa^{-1}(0) \quad \Rightarrow \quad \big| \nabla \kappa(y) \big| = 1, \quad \Delta \kappa(y) = 0; 
\end{equation}
\item The magnetic field $B$ is constant.
    \end{itemize}
While \eqref{eq:00g} is analytically restrictive, it is not geometrically restrictive: any one-dimensional submanifold of $\R^2$ is the zero set of a function $\kappa$ satisfying \eqref{eq:00g}. 
For instance, for the straight edge $\Gamma = \R e_2$, we can choose $\kappa(y)=y_2$ while for the unit circle $\Gamma = \mathbb{S}^1$ we can choose $\kappa(y)=\ln |y|$ (near $|y|=1$). When the  two above conditions hold, the quantities $r_t, B_t, \rho_t, \gamma_t, \varphi_t, c_t$ and $s_t$ do not depend on $t$, and we omit the subscript $t$.

The most elementary setup with these two conditions consists of $\kappa(x) = x_2$ and $A = -B x_2 e_1$. The corresponding Dirac operator is:
\begin{equation}\label{eq-2a}
    \slashed{D}_{0,B} := ( \epsi D_1+Bx_2) \sigma_1 + \epsi D_2 \sigma_2 + x_2 \sigma_3.
\end{equation}
We remark that the magnetic potential vanishes along the interface $\R e_1$ and is parallel to it. The equation $(\epsi D_t + \slashed{D}_{0,B}) \Psi = 0$ admits an explicit family of non-dispersive wavepacket solutions that propagate along $\Gamma$ at speed $c = (1+B^2)^{-1/2}$:
\begin{equation}\label{eq:00i}
    \Psi_{0,B}(x) := \frac{1}{\sqrt\eps} \psi_{0,B}\left( \dfrac{x-ct e_1}{\sqrt{\epsi}} \right), \quad \psi_{0,B}(z) := \int_\R e^{i\xione z_1-\frac{1}{2}\rho ( z_2+\gamma \xione)^2} \fhat(\xione) d\xione \cdot U_{2,\varphi}\matrice{1  \\ -1},
\end{equation}
where $\fhat$ is any Schwartz-class function, $\rho = \sqrt{1+B^2}$, $\gamma = \frac{B}{1+B^2}$, and $\varphi=\arctan B$. 

We extend this statement to interfaces tilted by an angle $\te \in \R$. Let $\slashed{D}_{\te,B}$ be the Dirac operator with domain wall and magnetic potential
\begin{equation}
    \kappa(x) = n \cdot x, \quad A(x) = B \kappa(x) \tau, \quad n = \matrice{-\sin \te \\ \cos \te}, \quad \tau = -\matrice{\cos \te \\ \sin \te}.
\end{equation}
The equation $(\epsi D_t + \slashed{D}_{\te,B}) \Psi = 0$ is unitarily equivalent to the case $\te=0$ and admits a family of solutions constructed from $\Psi_{0,B}$:
\begin{equation}
    \Psi_{\te,B}(x) := \frac{1}{\sqrt\eps} \psi_{\te,B}\left( \dfrac{x-ct \tau}{\sqrt{\epsi}} \right), \qquad  \psi_{\te,B}(z) := U_{3,\te} \RR_{\te} \psi_{0,B}(z). 
\end{equation}

Our main result produces approximate solutions to $(\epsi D_t + \slashed{D}) \Psi = 0$ as modulations of $\psi_{\te,B}$. To state it, we need the distribution $g_t$ on $\R^2$ defined by:
\begin{equation}\label{eq:00j}
    g_t(z) := \int_{\R^2} e^{i\zeta z + i \gamma (\te_t-\te_0) \zeta_1^2} d\zeta = \dfrac{e^{-i\pi/4}}{|2\pi\gamma (\te_t-\te_0)|^{1/2}} e^{i \frac{-z_1^2}{4 \gamma (\te_t-\te_0)}} \cdot \delta_0(z_2).
\end{equation}
The second equality is valid for $\te_t \neq \te_0$ and should be replaced by a Dirac mass when $\te_t = \te_0$. Our simplified theorem for constant magnetic fields reads then as follows.

\begin{theo}\label{thm:5} Assume that $\kappa$ satisfies the condition \eqref{eq:00g} and that $B$ is constant. For any $\rT >0$, the equation $(\epsi D_t + \slashed{D} ) \Psi = 0$
admits solutions that satisfy, uniformly for $t \in [-\rT,\rT]$ and $x \in \R^2$:
\begin{equation}\label{eq:00ii}
    \Psi(t,x) = \dfrac{e^{\frac{i}{\epsi}\chi(t,x)}}{\sqrt{\epsi}}  \psi\left( t,\dfrac{x-y_t}{\sqrt{\epsi}} \right) + O_{L^2}(\epsi^{1/2}), \qquad \text{where} \quad  \psi(t,z) := \big(\RR_{\te_t} g_t\big) * \psi_{\te_t,B}(z),
    \end{equation}
    \begin{equation} 
    \text{and} \quad
        \chi(t,x) = \int_0^t \dot{y_s} \cdot A(y_s) ds + A(y_t) \cdot (x-y_t) + (x-y_t) \cdot \left( \nabla A(y_t)^\top - B n_t \tau_t^\top\right) (x-y_t).
    \end{equation}
\end{theo}

The leading order term in \eqref{eq:00ii}:
\begin{itemize}
    \item[(i)] Propagates at speed $c = (1+B^2)^{-1/2}$ along $\Gamma$, in the prescribed direction $\tau$.
    \item[(ii)] Is semiclassically localized at the phase-space point $(y_t,A_t)$ with $A_t=A(y_t)$. This is where the eigenvalues of the symbol of $\Di$ are degenerate. 
    While the rapid oscillations generated by $A_t$ can be locally gauged away, they cannot be globally neglected when $\Gamma$ is a closed loop: this is the Aharonov--Bohm effect. 
    \item[(iii)] Disperses along $\Gamma$ at rate prescribed by the difference $\te_t-\te_0$:
    \begin{equation}\label{eq:00k}
  \big| \psi(t,z) \big| \leq \dfrac{C}{1+|\te_t-\te_0|^{1/2}}\sup_{z \in \R^2} \big| \psi(0,z) \big|.
\end{equation} 
In particular, the leading part in \eqref{eq:00ii} is controlled by $\epsi^{-1/2} |\te_t-\te_0|^{-1/2}$.
This is a relatively weak dispersion as it comes from dispersion of the wave envelop  $\R_{\te_t g_t} \star \psi_{\te_t,B}$ rather than dispersive relations of plane waves. It produces effects for times of order one, in contrast with dispersion in e.g. the semiclassical Schr\"odinger equation which arises at time $\epsi$.
\end{itemize}
We detail these three effects in \S\ref{sec-1.3}. Theorem \ref{thm:main} will extend the result of Theorem \ref{thm:5} to cover varying magnetic fields, general domain walls and longer times of validity. The corresponding wavepackets will more generally have a variable speed, given by the ODE \eqref{eq:yt}; a phase with properties identical to (ii) above; and a more complicated rate of dispersion.

\subsection{Effects of the magnetic field}\label{sec-1.3} We comment here on the structure of the wavepacket \eqref{eq:00i}, when $B$ is constant and $\kappa$ satisfies \eqref{eq:00g}; and detail how things change when these conditions are relaxed (see Theorem \ref{thm:main}).   

\iffalse________________________________________________________________
\begin{figure}
    \centering
    \includegraphics[width=5cm]{phase_AB.pdf}  
    \includegraphics[width=5cm]{1_2.pdf} 
    \includegraphics[width=5cm]{1.pdf} 
    \caption{The left figure illustrates the Aharonov-Bohm effect showing the phase of one spinor component as the wavepacket performs one full revolution around the unit circle with $A(x)=\vert x \vert^{-1}$. The figure in the center shows a sequence of snapshots of the state undergoing one full revolution on a circular edge profile $\kappa(y) :=\log(\vert y \vert)$ with constant magnetic field $B=1/2$ and on the right for $B=1$ corresponding to a maximal dispersion rate $\gamma=1/2$.}
    \label{fig:Fig_overview}
\end{figure}
_______________________________________________________________\fi

\begin{figure}[htbp]
\floatbox[{\capbeside\thisfloatsetup{capbesideposition={right,center},capbesidewidth=10cm}}]{figure}[\FBwidth]
{\caption{ \label{fig:AB} Aharonov--Bohm effect for the domain wall $\kappa(y) = \ln |y|$ (corresponding to a circle interface). The plot shows the rapid evolution of the phase of the first spinor component as the wavepacket performs a single revolution around the circle. See \S\ref{sec-7.3} for details.}}
{\begin{tikzpicture}
   \node at (0,0) {\includegraphics[width=6cm]{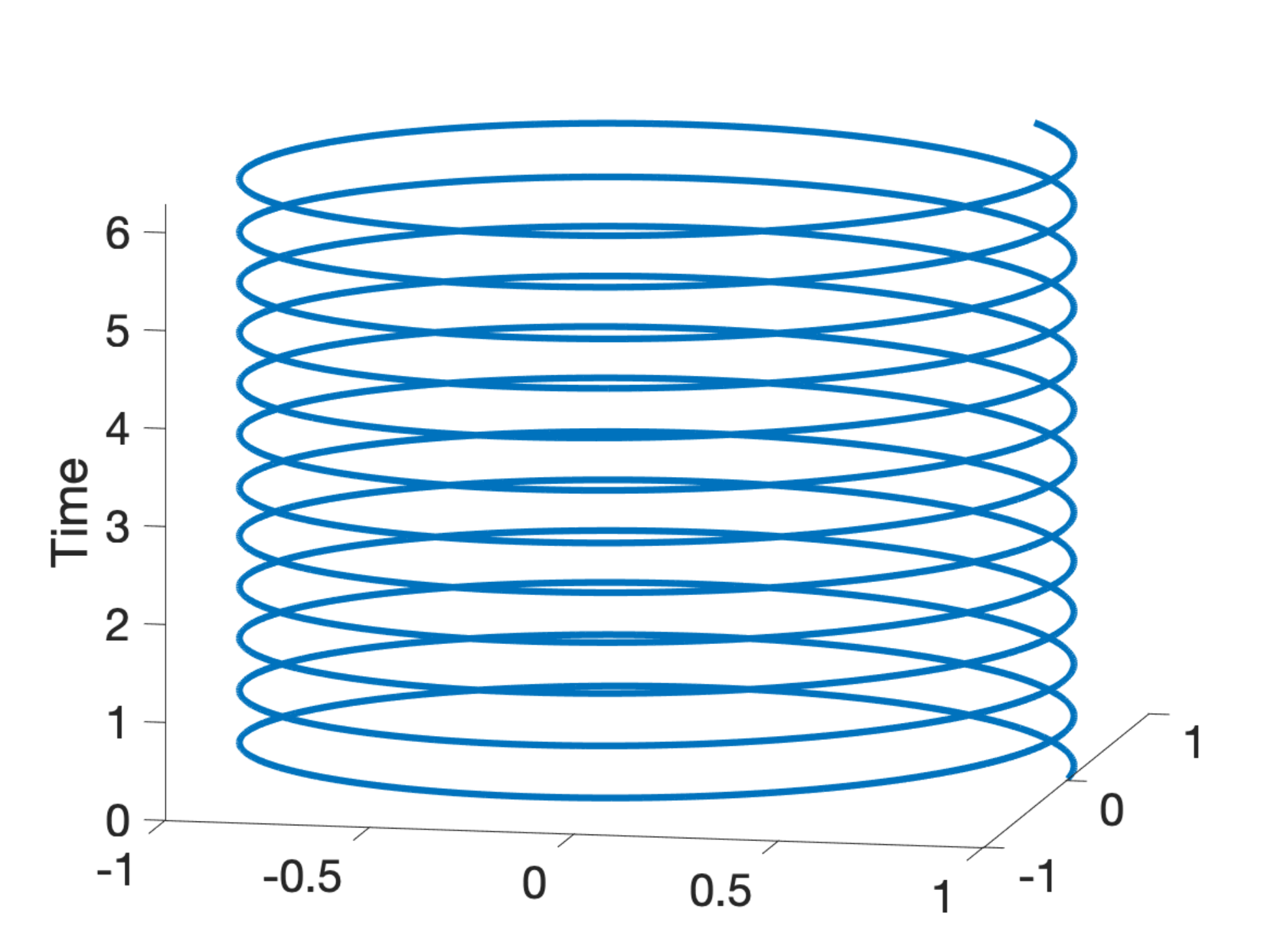}};
  \end{tikzpicture}}
\end{figure}

\begin{figure}[b]
    \centering
    \includegraphics[width=5cm]{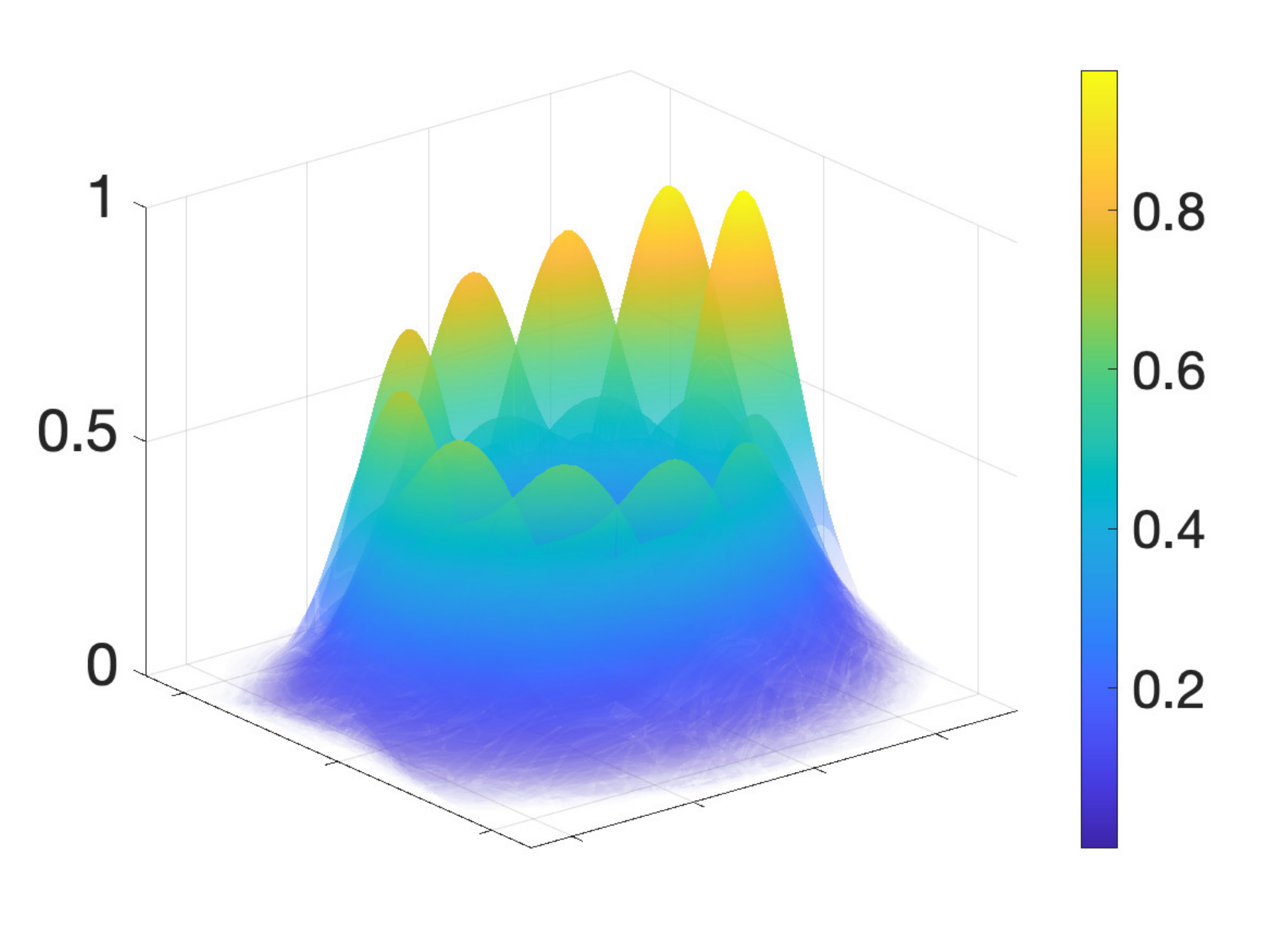} 
    \includegraphics[width=5cm]{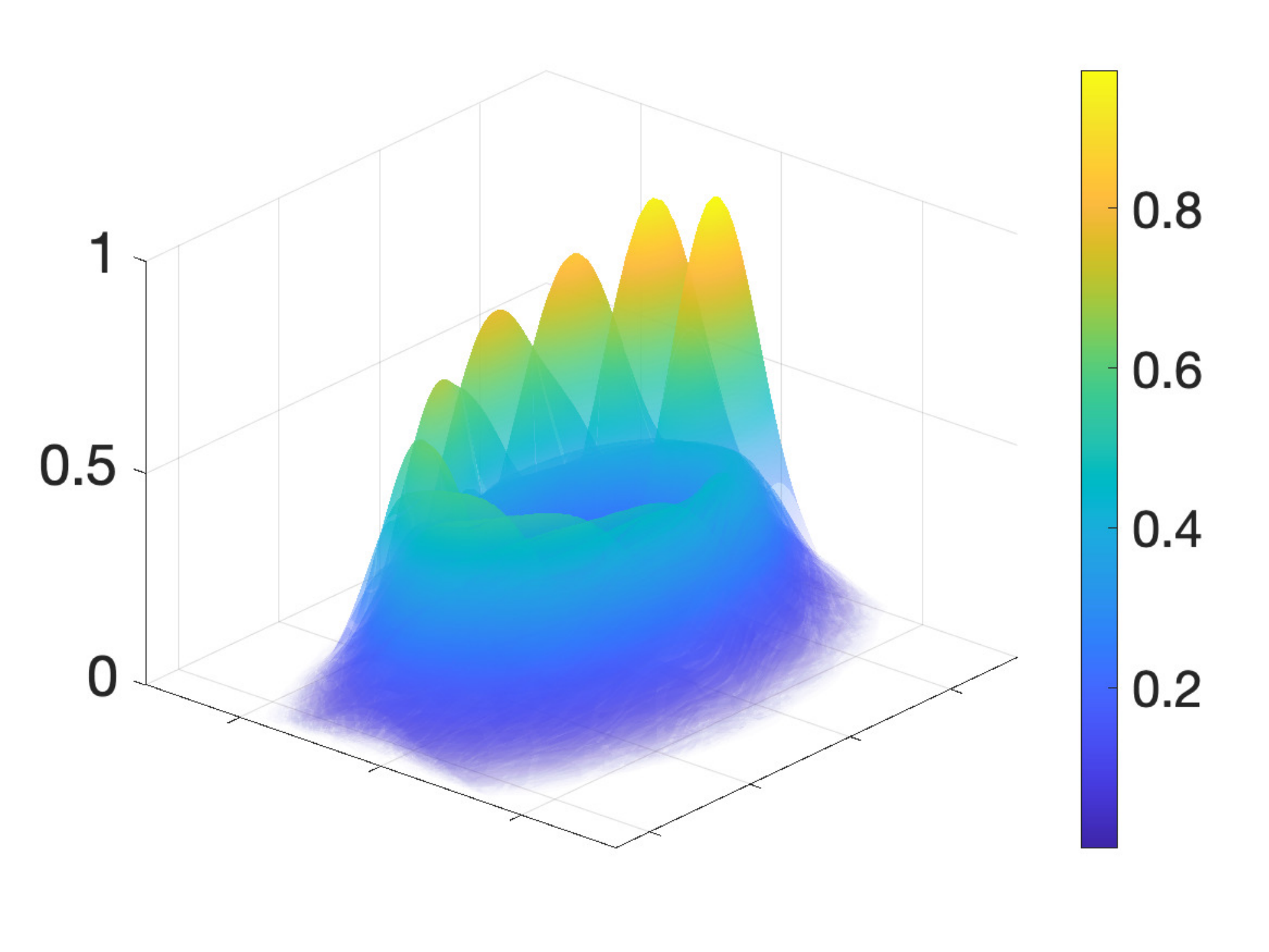} 
    \includegraphics[width=5cm]{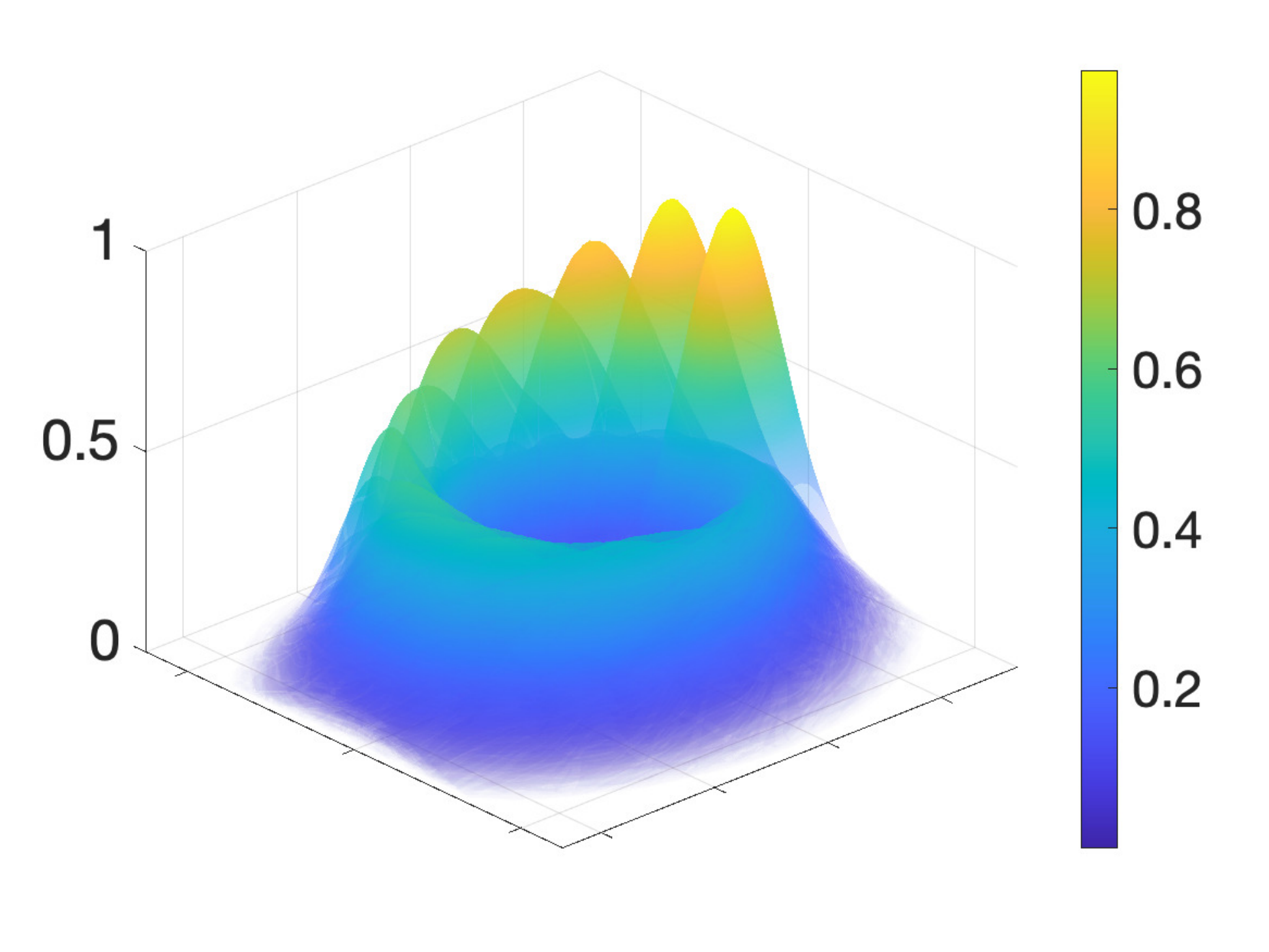}\\
    \includegraphics[width=5cm]{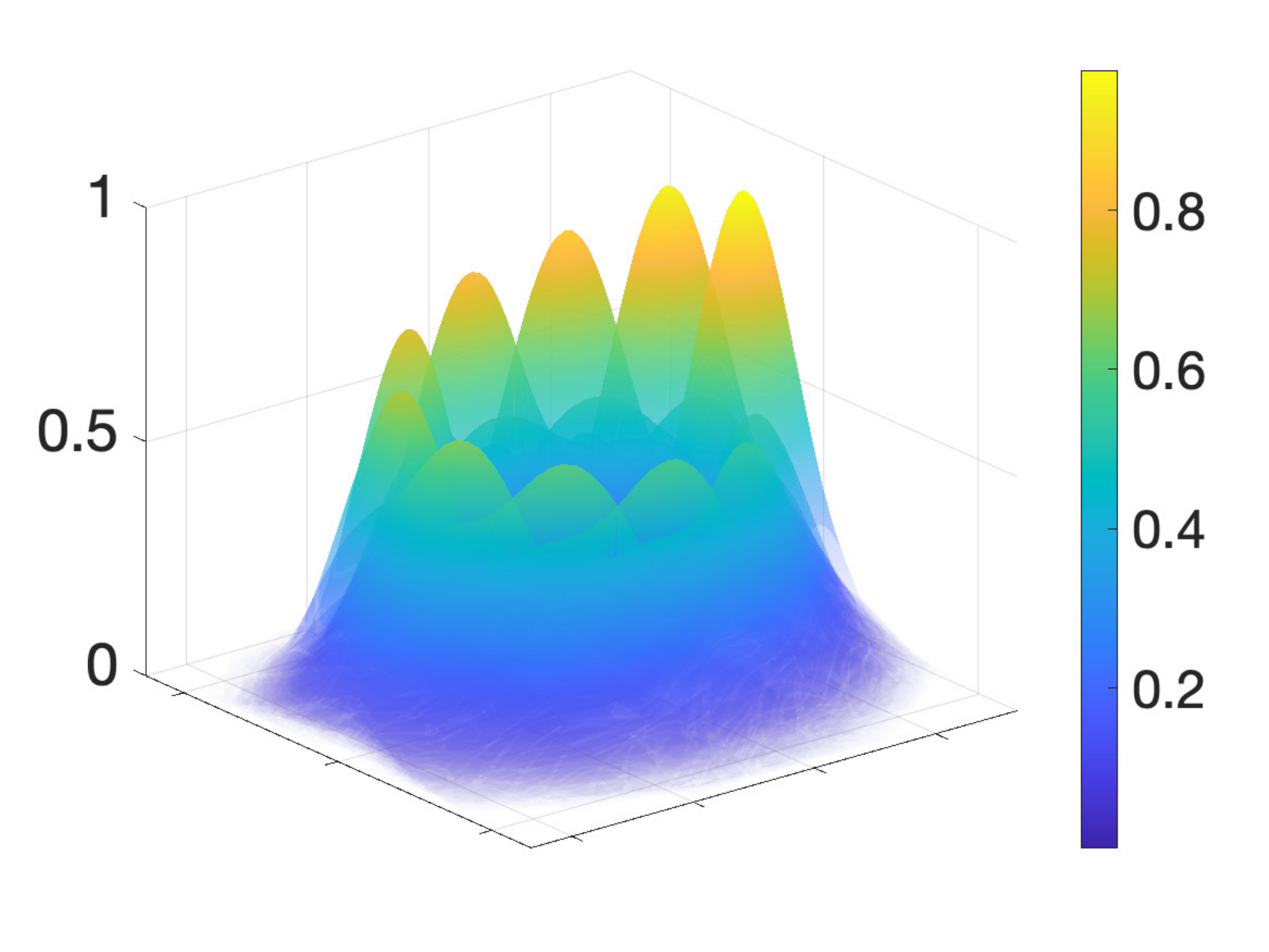} 
\includegraphics[width=5cm]{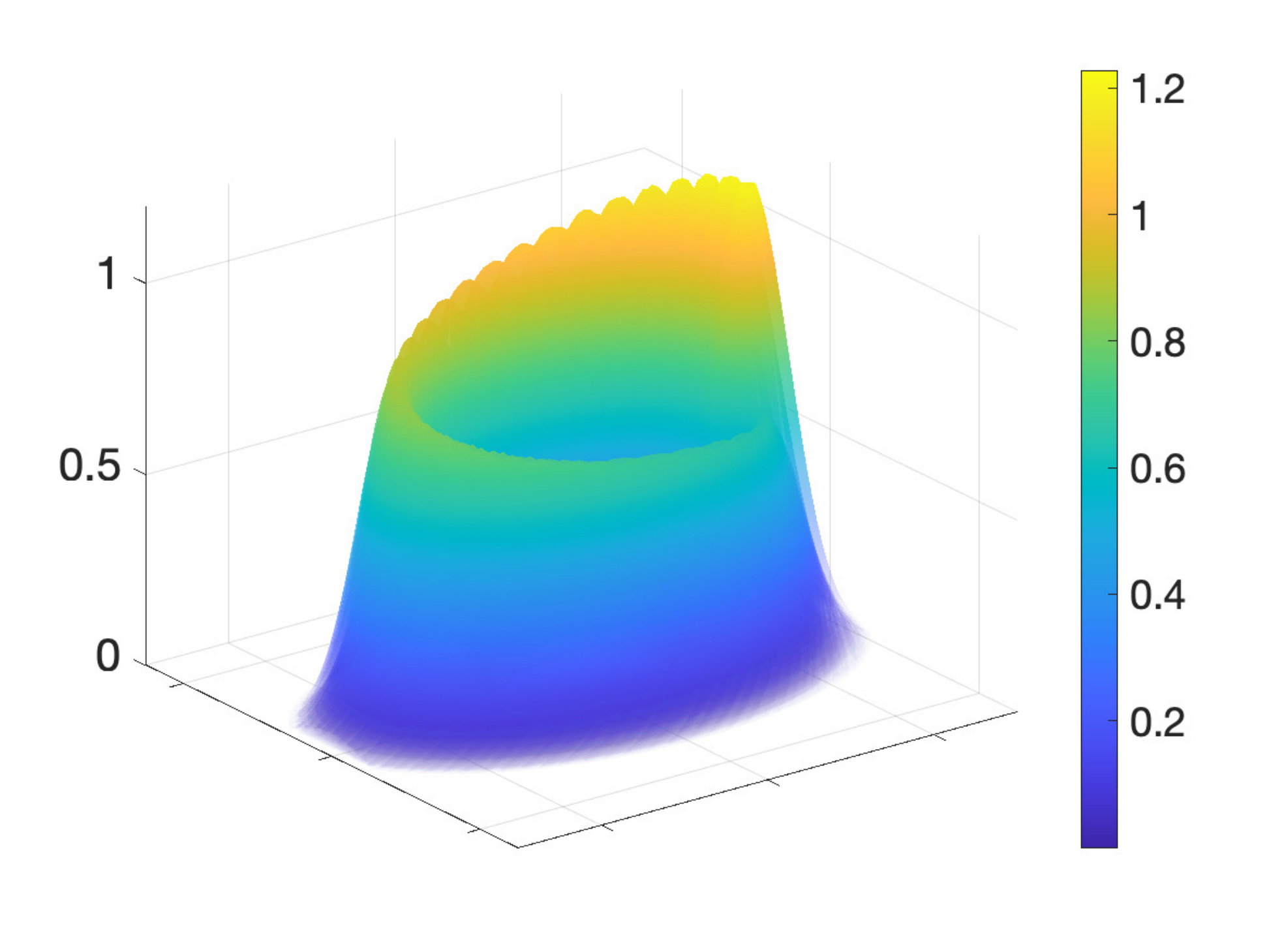} 
\includegraphics[width=5cm]{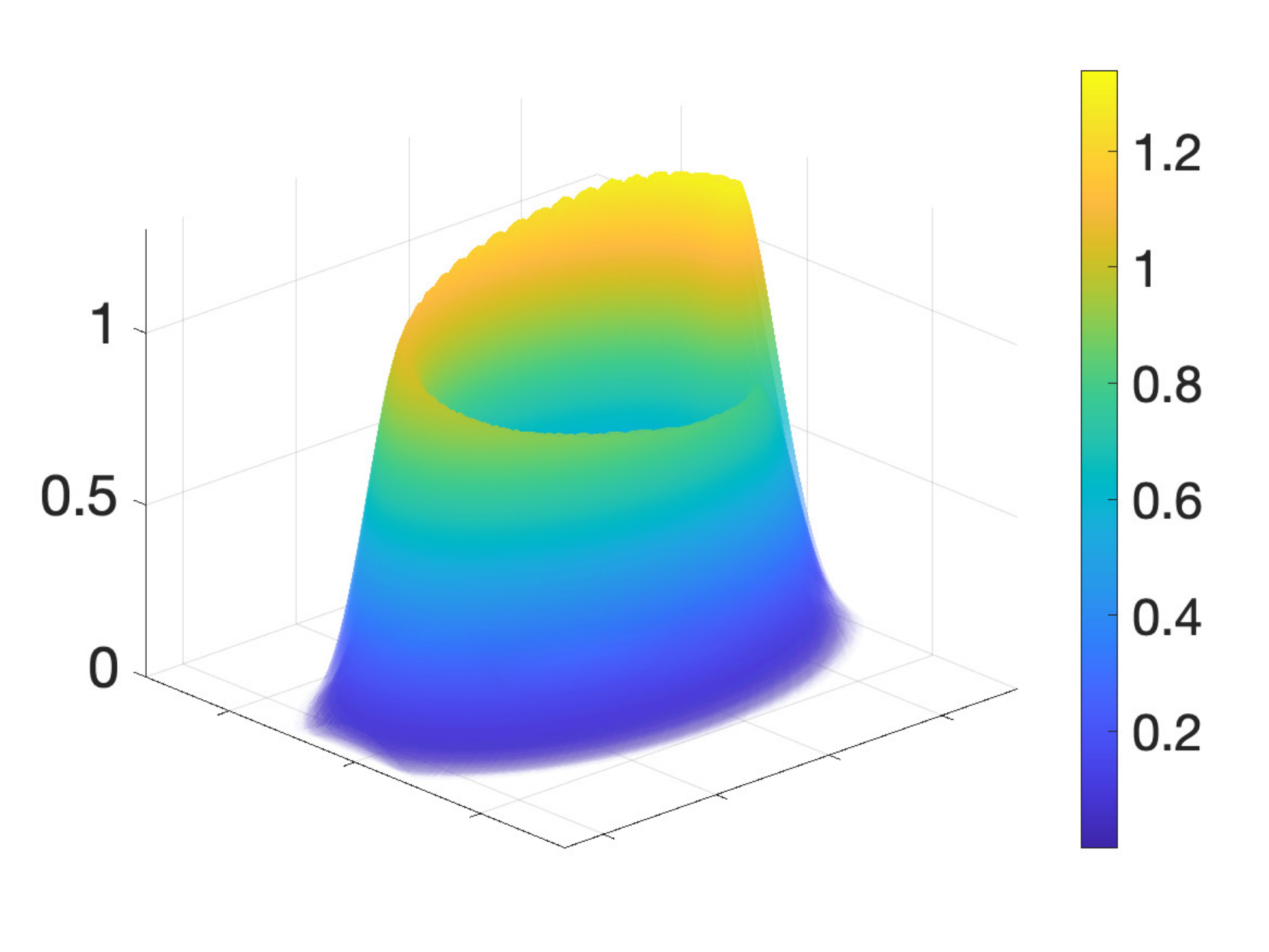}
    \caption{ Snapshots of the evolution of a wavepacket around the unit disc
    (with $\kappa(y) = \ln |y|$ and $\dot\theta=c$) in magnetic fields $B = 1/2,1/\sqrt{2},1$ (top row) and $B=2,3,4$ (bottom row). 
    We observe gradual dispersion consistent with our predicted dispersion rate $\gamma |\te_t-\te_0| = \gamma c t = B (1+B^2)^{-3/2} t$, see \eqref{eq-7m}. For equal distance of propagation the dispersion is maximal when $B=1$ while for equal time of propagation the dispersion is strongest when $B=1/\sqrt{2}$. }
    \label{fig:Fig_overview}
\end{figure}

Turning on a magnetic field systematically slows down the propagation. In the general setup of Theorem \ref{thm:main}, the wavepackets move at speed $(1+B_t^2)^{-1/2}$, which is smaller than $1$ whenever the magnetic field does not vanish.

The wavepacket is semiclassically localized at $(y_t,A_t)$. This point lies in the crossing set of the semiclassical symbol
\begin{equation}
    \slashed{D}(x,\xi) = \big(\xi_1-A_1(x)\big) \sigma_1 + \big(\xi_2-A_2(x)\big) \sigma_2  + \kappa(x) \sigma_3;
\end{equation}
that is, the eigenvalues of $\slashed{D}(y_t,A_t)$ are repeated. Hence,  $(y_t,A_t)$ is an exotic semiclassical trajectory, in the sense that it is not among those predicted by standard propagation of singularity, such as \cite{DH72}.

The semiclassical action $\dot{y_t} \cdot A_t$ generates a large phase-shift $e^{\frac{i}{\epsi}\int_0^t \dot{y_s} \cdot A_s ds}$ that can be locally -- but not globally -- gauged away. When $\Gamma$ is a loop, after a full revolution the phase shift relates to the magnetic flux $\Phi = \int_\Gamma A$: it is $e^{\frac{i}{\epsi} \Phi}$, see Figure \ref{fig:AB} for the case of the circle. 
This is the Aharonov--Bohm effect.

The envelop of the wavepacket typically disperses along $\Gamma$. Under the geometric condition \eqref{eq:00g}, the rate of change of dispersion is $|\te_t-\te_0|^{-1/2}$. Indeed, writing $\psi$ in terms of $g_t$ and $\psi_{0,B}$, and using the formula \eqref{eq:00j}, we have for $|\te_t-\te_0| \geq 1$:
\begin{align}
    \sup_{z \in \R^2} \big| \psi(t,z) \big| = \sup_{z \in \R^2} \big|g_t * \psi_{0,B}(z)\big| & \leq \dfrac{1}{\sqrt{ 2\pi \gamma |\te_t-\te_0|}} \sup_{z_2 \in \R} \int_\R |\psi(y_1,z_2)| dy_1 
    \\ 
    & \leq \dfrac{C}{\sqrt{\gamma |\te_t-\te_0|}} = \dfrac{C'}{\sqrt{ \gamma |\te_t-\te_0|}} \sup_{z \in \R^2} \big| \psi(0,z) \big|. \label{eq-7m}
\end{align}
This yields \eqref{eq:00k}. While the rate of change of dispersion $|\te_t-\te_0|^{-1/2}$ is bounded above if $\Gamma$ is asymptotically flat, it can be as small as $t^{-1/2}$ when $\Gamma$ is a loop or a spiral. In these cases, the wavepacket loses coherence over long times as displayed in Figure \ref{fig:Fig_overview}; it should be noted however that larger magnetic fields do not necessarily give rise to larger dispersion. In the general setup of Theorem \ref{thm:main}, the rate of dispersion $\nu_t$ takes a more complicated form: see Lemma \ref{lem-1i}, \eqref{eq-1a} and \eqref{eq-1d}. We comment that dispersion can have long-time effects. When the dispersion is strongest,  $\nu_t \sim t$ and our construction holds up to times $\rT\ll\eps^{-1/8}$; when it is weakest, $\nu_t = O(1)$ and we recover the time of validity  $\rT\ll\eps^{-1/2}$ of \cite{bal2021edge}.

When $B$ or $|\nabla \kappa|$ vary along $\Gamma$ (that is, outside the setup of Theorem \ref{thm:5}), an additional effect emerges: time-dependent anisotropic compression/stretching in the normal and tangent directions of $\Gamma$. We refer the reader to the formula \eqref{eq-00m} and Theorem \ref{thm:main}.
The compression factors, $\sqrt{\rho_t}$ (in the normal direction) and $c_0^2/c_t^2$ (in the tangent direction), remain bounded above and below in the limit $t \rightarrow \infty$. This contrasts with the parameter $\nu_t$ that controls the dispersion, which can grow like $t$.

\subsection{Strategy of proof} Conceptually speaking, our approach consists in constructing successive transformations of $\Di$ that bring us closer to the flat Dirac operator $\Di_{0,B}$ of \eqref{eq-2a}. The main steps are as follows:

\begin{itemize}
    \item[\textbf{1.}] In \S\ref{sec:gauge}, we conjugate $\Di$ by a gauge transform $e^{i\chi(t,x)/\epsi}$, cooked up so that the resulting magnetic potential $A-\nabla \chi$ vanishes along $\Gamma$ and is tangent to the level sets of $\kappa$. These are shared features with the flat Dirac operator $\Di_{0,B}$.  The expansion of the gauge term $e^{i\chi(t,x)/\epsi}$ near $y_t$ produces the large oscillatory phase of \eqref{eq:00ii}.
    
    \item[\textbf{2.}] In \S\ref{sec:rescaling} we look for solutions to the gauge-modified Dirac equation $(\epsi D_t + \tilde{\slashed{D}}) \tilde{\Psi} = 0$ as semiclassical wavepackets localized at $(y_t,0)$. A formal Taylor expansion produces a hierarchy of equations for the envelops.
    
    \item[\textbf{3.}] In \S\ref{sec:rotation}, we perform a series of spatial and spinorial rotations on the leading equation. The first two rotations (which already appear in \cite{bal2021edge}) flatten the interface; the last is magnetically induced and is among the new ingredients. The result is the leading equation that one would get starting from $\slashed{D}_{0,B}$.  
    
    \item[\textbf{4.}] Up to a partial shifted Fourier transform, the leading operator takes the same form as in the absence of magnetic fields \cite{bal2021edge}. In \S\ref{sec:model}, we compute explicitly its kernel and prove stability estimates -- later needed for the subleading transport equation.
    
    \item[\textbf{5.}] In \S\ref{sec:transport} we explicitly integrate the transport equation. In Fourier variables, the solutions have a quadratic phase with Hessian $\te_t-\te_0$ (in the setup of \S\ref{sec-1.2}). In physical space, this transfers to dispersion at rate $|\te_t-\te_0|^{-1/2}$.
    
    \item[\textbf{6.}] Higher-order approximations are constructed iteratively in \S\ref{sec:error}.  Combining them with the unitarity of $e^{-it\Di}$ and a Duhamel argument, we obtain our main result, Theorem \ref{thm:main}. It gives an approximate solution to $(\epsi D_t + \Di) \Psi = 0$ that propagates along $\Gamma$, slowed down by the magnetic field, and explicitly expressed through the above transformations. 
\end{itemize}
We then illustrate our findings with a series of numerical simulations in \S\ref{sec:num}.

\subsection{Related literature} For systems of semiclassical PDEs, the symbols that govern the macroscopic transport are the eigenvalues of the (matrix-valued) symbol. A symbolic diagonalization argument shows that their Hamiltonian flow governs the leading-order dynamics.
When the eigenvalues of the symbol are degenerate, the classical equations of motion break down. The situation studied here is among the simplest such cases. Our analysis show that the phase-space crossing set $\{ \big(x, A(x)\big) , x\in \Gamma\}$ support wavepackets with the dynamics \eqref{eq:yt}. As mentioned above, \eqref{eq:yt} is an exotic semiclassical trajectory not predicted by the standard results on propagation of singularities \cite{DH72}.
In other setups, wavepackets may e.g. start away from the crossing set, reach it, and undergo a Landau--Zener transition; see e.g. \cites{Hag94,HJ98,FG2,Col04}.

From a physical point of a view, our motivation stems from the ubiquity of $\Di$ in the field of topological phases of matter \cites{WI,moessner2021topological}, and in particular one-particle models of topological insulators and topological superconductors \cites{VO,Be13,PSB16,B19b,B19a}, which generically come with conical points \cite{Drouot:21}. The domain wall $\kappa(x)$ models the interface between two topologically distinct insulating phases \cites{FLW16,Drouot:19,DW20}. This in turn generates an asymmetric transport along the interface $\Gamma$ by a principle called the bulk-interface correspondence; see e.g. \cites{EG02,GP,PSB16,Drouot:19b,B20,bal2021topological,Drouot2020microlocal,bal2021}. The wavepackets analyzed here encode this asymmetry; see \cite[\S1.4]{bal2021edge} for a discussion when $A=0$.  The operator $\Di$ also emerges in the effective analysis of graphene  and its pseudomagnetic (strained) analogues, see for instance \cite{GRW21}.

Note that the magnetic field is essential in the integer quantum Hall effect \cites{TKNN,avron1994,BES94}, which was the first observed example of topologically non-trivial state of matter.
There, the insulating gaps are obtained from the degenerate Landau levels associated to a constant magnetic field. This contrasts with the situation considered here: the non-trivial topology imposed by the domain wall is stable against magnetic contributions; see \cite{bal2021topological}. This means that the magnetic field does not influence the existence of edge states. It however affect their quantitative features, see \S\ref{sec-1.3}.

Our results in Theorems \ref{thm:5} and \ref{thm:main} concern weakly dispersive wavepackets with a macroscopic center $y_t$ propagating along $\Gamma = \kappa^{-1}(0)$. Changing the metric (i.e. replacing $D$ by $\big(D_j+\Gamma_j(x))g^{jk}(x)\big)\sigma_k$) would likely preserve this structure. However, adding an electric potential $V(x)$ to $\Di$ modifies the energy landscape, hence the interface $\Gamma = \kappa^{-1}(0)$ between the two topological media cannot properly support wavepackets. These are likely to propagate instead within a thicker strip close to $\Gamma$, with splitting according to Landau--Zener transition rules.


  \subsection*{Acknowledgments} This work is a sequel to \cite{bal2021edge}, which started during a 2020 AIM workshop, Mathematics of topological insulators. The authors thank the organizers: Daniel Freed, Gian Michele Graf, Rafe Mazzeo and Michael Weinstein. In addition, the authors are very grateful to Clotilde Fermanian Kammerer, Jianfeng Lu, and Alexander Watson for initial discussions. The authors acknowledge support form the NSF grants DMS-2118608 and DMS-2054589 (AD), DMS-1908736 and EFMA-1641100 (GB).

\section{Gauge transformation} \label{sec:gauge}

This section constructs a gauge function that reduces $\Di$ to an operator with magnetic properties closer to those of the model operator \eqref{eq-2a}.

\subsection{Equivalent tangent magnetic potential} We construct first a magnetic potential $\tilde A$ on $\R^2$ such that $\nabla \times \tilde{A} = B$ near $\Gamma$, and $\tilde A$ vanishes along $\Gamma=\kappa^{-1}(0)$ and is carried by the vector field $\tau$. In particular, the operator
\begin{equation}\label{eq-3f}
    \tDi = \big( \epsi D_1 - \tilde A_1(x) \big) \sigma_1 + \big( \epsi D_2 - \tilde A_2(x) \big) \sigma_2 + \kappa(x) \sigma_3
\end{equation}
will share many of the characteristics of the model \eqref{eq-2a}: the magnetic potential is tangent to $\Gamma$ and vanishes along $\Gamma$. Moreover, from magnetic equivalence between $A$ and $\tilde A$, $\Di$ and $\tDi$ are locally conjugated. We use below the notation (with Euclidean distance $d$)
\begin{equation}
    U_\eta = \big\{ x \in \R^2 : d(x,\Gamma) < \eta \big\}.
\end{equation}

\begin{lemm}\label{lem-1e} There exist $\eta > 0$ and a function $\beta \in C^\infty(\R^2,\R)$ vanishing on $\Rm^2\backslash U_{2\eta}$ such that, defining $\tilde A:=\beta\kappa\tau$, we have
\begin{equation}\label{eq-0h}
    \nabla \times (A - \tilde A) = \nabla \times (A - \beta\kappa\tau) = 0 \quad \text{ on } U_\eta; \qquad \beta = \dfrac{B}{|\nabla \kappa|} \quad \text{ on } \Gamma.
\end{equation}
\end{lemm}

\begin{proof} We first note that it suffices to construct $\beta$ on $U_\eta$ and extend it to $\R^2$ as a smooth function with support in $U_{2\eta}$. For $f$ a smooth function, we observe that
\begin{equation}\label{eq-0w}
    \nabla \times (f \tau) = \nabla f \times \tau +f \nabla \times \tau = \p_n f + f \nabla \times \tau,
\end{equation}
where we used $\nabla f \times \tau = \p_1 f \tau_2 - \p_2 f \tau_1 = \p_n f$. Let $B = \nabla \times A$. To find $\beta$ such that $B = \nabla \times (\beta \kappa \tau)$, we first solve 
\begin{equation}\label{eq-0y}
   \p_n f + f  \nabla \times \tau = B, \ \ \ \ f|_\Gamma = 0.
\end{equation}

We note that the coefficients of $\p_n$ are in $C^\infty_b$. In particular, the flow $e^{s\p_n}(x)$ is defined for all times. We now define the map $\Phi : \Gamma \times \R \rightarrow \R^2$ by
\begin{equation}
    \Phi(x,s) = e^{s\p_n}(x).
\end{equation}
If $f$ solves \eqref{eq-0y}, then $\tf = f \circ \Phi$ solves 
\begin{equation}
    \p_s \tf + \tf \cdot (\nabla \times \tau) \circ \Phi = B \circ \Phi, \ \ \ \ \tf|_\Gamma = 0.
\end{equation}
This equation clearly admits a solution $\tf \in C^\infty(\Gamma \times \R,\R)$. Moreover, note that $\p_n$ is transverse to $\Gamma$; that $|\nabla \kappa|$ is uniformly bounded above and below; and that $\kappa \in C^\infty_b$, $\Phi$ is a diffeomorphism from $\Gamma \times (-\delta, \delta)$ to its range, which contains a neighborhood of the form $U_\eta$. Therefore, $\tf$ induces the solution $f = \tf \circ \Phi^{-1} \in C^\infty(U_\eta,\R)$ of \eqref{eq-0y}. 

Because $f$ vanishes on $\Gamma$ and $\kappa$ vanishes transversely on $\Gamma$ (with $|\nabla \kappa|$ bounded below), after potentially reducing $\eta$ we can write $f = \beta \kappa$ for some smooth function $\beta$ on $U_\eta$. From the equation \eqref{eq-0y} and the identity \eqref{eq-0w}, we conclude that
\begin{equation}
    \nabla \times A = B = \p_n f + f \nabla \times \tau = \nabla \times (f \tau) = \nabla \times (\beta \kappa \tau). 
\end{equation}
Moreover, using that $f$ and $\kappa$ vanish on $\Gamma$, we have along $\Gamma$:
\begin{equation}
    B = \p_n f = \p_n (\beta \kappa) = \beta \p_n \kappa = \beta |\nabla \kappa|.
\end{equation}
This completes the proof.\end{proof}

Thanks to \eqref{eq-0h}, $A$ and $\tilde A$ give rise to the same magnetic field on $U_\eta$, hence the difference $A-\tilde A$ is locally a gradient field:

\begin{lemm}\label{lem-1f} Let $\eta > 0$ given by Lemma \ref{lem-1e} and $U \subset U_\eta$ be a simply connected set. For any $(y_0,\eta_0) \in U \times \R$, there exists a unique $\chi_0 \in C^\infty(U,\R)$ such that 
\begin{equation}\label{eq-2p}
    \tilde{A} = A -\nabla \chi_0 \quad \text{ in } U; \quad \text{ and } \quad \chi_0(y_0) = \eta_0.
\end{equation}
\end{lemm}

\begin{proof}[Proof of Lemma \ref{lem-1f}] Thanks to \eqref{eq-0h}, we have $\nabla \times (A -\beta \kappa \tau) = 0$ on $U_\eta$, hence on $U$. Since $U$ is simply connected, by Poincar\'e's lemma there exists a unique $\chi_0 \in C^\infty(U,\R)$ such that $\nabla \chi_0 = A - \beta \kappa \tau$ and $\chi(y_0)=\eta_0$. 
\end{proof}

It follows from Lemma \ref{lem-1f} that the operators $\Di$ and $\tDi$ defined in \eqref{eq-3f} are locally conjugate: on open sets $U$ produced by Lemma \ref{lem-1f}, we have
\begin{equation}\label{eq-3d}
    e^{-i\chi_0/\epsi} \Di e^{i\chi_0/\epsi} = \tDi.
\end{equation}

\subsection{Global gauge} 
When $\Gamma$ is simply connected, we can pick $U = U_\eta$ in Lemma \ref{lem-1f}, and the conjugation relation \eqref{eq-3d} holds on a full neighborhood of $\Gamma$. This however fails when $\Gamma$ is a loop: $\chi_0$ is only defined on part of $\Gamma$. We circumvent this obstacle by using instead a time-dependent gauge that follows the center of mass $y_t$ of our wavepacket, i.e. defined on a set of the form
\begin{equation}
    \Omega_\delta = \big\{(t,x), \  t\in \R, \ |x-y_t| < \delta \big\}.
\end{equation}

\begin{prop}\label{lem-1t} \label{prop:globalgauge} Let $\tilde A$ be as defined in Lemma \ref{lem-1e}. There exist $\delta > 0$ and $\chi \in C^\infty(\R \times \R^2, \R)$ with support in $\Omega_{2\delta}$ with $\chi(0,y_0) = 0$ and such that for $(t,x) \in \Omega_\delta$:
\begin{equation}\label{eq-2h}
    \tilde{A}(x) = A(x) - \nabla \chi(t,x), \qquad \p_t \chi(t,x) = 0.
\end{equation}
\end{prop}

\begin{proof} As in the proof of Lemma \ref{lem-1e}, it suffices to construct $\chi$ on $\Omega_\delta$ and to extend it to $\R^2$ as a smooth function with support in $\Omega_{2\delta}$. Without loss of generality, we can assume that $\Gamma$ is connected. If $\Gamma$ is also simply connected, then we simply take $\delta = \eta$ and $\chi(t,x) = \chi_0(x)$. If $\Gamma$ is not simply connected, then it is a loop; in particular it is compact. 

Fix $t \in \R$. According to Lemma \ref{lem-1f}, there exists $\delta_t > 0$ and a smooth function $\chi(t,\cdot)$ defined on the ball $B(y_t,\delta_t)$ such that
\begin{equation}\label{eq-2k}
    \tilde{A}(x) = A(x) - \nabla \chi(t,x), \quad x \in B(y_t,\delta_t); \qquad \chi(t,y_t) = \int_0^t \dot{y_s} \cdot A(y_s) ds.
\end{equation}
Since $\Gamma$ is compact, we can pick $\delta_t$ independent of $t$; we write below $\delta = \delta_t$.  Varying $t$, we obtain a uniquely defined function $\chi$ on $\Omega_\delta$. We take time-derivative of both identities in \eqref{eq-2k}. The first one yields $\nabla \p_t \chi(t,x) = 0$ for $x \in B(y_t,\delta)$. The second one gives
\begin{equation}\label{eq-2n}
\p_t \chi(t,y_t) = \p_t \big(\chi(t,y_t) \big) - \dot{y_t} \cdot \nabla \chi(t,y_t)  
= \dot{y_t} \cdot A(y_t) - \dot{y_t} \cdot A(y_t) = 0,
\end{equation}
where we used $\nabla \chi = \beta \kappa \tau - A$ hence $\nabla \chi(t,y_t) = A(y_t)$. From these identities, we deduce that $\p_t \chi(t,x) = 0$ for $x \in B(y_t,\delta)$; in particular $\p_t \chi$ is smooth. Moreover, $\nabla \chi = A-\tilde A$ is also smooth. This implies that $\chi$ is smooth, which completes the proof. \end{proof}

Thanks to \eqref{eq-2h}, we have the relation
\begin{equation}\label{eq-2o}
    e^{-i\chi/\epsi} \big(\epsi D_t + \Di\big) e^{i\chi/\epsi} = \epsi D_t + \tDi + R, \qquad \text{where:}
\end{equation}
\begin{equation} \label{eq:R}
R= \p_t\chi + (\p_1\chi-A_1+\tilde A_1)\sigma_1 + (\p_2\chi-A_2+\tilde A_2)\sigma_2. 
\end{equation}

It follows from \eqref{eq-2o} that $\Psi$ solves \eqref{eq:D1} if and only if $\tilde \Psi = e^{-i\chi/\epsi}\Psi$ solves 
\begin{equation}\label{eq:D2}
   \big( \epsi D_t + \tilde {\slashed D} + R\big)  \tilde \Psi = 0.
\end{equation} 
While $R$ looks like a leading-order term, it will effectively be of order $\eps^\infty$ because it vanishes on $\Omega_\delta$ -- a domain where our wavepacket is concentrated. We will eventually treat it as a small source term in \S\ref{sec:error}.

\subsection{Gauge expansion} We conclude this section with an expansion of $\chi$ near $(t,y_t)$, which will serve to prove Theorem \ref{thm:5}. 

\begin{lemm} We have: 
\begin{equation}\label{eq-2g}
    \chi(t,y_t) = \int_0^t \dot{y_s} \cdot  A_s ds; \qquad \nabla \chi(t,y_t) = A_t; \qquad \nabla^2 \chi (t,y_t) = \nabla A^\top(y_t) - B_t n_t \tau_t^\top.
\end{equation}
\end{lemm}

\begin{proof} The first formula of \eqref{eq-2g} comes from the equation \eqref{eq-2k} defining $\chi$. The second one follows from
\begin{equation}
    \nabla \chi(t,y_t) = A(y_t)-\beta(y_t) \kappa(y_t) \tau(y_t) = A_t.
\end{equation}
In order to see the last one, we notice that 
\begin{align}
    \nabla^2 \chi (t,y_t) & = \nabla \cdot \big(\nabla^\top \chi  \big)(y_t) = \nabla \big( A^\top - \beta \kappa \tau^\top \big)(y_t) 
    \\
    & = \nabla A^\top(y_t) - \beta(y_t) \nabla \kappa(y_t) \tau^\top (y_t) = \nabla A^\top(y_t) - B_t n_t \tau_t^\top,
\end{align}
where we used $\beta_t =B_t/r_t$ on $\Gamma$, see \eqref{eq-0h}. This completes the proof. 
\end{proof}

\section{Local normal form}\label{sec:local}

We now start our construction of approximate solutions to $(\epsi D_t + \tilde {\slashed D})\tilde\Psi=0$. This section describes a series of transformations that allows us to explicitly invert an operator that governs the leading dynamics: a rescaling of $\tilde \Psi$ in natural coordinates, spatial and spinorial rotations, and finally a shifted partial Fourier transform.

\subsection{Spatial rescaling}\label{sec:rescaling} We first write the wavepacket in natural coordinates:
\begin{equation}\label{eq:psiz}
\psi(t,z)
= \eps^{\frac12} S \tilde \Psi (t,z),
\end{equation}
with the scaling transformation defined as
\begin{equation}\label{eq:S}
     Sf(z)=f(y_t+\sqrt\eps z),\qquad S^{-1}f(x)=f\left(\frac{x-y_t}{\sqrt\eps}\right).
\end{equation}
We observe that $\eps^{\frac12}S$ is an isometry on $L^2(\Rm^2)$, while formally,
\begin{equation}\label{eq:Sj}
 Sf(z) = \sum_{\alpha} \frac{\eps^{\frac{|\alpha|}2}z^\alpha}{\alpha!} \partial^\alpha f(y_t) = \dsum_{j\geq0} \eps^{\frac j2} S_jf(z),\quad S_jf(z):= \dsum_{|\alpha|=j} \frac{z^\alpha}{\alpha!} \partial^\alpha f(y_t).
\end{equation} 
This scaling emerges from the following consideration. For $\kappa(x) = x_2$, the wavepacket is confined by a harmonic-like oscillator $\epsi D_2 \sigma_2 + x_2 \sigma_3$ and hence comes with a natural scale $\sqrt\eps x_2$. This scale is then imposed to all directions when $\kappa$ admits more sophisticated variations, i.e. when $\Gamma$ is curved. 

We then verify that $\psi(t,z)$ solves
\begin{equation}\label{eq:L}
    L\psi=0 ;\qquad L= \eps^{-\frac12} S(\eps D_t+\tilde{\slashed D})S^{-1}  = \eps^{\frac12} D_t -\dot y_t \cdot D + (D+\eps^{-\frac12} Sh(z))\cdot \sigma
\end{equation}
with $D=(D_1,D_2,0)^t$, $\sigma=(\sigma_1,\sigma_2,\sigma_3)^t$, and $h=(-\tilde A_1,-\tilde A_2,\kappa)^t=(-\beta\kappa\tau,\kappa)^t$. We still use $\dot y_t$ for $(\dot y_t,0)^t\in\Rm^3$. 
Using \eqref{eq:Sj}, we find formally $L=\sum_{j\geq0} \eps^{\frac j2}L_j$, where we observe that $L_{-1}=0$ since $\tilde A$ and hence $h$ vanish on $\Gamma = \kappa^{-1}(0)$, and
\begin{eqnarray} 
  L_0 &=& -\dot y_t\cdot D + D\cdot\sigma + S_1h \cdot\sigma \label{eq:L0} \\
  L_1 &=& D_t + S_2h\cdot\sigma \label{eq:L1}\\
  L_j&=& S_{j+1}h\cdot\sigma  = \dsum_{|\alpha|=j+1} \frac{1}{\alpha!} z^\alpha \partial^\alpha h(y_t)\cdot\sigma,\qquad j\geq2 \label{eq:Lj}.
\end{eqnarray}
Our next objective is to transform $L_0$ in an appropriate basis so that its infinite dimensional kernel and its inverse on the orthogonal complement may be written explicitly. 
  
\subsection{Rotations}\label{sec:rotation} We now introduce spatial and spinorial rotations already mentioned in the introduction. We define, for $1\leq j\leq 3$ and $\theta\in\Rm$, the spinor rotation acting on $\Cm^2$:
\[  
  U_{j,\theta} = e^{-i\frac\theta2 \sigma_j} = \cos \left(\tfrac\theta2\right) -i \sin\left(\tfrac \theta2\right) \sigma_j.
\]
Associated to it is a standard three-dimensional rotation acting on $\Rm^3$ of angle $\theta$ about the $j$-axis (with usual orientation) $\tilde R_{j,\theta}$ such that for $c \in \mathbb R^3$:
\begin{equation}\label{eq:rotations}
  \tilde R_{j,\theta} c \cdot \sigma = c\cdot U_{j,\theta}^*\sigma U_{j,\theta},\qquad U_{j,\theta}^*\sigma U_{j,\theta} =: \tilde R_{j,\theta}^*\sigma.
\end{equation}

Note that 
\begin{equation}\label{eq:spinrot}
     U_{k,\theta}^*\sigma_j U_{k,\theta} = \delta_{j\not=k}\cos\theta \sigma_j+ \epsilon_{jkl}\sin\theta \sigma_l+ \delta_{jk}\sigma_j.
\end{equation}
From this, we deduce two formulas that we will use later:
\begin{equation}\label{eq:rotations23}
 \tilde R_{2,\varphi_t} = \matrice{ \cos\varphi_t &0& -\sin\varphi_t \\ 0&1&0 \\ \sin\varphi_t &0&\cos\varphi_t} = \matrice{ c_t & 0 & -s_t \\ 0 & 1 & 0 \\ s_t & 0 & c_t}, \qquad
 \tilde R_{3,\theta_t} = 
 \matrice{R_{\te_t} & 0 \\ 0 & 1},
\end{equation}
where $R_{\te_t}$ is the spatial rotation defined in \eqref{eq:coefst}. 
We finally define the $L^2$-unitary transform related to the spatial rotations
\begin{equation}\label{eq:mR}
   \mR_{\theta_t} f=f\circ R_{\theta_t}
\end{equation}
and the operator
\begin{equation}\label{eq:bUt}
    \bU_t = U_{2,\varphi_t} U_{3,\theta_t} \mR_{\te_t}.
\end{equation}
When $\varphi_t\equiv0$ (i.e., when $B=0$), this transformation already appears in \cite{bal2021edge}.
We now compute the operators $\bU_t^* L_j \bU_t$. We start with preliminary relations:

\begin{lemm}\label{lem-1c} We find:
\begin{equation}\label{eq-0m}
    \bU_t^* ( -\dot  y_t \cdot D) \bU_t = c_t D_1, 
    \quad 
    \bU_t^* D\cdot\sigma \bU_t  = (c_t\sigma_1+s_t\sigma_3) D_1 + \sigma_2 D_2,
    \quad
\end{equation}
\begin{equation}\label{eq-0q}
    \bU_t^* D_t \bU_t 
    = D_t-\frac{\dot\varphi_t}{2} \sigma_2 - \frac{\dot\te_t}{2} (-s_t \sigma_1 + c_t \sigma_3) - \dot{\theta_t} (z_1 D_2 - z_2 D_1).
\end{equation}
Moreover $\bU_t^* S_j h \cdot \sigma \bU_t$ is a multiplication operator by a time-dependent homogeneous polynomial of degree $j$ in $z$; and for $j=1,2$: 
\begin{equation}\label{eq-0o}
    \bU_t^* S_1 h \cdot \sigma \bU_t = \rho_t z_2 \sigma_3, \quad \bU_t^* S_2 h \cdot \sigma \bU_t = -r_t c_t z_2 \lr{z, \nabla \beta(y_t)} + q_t(z) \cdot \sigma
\end{equation}
where $q_t\cdot\sigma$ is a matrix-valued quadratic form in $z$ carried by $\sigma_2, \sigma_3$.
\end{lemm}

\begin{proof} 1. Space differentiation. We first observe that 
 \begin{equation}\label{eq-0n}
   \mR_{\theta_t}^*D \mR_{\theta_t} = \tilde R_{3,{\theta_t}}^*D,\quad \mR_{\theta_t}^*D\cdot v \mR_{\theta_t} = D\cdot \tilde R_{3,{\theta_t}} v. 
 \end{equation}
Applied to $v=-\dot y_t$ with $R_{\theta_t} \dot y_t=-c_te_1$ by construction, we find $-\tilde R_{3,{\theta_t}} \dot y_t \cdot D=c_t D_1$. Since spinorial rotations commute with scalars, we obtain the first relation in \eqref{eq-0m}. 

We now apply \eqref{eq-0n} to $v=\sigma$ and find
\[
   (\mR_{\theta_t} U_{3,{\theta_t}})^* D\cdot\sigma \mR_{\theta_t} U_{3,{\theta_t}} = D\cdot\sigma.
\]
Therefore, using \eqref{eq:rotations23}, we find 
\[
  \bU_t^* D\cdot\sigma \bU_t = U_{2,\varphi_t}^* D\cdot\sigma U_{2,\varphi_t} = D \cdot \tilde{R}_{2,\varphi_t}^* \sigma = (c_t\sigma_1+s_t\sigma_3) D_1 + \sigma_2 D_2.
\]
This proves the second identity in \eqref{eq-0m}. 

2. Time differentiation. We note that 
\begin{equation}
    \mR_{\theta_t}^*D_t\mR_{\theta_t}=D_t-\dot{\theta_t} Jz \cdot D_z,
 \quad
 U_{3,\te_t}^*D_t U_{3,\te_t} = D_t-\frac{\dot\te_t}{2} \sigma_3,
\quad 
U_{2,\varphi_t}^*D_t U_{2,\varphi_t} = D_t-\frac{\dot\varphi_t}{2} \sigma_2.
\end{equation}
We deduce that
\begin{align}
    \bU_t^* D_t \bU_t & = U_{2,\varphi_t}^* \Big( D_t-\frac{\dot\te_t}{2} \sigma_3 \Big) U_{2,\varphi_t} - \dot{\theta_t}  (z_1 D_2 - z_2 D_1)
    \\
    & = D_t-\frac{\dot\varphi_t}{2} \sigma_2 - \frac{\dot\te_t}{2} (-s_t \sigma_1 + c_t \sigma_3) - \dot{\theta_t}  (z_1 D_2 - z_2 D_1).
\end{align}
This proves \eqref{eq-0q}.

3. Multiplicative operators.  We observe that $\mR^*_{\theta_t} Sh \mR_{\theta_t}= h(y_t+\sqrt\eps R_{\theta_t}z)$ as a multiplication operator so that 
\[
  \bU_t^* Sh\cdot\sigma \bU_t= \tilde R_{2,\varphi_t} \tilde R_{3,\theta_t} h (y_t+\sqrt\eps R_{\theta_t}z) \cdot\sigma =\sum_{j\geq1} \eps^{\frac j2} \dsum_{|\alpha|=j} z^\alpha \nu_{\alpha} \cdot \sigma,
\]
where we have defined the $3-$vectors $\nu_\alpha = \frac{1}{\alpha!} \tilde R_{2,\varphi_t}\tilde R_{3,\theta_t} (R_{-\theta_t}\nabla)^\alpha h(y_t)$. Indeed, 
\[
  h(y_t+\sqrt \eps R_\theta z) 
  = \dsum_{j\geq1} \frac{1}{j!} (\sqrt\eps z\cdot R_{-\theta} \nabla)^j h(y_t) = \dsum_{j\geq1} \eps^{\frac j2} \dsum_{|\alpha|=j} \frac{z^\alpha}{\alpha!} (R_{-\theta}\nabla)^\alpha h(y_t)
\]
using multinomial coefficients.  Therefore, $S_j h \cdot \sigma$ is a multiplicative operator by a homogeneous polynomial of degree $j$ depending on $t$;  composing with rotations preserves this feature: $\bU_t^* S_j h \cdot \sigma \bU_t$ is also a homogeneous polynomial of degree $j$.

We now focus on $j=1$. Define $\tilde h = (-\beta \tau,1)^t$ so that $h = \kappa \tilde h$.
Since $\kappa(y_t) =0$, we have $S_1 h = \tilde{h}(y_t) \cdot S_1 \kappa$. Since spinorial rotations commute with the scalar $S_1 \kappa$,
\begin{equation}\label{eq-0p}
    \bU_t^* S_1 h \cdot \sigma \bU_t = \RR_{\te_t}^* S_1 \kappa \RR_{\te_t} \  \bU_t^*\tilde{h}(y_t) \cdot  \sigma   \bU_t.
\end{equation}
We now compute separately $\mR_{\theta_t}^* S_1 \kappa \mR_{\te_t}$ and $\bU_t^*\tilde{h}(y_t) \cdot  \sigma   \bU_t$ (note that $\RR_{\te_t}$ does not technically affect that last term since $\tilde{h}(y_t)$ does not depend on $z$). We have:
\begin{equation}\label{eq-1g}
    \mR_{\theta_t}^* S_1 \kappa \cdot \sigma \mR_{\te_t} = \mR_{\theta_t}^* z \cdot \nabla \kappa(y_t) \mR_{\te_t} = z \cdot R_{\te_t}\nabla \kappa(y_t) = r_t z_2.
\end{equation}
Moreover, since $B_t = \beta_t r_t$, we have $\tilde h (y_t) = (-r_t^{-1} B_t \tau_t, 1)^t$ hence $R_{3,\te_t} \tilde{h} = (r_t^{-1}B_t,0,1)$. We deduce
\begin{equation}
U_{3,\te_t}^* \tilde{h}(y_t) \cdot  \sigma   U_{3,\te_t} =  R_{3,\te_t} \tilde h(y_t) \cdot \sigma  = \dfrac{B_t \sigma_1}{r_t} + \sigma_3 = \dfrac{\rho_t}{r_t} (s_t \sigma_1 + c_t \sigma_3).
\end{equation}
By \eqref{eq:rotations23}, $U_{2,\varphi_t}^* (s_t \sigma_1 + c_t \sigma_3) U_{2,\varphi_t} = \sigma_3$, which is the main motivation for the definition of $U_{2,\varphi_t}$. Hence,
\begin{equation}\label{eq-1h}
   \bU_t^*\tilde{h}(y_t) \cdot  \sigma  \bU_t = U_{2,\varphi_t}^*  U_{3,\te_t}^* \tilde{h}(y_t) \cdot  \sigma   U_{3,\te_t}U_{2,\varphi_t} = \dfrac{\rho_t}{r_t} \sigma_3.
\end{equation}
Multiplying \eqref{eq-1g} and \eqref{eq-1h} and going back to \eqref{eq-0p}, we conclude that:
\begin{equation}\label{eq-0c}
    \bU_t^* S_1 h \cdot \sigma \bU_t = \rho_t z_2 \sigma_3.
\end{equation}

We go on with the case $j=2$. We compute $\bU_t^* S_2 h \cdot \sigma \bU_t$. We recall that $h = \kappa \tilde h$, therefore (using again $\kappa(y_t) = 0$):
\begin{equation}\label{eq-0u}
    S_2 h = 2 (S_1 \kappa) (S_1 \tilde{h})  + (S_2 \kappa) \tilde{h}(y_t). 
\end{equation}
We start with the contribution of $(S_2 \kappa) \tilde{h}(y_t)$, i.e. the term $\bU_t^* (S_2 \kappa) \tilde{h}(y_t)  \cdot \sigma \bU_t$. Using that $S_2 \kappa$ is a scalar and $\tilde h \cdot \sigma$ does not depend on $z$ we have
\begin{equation}\label{eq-0v}
    \bU_t^* (S_2 \kappa) \tilde{h}(y_t) \bU_t = \RR_{\te_t}^* (S_2 \kappa) \RR_{\te_t} \ \bU_t^* \tilde{h}(y_t) \cdot  \sigma \bU_t = \RR_{\te_t}^* (S_2 \kappa) \RR_{\te_t} \dfrac{\rho_t}{r_t} \sigma_3,
\end{equation}
where we used \eqref{eq-1h} in the last equality. This term is carried by $\sigma_3$.

We now focus on $\bU_t^* (S_1 \kappa) (S_1 \tilde{h}(y_t) \cdot \sigma) \bU_t$.
Using \eqref{eq-1g} and that spinorial rotations commute with scalars, we obtain:
\begin{equation}\label{eq-1i}
    \bU_t^* (S_1 \kappa) (S_1 \tilde{h}) \cdot \sigma \bU_t = \mR_{\te_t}^* (S_1 \kappa) \mR_{\te_t}\, \bU_t (S_1 \tilde{h}) \cdot \sigma \bU_t = r_t z_2 \, \bU_t (S_1 \tilde{h}) \cdot \sigma \bU_t.
\end{equation}
It remains to compute $\bU_t (S_1 \tilde{h}) \cdot \sigma \bU_t$. We recall that $\tilde h = (-\beta \tau,1)^t$. Moreover, $\tau$ has unit norm: $\lr{\tau, \tau} = 1$. Taking derivatives of this expression shows that $\p_j \tau$ is normal to $\tau$; we write $\p_j \tau = \alpha_j n$ below. We deduce that
\begin{align}
    (S_1 \tilde{h}) \cdot \sigma & = -\sum_{j,k=1}^2 z_j \p_j (\beta \tau_k) \sigma_k = -\sum_{j,k=1}^2 z_j (\p_j \beta) \tau_k \sigma_k - \beta \sum_{j,k=1}^2 z_j \p_j \tau_k \sigma_k
    \\ & = -(S_1\beta) \tau \cdot \sigma - \beta \sum_{j,k=1}^2 z_j \alpha_j n_k \sigma_k = -(S_1\beta) \tau \cdot \sigma - \lr{z,\alpha} n \cdot \sigma.
\end{align}
We observe that:
\begin{equation}
    U_{3,\theta_t}^*  \tau \cdot \sigma  U_{3,\theta_t} =  R_{3,\theta_t} \tau \cdot \sigma = -\sigma_1, \qquad U_{2,\varphi_t}^*  \sigma_1  U_{2,\varphi_t} = R_{2,\varphi_t} e_1 \cdot \sigma = c_t \sigma_1 - s_t \sigma_3;
 \end{equation}
\begin{equation}
     U_{3,\theta_t}^*  n \cdot \sigma  U_{3,\theta_t} =  R_{3,\theta_t} n \cdot \sigma = \sigma_2, \qquad U_{2,\varphi_t}^*  \sigma_2  U_{2,\varphi_t} = \sigma_2. 
\end{equation}
Therefore, the following equality is valid modulo terms carried by $\sigma_2, \sigma_3$:
\begin{equation}
    \bU_t (S_1 \tilde{h}) \cdot \sigma \bU_t = \RR_{\te_t}^* S_1\beta \RR_{\te_t} c_t \sigma_1 = c_t \lr{z, R_{\te_t} \nabla \beta(y_t)} \sigma_1.
\end{equation}
Going back to  \eqref{eq-0u}, \eqref{eq-0v} and \eqref{eq-1i}, we conclude that  -- again with equality valid  modulo terms carried by $\sigma_2, \sigma_3$:
\begin{align}
    \bU_t^* S_2 h \cdot \sigma \bU_t = \bU_t^* (S_1 \kappa) (S_1 \tilde{h}) \cdot \sigma \bU_t = -c_t r_t z_2 \lr{z, R_{\te_t} \nabla \beta(y_t)} \sigma_1.
\end{align}
This completes the proof of \eqref{eq-0o}, hence of the lemma. 
\end{proof}

Thanks to Lemma \ref{lem-1c}, we can describe  $\bU_t^* L_j \bU_t$. For $j \geq 2$, $\bU_t^* L_j \bU_t$ is a multiplication operator by a homogeneous polynomial of degree $j$. For $j = 0, 1$, we have explicit expressions:
\begin{align}
  \label{eq-1s}  \bU_t^* L_0 \bU_t & = c_t (1+\sigma_1) D_1 + \sigma_2 D_2 + (\rho_t z_2 + s_t D_1) \sigma_3;
    \\
  \label{eq-1r}   \bU_t^* L_1 \bU_t & = D_t - \dot{\te_t} (z_1 D_2 - z_2 D_1) + \left( \dfrac{\dot{\te_t}s_t}{2} + j_t z_1z_2 + k_t z_2^2\right) \sigma_1 + E_1,
\end{align}
where $E_1$ is a multiplication operator by a polynomial of degree two and carried by $\sigma_2$ and $\sigma_3$ (we will see later that it does not contribute to leading order). The coefficients $j_t$ and $k_t$ are explicitly given in terms of $\beta$:
\begin{equation}\label{eq-1c}
   \matrice{j_t \\ k_t} 
    = r_t c_t R_{\te_t} \nabla \beta(y_t) = r_t c_t \matrice{-\p_\tau \beta(y_t) \\ \p_n \beta(y_t)}.
\end{equation}

We end this section with a few  important relations:
\begin{equation}\label{eq-1d}
  j_t = -r_t \dot{\beta_t}; \quad
   j_t \gamma_t =  \dfrac{d \ln c_t}{dt}; \quad k_t = \dfrac{c_t}{2} \left( \p_n B(y_t) - B_t \dfrac{\Delta \kappa(y_t)}{r_t}\right).
\end{equation}

\begin{proof}[Proof of \eqref{eq-1d}] We start with the first identity in \eqref{eq-1d}. We recall that $\dot{y_t} = c_t\tau_t$. Since $\p_\tau$ is tangent to $\Gamma$ and $\beta = B/|\nabla \kappa|$ along $\Gamma$:
\begin{equation}\label{eq-1m}
    j_t = -r_t c_t \p_\tau \beta(y_t) = -r_t c_t \tau_t \cdot \nabla \left( \dfrac{B}{|\nabla \kappa|} \right) (y_t) = -r_t \dot{y_t} \cdot \nabla \left( \dfrac{B}{|\nabla \kappa|} \right) (y_t) = -r_t \dot{\beta_t}.
\end{equation}

Regarding the second identity: using \eqref{eq-1m} and $\beta_t=B_t/r_t$ on $\Gamma$, we obtain:
\begin{equation}
   j_t \gamma_t = - \dfrac{r_t B_t}{r_t^2+B_t^2} \dot{\beta_t} = - \dfrac{\beta_t \dot{\beta_t}}{1+\beta_t^2}  = -\dfrac{1}{2} \dfrac{d}{dt}  \ln\big(1+\beta_t^2\big) = -\dfrac{1}{2} \dfrac{d}{dt}  \ln\frac{1}{c_t^2}= \dfrac{d \ln c_t}{dt}.
\end{equation}

For the last identity, we first note that $|\nabla \kappa| = \p_n \kappa$. Since $f = \beta \kappa$ solves \eqref{eq-0y}:
\begin{equation}
    B = \p_n f + f \nabla \times \tau  = \p_n(\beta \kappa) + \beta \kappa \nabla \times \tau  = \kappa\p_n \beta  + \beta \p_n \kappa + \beta \kappa \nabla \times \tau. 
    \end{equation}
    Therefore, since $\kappa$ vanishes on $\Gamma$, we deduce that on $\Gamma$:
    \begin{equation}\label{eq-1b}
        \p_n B = 2\p_n \kappa \p_n \beta + \beta \big(\p_n^2 \kappa + \p_n \kappa \nabla \times \tau \big) = 2|\nabla \kappa| \p_n \beta + B\left( \dfrac{\p_n^2 \kappa }{|\nabla \kappa|}+\nabla \times \tau\right).
\end{equation}
Moreover, using $|\nabla \kappa| = \p_n \tau$ and $\nabla f\times J\nabla g=\nabla f\cdot\nabla g$:
\begin{equation}\label{eq-1l}
    \nabla \times \tau =  \dfrac{\nabla \times J\nabla \kappa}{|\nabla \kappa|} + \nabla \dfrac{1}{|\nabla \kappa|} \cdot J\nabla \kappa = \dfrac{\Delta \kappa}{|\nabla \kappa|} - \dfrac{\p_n^2 \kappa}{|\nabla \kappa|}.
\end{equation}
Plugging \eqref{eq-1l} into \eqref{eq-1b} and using the defining equation $\kappa_t=r_tc_t\partial_n\beta(y_t)$  yield \eqref{eq-1d}.
\end{proof}

\subsection{Shifted Fourier transform} We finally introduce the $L^2-$ unitary transformation
\begin{equation}\label{eq:FS}
   \mV_ta(z) =\dfrac{1}{\sqrt{2\pi}} \dint_{\Rm} e^{iz_1\xione} a(\xione,z_2+\gamma_t\xione) d\xione.
\end{equation}
Its inverse is given explicitly by
\begin{equation}\label{eq:FSinv}
   \mV_t^* a(\Xitot) = \dfrac{1}{\sqrt{2\pi}} \dint_{\Rm} e^{-iz_1\xione} a(z_1,\zetatwo-\gamma_t\xi) dz_1. 
\end{equation}
Above, $\xione$ is the one-dimensional dual Fourier variable to $z_1$ while the one-dimensional variable $\zetatwo$ is a shift of $z_2$ given by $\zetatwo=z_2+\gamma_t\xione$. 
This shifted partial Fourier transform, which reduces to a partial Fourier transform in the first variable $z_1\to\xione$ when $B=0$ and introduces an additional shift $\gamma_t\xione$ to $z_2$ otherwise is reminiscent of the parametrization of Landau level eigenfunctions in a Landau gauge when $B\not=0$.
The construction of our wavepackets involves many functions of the form $a(t,\Xitot)=f(t,\xione)\phi(\zetatwo)$, 
with $f$ describing the wavepacket profile along $\Gamma$ while $\phi$ describes the wavepacket profile across $\Gamma$.

The global transformation we need to consider is the composition of $\bU_t$ introduced in the preceding section and $\mV_t$ given above. We thus set $\fU_t = \bU_t \mV_t$. The operator in this new set of variables is $T:= \fU_t^* L \fU_t$ with $L$ defined in \eqref{eq:L}. We similarly define $T_j = \fU_t^* L_j \fU_t$ for $j\geq0$. We will also be using $D_\xi=-i\partial_\xi$ and $D_\zeta=-i\partial_\zeta$.

In $T_1$, it will turn out that the terms carried by $\sigma_2$ or $\sigma_3$ and the terms that are odd in $(D_\zeta, \zetatwo)$ do not appear to leading order in our expansion. Therefore, we use the notation 
\begin{equation}
   P_1  \equiv P_2
\end{equation}
if $P_1-P_2$ is a differential operator made of odd terms in $(\zetatwo,D_\zeta)$ and terms carried by $\sigma_2, \sigma_3$.

\begin{lemm}\label{lem-1d} With  $\fU_t = \bU_t \VV_t$ and $T_j = \fU_t^* L_j \fU_t$, we have:
\begin{align}\label{eq:T0}
    T_0 & = c_t (1+\sigma_1) \xione + D_\zeta \sigma_2  + \rho_t\zetatwo \sigma_3,
    \\
   \label{eq-1t} T_1 & \equiv  D_t + \dot{\te_t} \gamma_t \big(D_\zeta^2-\xione^2\big)  + \Big( \dfrac{\dot{\te_t}s_t}{2}  +  j_t \gamma_t (D_\xi \xione - D_\zeta \zetatwo) + k_t \big(\zetatwo^2 + \gamma_t^2 \xione^2 \big)\Big) \sigma_1.
\end{align}
Moreover $T_j$ is a linear combination  of differential operators of the form $\xi^k \zeta^\ell D_\xi^m D_\zeta^{j+1-k-\ell-m}$, with $k,l\geq0$ and coefficients bounded with respect to $t$.
\end{lemm}

\begin{proof} We note that $\mV_t$ satisfies the canonical relations
\begin{equation}\label{eq:relmVt}
  \VV_t^* (z_2 + \gamma_t D_1) \VV_t =  \zetatwo, \quad \VV_t^* D_2 \VV_t = D_\zeta, \quad \VV_t^* D_1 \VV_t = \xione, \quad \mV_t^* z_1 \mV_t = -(D_\xi+\gamma_t D_\zeta).
\end{equation}
We deduce \eqref{eq:T0} for $T_0 = \fU_t^* L_0 \fU_t$ from the expression  \eqref{eq-1s}. Thanks to \eqref{eq:relmVt} we obtain $\VV_t^* z_2 \VV_t = \zetatwo - \gamma_t \xione$, hence using the notation $\equiv$: 
\begin{align}
   \VV_t^* z_1 z_2 \VV_t & = -(D_\xi+\gamma_t D_\zeta) (\zetatwo - \gamma_t \xione) \equiv  \gamma_t (D_\xi \xione - D_\zeta \zetatwo)
\\
   \VV_t^* z_2^2 \VV_t & = (\zetatwo - \gamma_t \xione)^2 \equiv \zetatwo^2 + \gamma_t^2 \xi^2
\\
\VV_t^*  (z_1 D_2 - z_2 D_1) \VV_t & = -(D_\xi+\gamma_t D_\zeta) D_\zeta - (\zetatwo - \gamma_t \xione) \xione \equiv - \gamma_t (D_\zeta^2-\xione^2).
\end{align}
Moreover, we observe that 
\begin{equation}
  \mV_t^*D_t \mV_t = D_t + \dot \gamma_t \xione D_\zeta \equiv D_t. 
\end{equation}

 We deduce from \eqref{eq-1r} the expression for $T_1$:
 \begin{align}
   T_1 \equiv   D_t + \dot{\te_t} \gamma_t (D_\zeta^2-\xione^2)  + \Big( \dfrac{\dot{\te_t}s_t}{2}  +  j_t \gamma_t (D_\xi \xione - D_\zeta \zetatwo) + k_t \big(\zetatwo^2 + \gamma_t^2 \xione^2 \big) \Big) \sigma_1.
 \end{align}
 Finally, since $L_j$ is multiplication operator by a homogeneous polynomial of degree $j+1$, $T_j$ is a linear combination  of differential operators of the form $\xi^k \zeta^\ell D_\xi^m D_\zeta^{j+1-k-\ell-m}$, with coefficients bounded with respect to $t$. This completes the proof. 
\end{proof}



\subsection{Conjugation and asymptotic expansion}
Introducing $a=\fU_t^*\psi$ and recalling that $T=\fU_t^*L\fU_t$, equation \eqref{eq:L} is equivalent to $Ta=0$. With $L_j$ given in \eqref{eq:L0}--\eqref{eq:Lj} and $T_j=\fU_t^*L_j\fU_t$, we decompose $a$ as $a=\sum_{j\geq0}\eps^{\frac j2}a_j$ to reduce the equation $Ta=0$ to the triangular system
\begin{equation}\label{eq:seqT}
  T_0 a_0 =0 ,\quad T_1a_0+T_0a_1=0,\quad \dsum_{k=0}^{j} T_{j-k}a_k =0,\ j\geq2.
\end{equation}
In \S\ref{sec:model} we identify the kernel of $T_0$ and we produce explicit formula for its inverse on the orthogonal complement. We describe how to solve the subleading equation $T_1a_0+T_0a_1=0$ in \S\ref{sec:transport}. We finally present the higher-order asymptotic expansions and corresponding error estimates in \S\ref{sec:error}.

\section{Inversion of leading operator} \label{sec:model}

To study the kernel of the operator
\begin{equation}
     T_0 = c_t (1+\sigma_1) \xione + D_\zeta \sigma_2  + \rho_t\zetatwo \sigma_3,
\end{equation}
and invert it on the orthogonal complement, we first bring it to a normal form. Thanks to the change of variables $(\xi,\zeta) \mapsto (\rho_t^{1/2}c_t^{-1}\xinorm,\rho_t^{-1/2}\ynorm)$, we can assume that $\rho_t = c_t = 1$. Moreover, with $Q = Q^*= \frac{1}{\sqrt{2}}(\sigma_1+\sigma_3)$, we have
\begin{equation}
    Q \sigma_1 Q=\sigma_3, \quad  Q \sigma_2 Q=-\sigma_2, \quad Q \sigma_3 Q=\sigma_1.
\end{equation}
It suffices then to work with the model operator
\begin{equation}\label{eq:H}
  H=\xinorm (1+\sigma_3) + \sigma_1 \ynorm  - D_\ynorm  \sigma_2 = \begin{pmatrix} 2\xi & \fa_\ynorm  \\ \fa_\ynorm ^* & 0\end{pmatrix} ,\quad \fa_\ynorm  = \partial_\ynorm +\ynorm ,
\end{equation}
instead of $T_0$. Multiplication by $2\xi$ may also be written as $\fa_\xinorm+\fa_\xinorm^*$ with $\fa_\xinorm = \partial_{\xinorm}+\xinorm$.

\subsection{Functional setting.} \label{sec:hilbert}
To study the operator $H$, we use the standard basis of Hermite functions on $L^2(\Rm)$ given by
\begin{equation}
    h_0(\zeta)=\pi^{-\frac14}e^{-\frac12 \zeta^2}, \qquad h_n(\zeta) = \dfrac{1}{2^{\frac n2} \sqrt{n!}} (\fa_\zeta^*)^n h_0(\zeta),\ n\geq1.
\end{equation}
They satisfy the relations
\begin{equation}\label{eq:aop}
  \fa_\zeta h_0=0;\qquad \fa_\zeta h_n=\sqrt{2n} h_{n-1},\  n\geq1; \qquad \fa_\zeta^* h_n=\sqrt{2n+2} h_{n+1}, \ n\geq0.
\end{equation}
With the decomposition $L^2(\Rm)\ni \psi(\zeta)=\sum_{n\geq0} \psi_nh_n(\zeta)$, we may then define
\begin{align}\label{eq:inva}
    \fa_\zeta^{-1} \psi(\zeta) = \sum_{n\geq1} \frac{1}{\sqrt{2n}} \psi_{n-1} h_{n}(\zeta), \quad & \psi \in L^2(\Rm) \\
    (\fa_\zeta^*)^{-1} \psi(\zeta) = \sum_{n\geq0} \frac{1}{\sqrt{2n+2}} \psi_{n+1} h_{n}(\zeta), \quad & \psi \in ({\rm Ker}\,  \fa_\zeta)^\perp = \overline{{\rm Ran}\, \fa_\zeta^*} = \{\psi\in L^2(\Rm); \ \psi_0=0\}.
\end{align}

A natural scale of Hilbert spaces for $p\in\Nm$ associated to this decomposition is $\mS_p(\Rm)=\{\psi \in L^2(\Rm); (\fa_\zeta^*)^p\psi \in  L^2(\Rm)\}$ endowed with the norm
\begin{align}\label{eq:normp1}
    \|\psi\|^2_p = \|(\fa_\zeta^*)^p \psi \|^2_0 \ \cong\  \|\psi\|^2_0 + \|D_\zeta^p \psi\|_0^2 + \|\zeta^p \psi\|^2_0,  
\end{align}
where $\|\cdot\|_0$ is the usual $L^2(\Rm)-$norm and where $a\cong b$ when there exists $C_p>0$ such that $C_p^{-1}a\leq b \leq C_p a$. That the two above expressions for the norm are equivalent may be obtained by induction from the classical result
\begin{align}
    \|\fa_\zeta^* \psi \|^2_0 = \|D_\zeta \psi\|^2_0 + \|\zeta \psi\|^2_0 + \|\psi\|^2_0
\end{align}
for $\psi\in \mS_1(\Rm)$. Note that $\mS_0(\Rm)\equiv L^2(\Rm)$.

With the decomposition  $L^2(\Rm^2)\ni \psi(\xi,\zeta)=\sum_{m,n\geq0}\psi_{mn}h_m(\xi)h_n(\zeta)$ in two-dimensional spaces, we similarly define $\mS_p(\Rm^2)= \{ \psi\in L^2(\Rm^2); (\fa_\xi^*)^p\psi\in L^2(\Rm^2) \mbox{ and } (\fa_\zeta^*)^p\psi\in L^2(\Rm^2)\}$ endowed with the norms
\begin{align}\label{eq:normp2}
    \|\psi\|^2_p = \dsum_{j=0}^p \|(\fa_\xi^*)^j(\fa_\zeta^*)^{p-j} \psi \|^2_0 \ \cong\   \|(1+|\xi|^p+|\zeta|^p) \psi\|^2_0+ \|D_\xi^p \psi\|_0^2 + \|D_\zeta^p \psi\|_0^2.
\end{align}
For vector-valued functions, we also define the spaces $\mS_p(\Rm^d,\Cm^q)$ component-wise. Below we simply write $\mS_p$ when the domain $\Rm^d$ and range $\Cm^q$ are clear and $\|\cdot\|_p$ the associated norm.

\subsection{Inversion of model operator.}
As the lemma below demonstrates, the kernel $N\subset L^2(\Rm^2,\Cm^2)$ of $H$ and its orthogonal complement are given by
\begin{align}
    \label{eq:N}
      N=&\left\{\psi(\xinorm,\ynorm)=f(\xinorm)h_0(\ynorm) \matrice{ 0 \\ 1},\ f\in L^2(\Rm)\right\},\\  N^\perp =& \left\{\psi \in L^2(\Rm^2,\Cm^2),\ \big(\psi_2(\xi,\cdot),h_0(\cdot)\big)_2=0,\  \forall \xi\in\Rm \right\},
\end{align}
where $(\cdot,\cdot)_2$ is the standard $L^2(\Rm)-$inner product in the second variable $\zeta$ and $\psi=(\psi_1,\psi_2)^t$.
\begin{lemm}\label{lem:invT0} For $p\in \Nm$, we define the bounded linear operator $H^{-1}$ from $N^\perp\cap\mS_{p+1}$ to $N^\perp\cap \mS_{p}$ as
\begin{align}
    H^{-1} = \matrice{ 0  & (\fa_\zeta^*)^{-1} \\ \fa_\zeta^{-1} & - \fa_\zeta^{-1} 2\xi (\fa_\zeta^*)^{-1} } .
\end{align}

All solutions $\psi\in \mS_{p}$ of the equation $H\psi=g$ for $g \in N^\perp\cap\mS_{p+1}$ are of the form $\psi=H^{-1}g+\psi_0$ with $\psi_0$ arbitrary in $N\cap\mS_{p}$. In particular ${\rm Ker}\, H=N$ on $\mS_0$.
\end{lemm}
\begin{proof} 
The system $H\psi=g$ is equivalent to $\fa^*_\zeta \psi_1=g_2$ and $\fa_{\zeta}\psi_2 + 2\xi \psi_1 = g_1$. Using \eqref{eq:inva}, we find for $(g_2,h_0(\zeta))_2=0$, i.e., $g\in N^\perp$, that $\psi_1=(\fa^*_\zeta)^{-1}g_2$ and then $\psi_2=\fa_\zeta^{-1} (g_1-2\xi\psi_1)$ and hence the above result. We then note that $\fa_\zeta^{-1}$ has range in $N^\perp$. 

That $\psi\in\mS_{p}$ when $g\in\mS_{p+1}$ follows directly from the fact that the following operators are bounded: $\fa_\zeta^{-1}:\mS_{p+1}\to\mS_{p+1}\cap N^\perp$, $(\fa_\zeta^*)^{-1}:\mS_{p+1}\cap N^\perp\to\mS_{p+1}$, as well as $\xi:\mS_{p+1}\to \mS_p$.

   
Any function $\psi\in\mS_{p}$ solution of $H\psi=g$ may be decomposed as $\psi_0+\psi_1$ with $\psi_0\in N$ and $\psi_1\in N^\perp$. Since $H\psi_0=0$ and hence $\psi_1=H^{-1}g$, $\psi_0$ is arbitrary in ${\rm Ker}\,H=N$.
\end{proof}

\subsection{Microscopic balance.} 

Our geometric assumptions impose that $C^{-1}\geq\rho_t\geq C>0$ and that $0<c_0\leq c_t<1$ for constants $C$ and $c_0$ independent of $t\geq0$. 

For a fixed time $t$, we now solve the equation $T_0 a=b$. 
We introduce
\begin{equation}\label{eq:phi}
  \phi_t(\zetatwo) := \left(\frac {\rho_t}{4\pi}\right)^{\frac14} e^{-\frac{\rho_t}{2} \zetatwo^2} \matrice{ 1\\-1 },\qquad N_t = {\rm Ker}\ T_0 = \big\{ f(\xione) \phi_t(\zetatwo); \ f\in \mS_0(\Rm;\Cm) \big\}.
\end{equation}
The normalization implies that $\|\phi_t\|_{L^2(\Rm,\Cm^2)}=1$. That $N_t$ is the kernel of $T_0$ comes from Lemma \ref{lem:invT0} and the change of variables given in the last section.
We still denote by $\mS_p(\Rm^2,\Cm^2)$ the spaces of functions in the $(\xione,\zetatwo)$ variables, which are equivalent to the corresponding spaces in the variables $(\rho_t^{-1}c_t\xione,\rho_t^{1/2}\zetatwo)$ by assumption on $(c_t,\rho_t)$. 

The leading order equation is then solved as follows.

\begin{lemm}\label{lem:T0}  
Let $b\in N_t^\perp \cap \mS_{p+1}$. The equation $T_0a=b$ admits a unique solution $T_0^{-1}b:=a\in N_t^\perp \cap \mS_{p}$ with inverse operator given explicitly by
\begin{align}
    T_0^{-1} = Q \matrice{ 0 & (\rho_t\zeta-\partial_\zeta)^{-1} \\ (\rho_t\zeta+\partial_\zeta)^{-1} & - 2\xi c_t (\rho_t\zeta+\partial_\zeta)^{-1}(\rho_t\zeta-\partial_\zeta)^{-1} } Q ,\quad Q=\frac{1}{\sqrt 2} (\sigma_1+\sigma_3).
\end{align}
All solutions of that equation in $\mS_{p}$ are of the form $a=T_0^{-1}b+f(\xione)\phi(\zetatwo)$ for arbitrary $f(\xione)\in \mS_p(\Rm,\Cm)$.
\end{lemm}
\begin{proof}
This is a corollary of Lemma \ref{lem:invT0} and the invertible transform from $T_0$ to $H$.
The operator $T_0$ is thus invertible on the orthogonal complement of $\phi_t(\zetatwo)$ with $T_0a=b$ solvable if and only if $(b,\phi_t)_2=0$ in which case all solutions are given by $a=T_0^{-1}b+f(\xione)\phi_t$ for $f$ arbitrary in $\mS_p$. 
\end{proof}

The solution to the leading equation $T_0a_0=0$ is therefore given by 
\begin{equation}\label{eq:a0}
  a_0(t,\Xitot) = f_0(t,\xione) \phi_t(\zetatwo)
\end{equation}
with $f_0(t,\xione)$ arbitrary at this level. 

%
%
%

\section{Transport equation}\label{sec:transport}
Lemma \ref{lem:T0} states that the subleading equation in \eqref{eq:seqT}, $T_0a_1=-T_1a_0$, admits a solution if and only if $T_1a_0(t,\cdot) \in N_t^\perp$ for every $t$. From the expressions \eqref{eq:a0} of $a_0$ and \eqref{eq:phi} of $N_t$ and $\phi_t(\zeta)$, we deduce that $T_0a_1=-T_1a_0$ is solvable if and only if $\mT f_0 = 0$, where 
\begin{equation}\label{eq:mT}
    \mT f(t,\xione):=  \int_{\mathbb{R}} \phi_t(\zetatwo)\cdot T_1 [f(t,\xione) \phi_t(\zetatwo)] \ d\zetatwo.
\end{equation}
In this section we provide an explicit expression for $\mT$ and for the solutions to $\mT f_0 = 0$ and more generally $\mT f_0=g$. We also provide a functional setting to analyze the map $g\to f_0$ when the dispersion is strongest.

\subsection{Derivation of the transport operator}

\begin{lemm}\label{lem-1h} We have the identity
\begin{equation}\label{eq:transport}
\mT = D_t  - \dfrac{k_t}{2\rho_t}  - j_t \gamma_t \dfrac{\xione D_\xi +  D_\xi \xione}{2} -  \left(\dot{\te_t} \gamma_t + k_t\gamma_t^2  \right)\xione^2.
\end{equation}
\end{lemm}

\begin{proof} We first observe that 
\begin{equation}
    \left \langle \matrice{1 \\ -1}, \sigma_j \matrice{1 \\ -1} \right \rangle = -2 \delta_{1j}. 
\end{equation}
Therefore, the terms in $T_1$ that are carried by $\sigma_2, \sigma_3$ do not contribute to $\mT$. Likewise, we observe that
\begin{equation}
    \int_{\mathbb{R}} \phi(\zetatwo) \cdot \zetatwo \phi(\zetatwo) d\zetatwo = \int_{\mathbb{R}} \phi(\zetatwo) \cdot D_\zeta \phi(\zetatwo) d\zetatwo = 0. 
\end{equation}
Hence, for the purpose of computing $\mT$, we can ignore terms in $T_1$ carried by $\sigma_2$ or $\sigma_3$, and terms linear in $(\zetatwo,D_\zeta)$. In other words, we can replace $T_1$ in \eqref{eq:mT} by the right hand side of \eqref{eq-1t}:
\begin{equation}
    D_t + \dot{\te_t} \gamma_t \big(D_\zeta^2-\xione^2\big)  + \left( \dfrac{\dot{\te_t}s_t}{2}  +  j_t \gamma_t (D_\xi \xione - D_\zeta \zetatwo) + k_t \big(\zetatwo^2 + \gamma_t^2 \xione^2 \big)\right) \sigma_1.
\end{equation}

We note moreover that
\begin{equation}\label{eq-0x}
     D_\zeta \zetatwo\phi =  i(\rho_t \zetatwo^2-1) \phi, \quad D_\zeta^2 \phi(\zetatwo) = (\rho_t - \rho_t^2 \zetatwo^2) \phi(\zetatwo), \ \ \ \ D_t \phi(\zetatwo) = i \dfrac{\dot{\rho_t}}{2} \zetatwo^2 \phi(\zetatwo) - i \dfrac{\dot{\rho_t}}{4\rho_t} \phi(\zetatwo).
\end{equation}

Therefore, we deduce that 
\begin{equation}
\begin{array}{l}
   \mT =\Big(\dfrac{\rho_t}{\pi}\Big)^{\frac12} \dint_{\mathbb{R}}    e^{-\rho_t \zetatwo^2} 
   \Big( D_t + i \dfrac{\dot{\rho_t}}{2} \zetatwo^2 -i\frac1{4} \frac{\dot\rho_t}{\rho_t}
   + \dot{\te_t} \gamma_t \big(\rho_t   - \rho_t^2 \zetatwo^2-\xione^2\big) \\[3mm]  \qquad - \Big( \dfrac{\dot{\te_t}s_t}{2}  +  j_t \gamma_t (D_\xi \xione - i \rho_t \zetatwo^2 + i) + k_t \big(\zetatwo^2 + \gamma_t^2 \xione^2 \big)\Big) \Big)
        d\zetatwo.
\end{array}    
\end{equation}

Furthermore, we have
\begin{equation}
   \left(\frac {\rho_t}{\pi}\right)^{\frac12} \int_{\mathbb{R}} 
   e^{-\rho_t \zetatwo^2}  d\zetatwo = 1, \ \ \ \ \left(\frac {\rho_t}{\pi}\right)^{\frac12} \int_{\mathbb{R}} \zetatwo^2 
   e^{-\rho_t \zetatwo^2}  d\zetatwo = \dfrac{1}{2\rho_t}.
\end{equation}
Hence, after performing the integration and realizing that the coefficients involving $\dot\rho_t$ cancel out, we obtain the formula
\begin{align}
\label{eq-0r} \mT & = 
    D_t 
    + \dot{\te_t} \gamma_t \left(\dfrac{\rho_t}{2}-\xione^2\right) - \dfrac{\dot{\te_t}s_t}{2}  -  j_t \gamma_t \left(D_\xi \xione + \dfrac{i}{2}\right) - k_t \left(\dfrac{1}{2\rho_t} + \gamma_t^2 \xione^2 \right)
    \\ 
    \label{eq-0s} 
    & = D_t 
    - \dot{\te_t} \gamma_t \xione^2 -  j_t \gamma_t \dfrac{\xione D_\xi +  D_\xi \xione}{2} - k_t \left(\dfrac{1}{2\rho_t} + \gamma_t^2 \xione^2 \right)
    \\ 
    & = D_t 
    - \dfrac{k_t}{2\rho_t}  -  j_t \gamma_t \dfrac{\xione D_\xi +  D_\xi \xione}{2} -  \left(\dot{\te_t} \gamma_t + k_t\gamma_t^2  \right)\xione^2, 
\end{align}
where in the second line we used $\gamma_t \rho_t = s_t$.
\end{proof}

\subsection{Solving the transport equation} \label{sec-6.2}
To produce an explicit solution of the transport equation $\mT f = 0$ we define $(\lambda_t,\mu_t,\nu_t)$ such that
\begin{equation}\label{eq-1a}
   \lambda_t =  \int_0^t \dfrac{k_s}{2\rho_s}ds, \qquad \nu_t =  2\int_0^t \dfrac{c_0^2}{c_s^2} \left( \dot{\te_s} \gamma_s + k_s \gamma_s^2\right) ds, \qquad e^{\mu_t} = \frac{c_t}{c_0}.
\end{equation}

\begin{lemm}\label{lem-1g} The solution to $\mT f = 0$ is given by
\begin{equation}\label{eq-1n}
  f(t,\xione) =  \exp\left( i \lambda_t + i\dfrac{\nu_t}{2} \left(e^{\mu_t}\xi \right)^2 \right) e^{\frac{\mu_t}2} f\left(0, e^{\mu_t} \xi \right).
\end{equation}
\end{lemm}

\begin{proof}
1. We recall that the self-adjoint operator $\frac{1}{2} (\xione D_\xi +  D_\xi \xione)$ generates the semigroup of ($L^2-$unitary) dilations; that is, for $F$ independent of $t$:
\begin{equation}\label{eq-1k}
     \left(D_\mu - \dfrac{\xione D_\xi + D_\xi \xione}{2}\right) U_\mu F = 0, \quad U_\mu F(\xione) = e^{\frac\mu 2} F\big(e^\mu \xione\big).
\end{equation}
We note that $e^{\mu_0} =1$; moreover, thanks to \eqref{eq-1d}, we have $\dot{\mu_t} = \p_t \ln(c_t) = j_t \gamma_t$. Therefore, by the chain rule and \eqref{eq-1k}, we have
\begin{equation}
    U_{\mu_t}^{-1} \left(D_t - j_t \gamma_t \dfrac{\xione D_\xi + D_\xi \xione}{2}\right) U_{\mu_t} = D_t.
\end{equation}

2. We deduce, using $U_{\mu_t}^{-1} \xione^2 U_{\mu_t}=e^{-2\mu_t}\xione^2$, that  
\begin{align}
   U_{\mu_t}^{-1} \mT U_{\mu_t} & =  D_t - \dfrac{k_t}{2\rho_t} - e^{-2\mu_t} \left( \dot{\te_t} \gamma_t + k_t \gamma_t^2\right) \xione^2
   = D_t - \dot{\lambda_t} - \dfrac{1}{2}\dot{\nu_t} \xione^2  = e^{i\lambda_t + i\nu_t \frac{\xione^2}{2}} D_t e^{-i\lambda_t - i \nu_t \frac{\xione^2}{2}},
\end{align}
where we used the relations \eqref{eq-1a} for $\lambda_t$ and $\nu_t$. For the formula \eqref{eq-1k} for $U_\mu$, we deduce that the solution to $\mT f=0=\mT U_{\mu_t}e^{i\lambda_t + i \frac{\nu_t}{2} \xione^2} g$ is $D_tg=0$ and hence
\begin{equation}
    f(t,\cdot) = U_{\mu_t} e^{i\lambda_t + i \frac{\nu_t}{2} \xione^2} f(0,\cdot).
\end{equation}
Since $e^{\mu_t} = c_t/c_0$, we conclude that
\begin{equation}\label{eq-3r}
    f(t,\xione) = e^{\frac{\mu_t}2} \exp\left( i \lambda_t + i\frac{\nu_t}{2} \left( e^{\mu_t}\xi \right)^2 \right) f\left(0, e^{\mu_t}\xi \right).
\end{equation}
This completes the proof. 
\end{proof}

\subsection{Dispersive estimate}

In this section, we study the $L^\infty$-decay of the leading order solution of \eqref{eq:L}, $\psi_0 = \fU_t a_0$, where $a_0$ takes the form prescribed by \eqref{eq:a0} and \eqref{eq-1n}:
\begin{equation} \label{eq-00m}
    a_0 = e^{i\lambda_t} \tilde a_0 \matrice{1\\-1}, \quad  \tilde a_0(t,\xi,\zeta)= \left( \dfrac{\rho_t}{4\pi}\right)^{\frac14}  e^{\frac{\mu_t}2}  \exp\left(i\dfrac{\nu_t}{2} \left( e^{\mu_t}\xione \right)^2 -\frac{\rho_t}{2} \zetatwo^2 \right)  \hf\left(e^{\mu_t}\xione \right),
\end{equation}
where $\hf(\xi)$ denotes the Fourier transform of a function $f(z_1)$. 



\begin{lemm} \label{lem-1i} (i) With $\psi_0 = \fU_t a_0$ and $a_0$ given by \eqref{eq-00m}, we have whenever $B_t \neq 0$ or $\nu_t \neq 0$:
\begin{equation}\label{eq-3i}
    \psi_0(t,z) = e^{i\lambda_t} \bU_t \big( G_t *_1 f \big) 
    (e^{-\mu_t} z_1, e^{\mu_t} z_2) 
    \matrice{1\\-1},
\end{equation}
where $*_1$ denotes convolution with respect to the first variable and 
    \begin{equation}\label{eq-3h}
    G_t(e^{-\mu_t} z_1,e^{\mu_t}z_2) = 
    \left( \dfrac{\rho_t}{4\pi}\right)^{\frac14}
    \Big( \frac{e^{\mu_t}}{Q_t} \Big)^{\frac12}
     e^{-\frac12 \rho_tz_2^2}  e^{-\frac12 Q_t^{-1}(z_1+i s_t z_2)^2}, \quad Q_t = s_t\gamma_t -i e^{2\mu_t}\nu_t.
\end{equation}
We recall that $s_t=\gamma_t\rho_t$. 

(ii) In particular, there exists $C > 0$ such that as long as $\nu_t \neq 0$:
\begin{equation}\label{eq-3k}
    \sup_{z \in \R^2} \big| \psi_0(t,z) \big| \leq   C  \min \Big ( \dfrac{\| f \|_{L^1}}{|\nu_t|^{1/2}}  \,,\,\|\fhat\|_{L_1} \Big).
\end{equation}
\end{lemm}

\begin{proof} We recall that $\fU_t = \bU_t \mV_t$ and we first write a formula for $\mV_t \tilde a_0$. We have:
\begin{equation}
    \mV_t \tilde a_0(z) =  \left( \dfrac{\rho_t}{4\pi}\right)^{\frac14}e^{\frac{\mu_t}2} \int_\R e^{iz_1\xione}  \exp\left( i\dfrac{\nu_t}{2} \left( e^{\mu_t}\xione \right)^2 -\frac{\rho_t}{2} (z_2+\gamma_t \xi)^2 \right) \fhat\left(e^{\mu_t}\xione \right) \dfrac{d\xione}{\sqrt{2\pi}}.
\end{equation}
Hence,
\begin{align}
    \mV_t \tilde a_0(e^{\mu_t}z_1,\frac{z_2}{e^{\mu_t}}) &= \left( \dfrac{\rho_t}{4\pi}\right)^{\frac14}
    \int_\R e^{i e^{\mu_t} z_1\xione}  
    \exp\left( i\dfrac{\nu_t}{2} \left( e^{\mu_t}\xione \right)^2 -\frac{\rho_te^{-2\mu_t}}{2} (z_2+\gamma_t e^{\mu_t}\xi)^2 \right) 
    \fhat\left(e^{\mu_t}\xione \right) \dfrac{e^{\frac{\mu_t}2} d\xione}{\sqrt{2\pi}}    
    \\ \label{eq:secondbound}
    &=\left( \dfrac{\rho_t}{4\pi}\right)^{\frac14}
     \int_\R e^{iz_1\xione}  \exp\left( i\dfrac{\nu_t}{2} 
    \xione^2 -\frac{\rho_te^{-2\mu_t}}{2} (z_2+\gamma_t \xi)^2 \right) 
    \fhat\left(\xione \right) \dfrac{e^{-\frac{\mu_t}2}d\xione}{\sqrt{2\pi}}.
\end{align}
The left-hand side involves the inverse Fourier transform in $z_1$ of a product and may therefore be written as the convolution
\begin{align}\label{eq:convta0}
    \mV_t \tilde a_0(e^{\mu_t}z_1,z_2) = (f *_1 G_t)(z_1,e^{\mu_t}z_2)
\end{align}
where the Gaussian $G_t$ is given by
\begin{equation}\label{eq-3m}
    G_t(z) = \left( \dfrac{\rho_t}{4\pi}\right)^{\frac14}e^{\frac{-\mu_t}2} \int_\R e^{iz_1 \xi }  \exp\left( i\dfrac{\nu_t}{2} \xi^2 -\frac{\rho_te^{-2\mu_t}}{2} (z_2+\gamma_t \xi )^2 \right) \dfrac{d\xione}{\sqrt{2\pi}},
\end{equation}
so that 
\begin{align}
    G_t(z_1,e^{\mu_t}z_2) &= \left( \dfrac{\rho_t}{4\pi}\right)^{\frac14}e^{\frac{-\mu_t}2} e^{-\frac12 \rho_tz_2^2} \dint_{\Rm} e^{-\frac12 \tilde Q_t\xi^2} e^{i\xi(z_1+i\rho_t\gamma_t e^{-\mu_t} z_2)} \dfrac{d\xione}{\sqrt{2\pi}},\qquad \tilde Q_t = \rho_t\gamma_t^2e^{-2\mu_t} -i\nu_t 
    \\ \label{eq:Qt}
    &= \left( \dfrac{\rho_t}{4\pi}\right)^{\frac14}e^{\frac{-\mu_t}2} e^{-\frac12 \rho_tz_2^2} \tilde Q_t^{-\frac12} e^{-\frac12 \tilde Q_t^{-1}(z_1+i\rho_t\gamma_t e^{-\mu_t} z_2)^2},
\end{align}
as the Fourier transform of a Gaussian function when $\tilde Q_t\not=0$.  Since $Q_t=e^{2\mu_t}\tilde Q_t$, we find \eqref{eq-3h}. Then \eqref{eq-00m} follows from applying $\bU_t$ to $\mV_t a_0$ using \eqref{eq-00m} and  \eqref{eq:convta0}.

To prove the second part of the lemma, we first take $L^\infty$-norms on both sides of \eqref{eq:convta0} and we apply Young's $L^1-L^\infty$ convolution inequality  (in the variable $z_1$). This produces:
\begin{align}\label{eq-3q}
    \sup_{z \in \R} \big| \mV_t \tilde a_0(z) \big| & = \sup_{z_2 \in \R} \sup_{z_1 \in \R}\big| f *_1 G_t(z) \big| \leq \sup_{z_2 \in \R} \| f \|_{L^1} \sup_{z_1 \in \R} \big| G_t(z) \big| = \| f \|_{L^1}  \sup_{z \in \R} \big| G_t(z) \big|.
\end{align}
The Gaussian $|G_t|$ attains its maximum at $z = 0$ because the real part of the quadratic form in $z$ in \eqref{eq-3h} satisfies
\begin{align}
    \Re \Big(\rho_tz_2^2 + \frac{1}{ Q_t}(z_1+is_t z_2)^2\Big) = \frac{\rho_t}{(s_t\gamma_t)^2+\tilde \nu_t^2}  (\gamma_t z_1 - \tilde \nu_t z_2)^2 \geq0, \qquad \tilde\nu_t=e^{2\mu_t}\nu_t,
\end{align}
as we verify by an elementary computation using $\rho_t\gamma_t=s_t$.
Using \eqref{eq-3q}, we find
\begin{equation}
        \sup_{z \in \R} \big| \mV_t \tilde a_0(z) \big| \leq \| f \|_{L^1}  \big| G_t(0) \big| =  \left( \dfrac{\rho_t}{4\pi}\right)^{\frac14}\left( \dfrac{e^{\mu_t}}{|Q_t|} \right)^{\frac12} \| f \|_{L^1}.
\end{equation}
From \eqref{eq:secondbound}, we also obtain that $\sup_{z \in \R} \big| \mV_t \tilde a_0(z) \big|\leq C \|\hat f\|_{L^1}$.
To end up with \eqref{eq-3k}, we observe that $\bU_t$ preserves $L^\infty$-norms, that $|Q_t| \geq |\nu_t|$ and that $c_t$ is bounded above and below for all times $t$. 
\end{proof}

When $Q_t = 0$, we have $\nu_t = 0$, $c_t = 1$, $\rho_t = r_t$ and the formulas \eqref{eq:convta0} and \eqref{eq-3m} remain valid. But instead of being a Gaussian, $G_t$ is now a multiple of the Dirac mass, and
\begin{equation}
    G_t(z) =  c_0^{1/2} \left( \dfrac{r_t}{4\pi}\right)^{1/4}\exp\left( -\frac{r_t}{2} z_2^2 \right) \delta_0(z_1), 
    \qquad \psi_0(t,z) = c_0^{1/2} \left( \dfrac{r_t}{4\pi}\right)^{1/4} \bU_t e^{-\frac{r_t}{2} z_2^2} f(z_1).
\end{equation}




\subsection{Dispersion-dependent functional setting} 
\label{sec:functional}


The solution of $\mT f=0$ with given initial condition admits an explicit expression as we saw in Lemma \ref{lem-1g}. In the analysis of the asymptotic expansion of $\psi(t,z)$ in powers of $\eps$, we need to solve transport equations of the form $\mT f=g$ with time-dependent source terms. To quantify the stability of the inverse transport operators and that of other relevant transforms, we introduce the following functional setting.

We recall the functional spaces $\mS_p$ were defined in section \ref{sec:hilbert}. To handle the time-dependence for $t\in [0,\rT]$ of the wavepackets, we introduce the spaces for $p\in\Nm$ and $k=\lfloor \frac p2 \rfloor \in \Nm$ (i.e., $p=2k$ or $p=2k+1$) defined for an interval ${\rm I}\subset\Rm$ by
\begin{equation}\label{eq:CSp}
  \mCS_p({\rm I},\Rm^d,\Cm^q) = \cap_{r=0}^k C^r({\rm I}; \mS_{p-2r}(\Rm^d,\Cm^q))
\end{equation}
with norm given by the sum of the natural norms for the above spaces (see \eqref{eq:normsmCS} below). The spaces are constructed so that any derivative in time corresponds to a loss of order $2$ in the remaining variable. 


To quantify the effects of dispersion, we define 
\begin{align} \label{eq:anut}
    \anut= 1+\sup_{0\leq t\leq \rT} |\nu_{t}|  \quad \mbox{ and } \quad \vsa=\anut^{-\frac12}.
\end{align}
We saw in Lemma \ref{lem-1g} that dispersion resulted in a multiplication operator of the form $\phi_t(\xi)=e^{i\frac12 \nu(t)\xi^2}$. 
The operator $D_t \phi_t = \phi_t (D_t+\frac12 \nu'(t)\xi^2)$. This explains why the spaces $\mCS_p$ are constructed so that both $\partial_t$ and $\xi^2$ map  $\mCS_{p+2}$ to $\mCS_p$. 

Similarly, $\partial_\xi \phi_t = \phi_t(\partial_\xi+i\nu(t)\xi)$, so that (the operator of multiplication by) $\phi_t$ is large as an operator on $\mCS_1$ when $\anut$ is. 
Yet clearly, $\phi_t^2$ is comparable to $\phi_t$ in the same sense. Since the construction of our wavepackets requires repeated application of operators of the form $\phi_t$, we introduce scaled metrics on $\mS_p$ and $\mCS_p$ so that application of $\phi_t$ results in a bounded operation independent of $\anut$. 

This is simply achieved by replacing $\fa_\xi^*=-\partial_\xi+\xi$ by $\fa_{\xi\vsa}^*=-\vsa \partial_\xi + \vsa^{-1}\xi$ and endowing the spaces $\mS_p(\Rm)$ and $\mS_p(\Rm^2)$ respectively with the norms
\begin{align} \label{eq:normsmS}
 \|\psi\|^2_{p\vsa} = \|(\fa_{\xi\vsa}^*)^{p} \psi\|^2_0 \quad \mbox{ and } \quad \|\psi\|_{p\vsa}^2 =  \dsum_{j=0}^{p} \|(\fa_{\xi\vsa}^*)^j (\fa_\zeta^*)^{p-j}\psi \|^2_0.
\end{align}
We call $\mS_{p\vsa}$ the spaces $\mS_p$ endowed with these dispersion-scaled norms.

The space $\mCS_p({\rm I},\Rm^d,\Cm)$  are similarly endowed for $d=1,2$ with the norms (and called $\mCS_{p\vsa}$)
\begin{align}\label{eq:normsmCS}
     \|\psi\|^2_{p\vsa} = \sup_{t\in {\rm I}} \dsum_{r=0}^{\lfloor \frac p2 \rfloor} \|(\fa_{\xi\vsa}^*)^{p-2r} \partial_t^r \psi\|^2_0 \ \mbox{ and } \ \|\psi\|_{p\vsa}^2 =  \sup_{t\in {\rm I}} \dsum_{r=0}^{\lfloor \frac p2 \rfloor} \dsum_{j=0}^{p-2r} \|(\fa_{\xi\vsa}^*)^j (\fa_\zeta^*)^{p-2r-j} \partial_t^r\psi \|^2_0.
\end{align}
Spaces of vector-valued functions are similarly constructed componentwise.
Note that $\vsa^{p}\|f\|_p \lesssim  \|f\|_{p\vsa} \lesssim  \vsa^{-p} \|f\|_{p}$. Here and below, we use the notation $a\lesssim b$ to mean the existence of a $\vsa-$independent constant $C$ such that $a\leq Cb$.
Here is a number of dispersion-dependent estimates we will be using.
%
\begin{lemm}\label{lem:pvsa}
 All operator bounds below are meant to be $\vsa$-independent bounds.
 
 Any operator $B\in \{ \vsa D_\xi, \vsa^{-1}\xi, D_\zeta, \zeta\}$ is bounded from $\mS_{p+1,\vsa}$ to $\mS_{p\vsa}$ and from $\mCS_{p+1,\vsa}$ to $\mCS_{p\vsa}$.  The operator $\partial_t$ is bounded from $\mCS_{p+2,\vsa}$ to $\mCS_{p\vsa}$.
 
 Let $\phi(\xi)=e^{i\frac12 \nu \xi^2}$ with $|\nu|\lesssim \anut$. Then the operator of multiplication by $\phi(\xi)$ is bounded from $\mS_{p\vsa}$ to itself. 
 
 Let $\phi(t,\xi)=e^{i\frac12 \nu(t) \xi^2}$ with $\sup_{t\in {\rm I}}|\nu(t)|\lesssim \anut$ and $\sup_{t\in {\rm I}} |\nu^{(j)}(t)|\leq C_j$ for $j\geq1$. Then the operator of multiplication by $\phi(t,\xi)$ is bounded from $\mCS_{p\vsa}$ to itself. 
\end{lemm}
\begin{proof}
The first statement comes from the construction of the spaces since $\fa^*_{\xi\vsa}$ controls $\vsa\partial_\xi$ and $\vsa^{-1}\xi$ in the sense that \begin{align}
    \|\fa^*_{\xi\vsa} \psi\|_0^2 = \vsa^2\|\partial_\xi \psi\|_0^2 + \vsa^{-2}\|\xi \psi\|_0^2 + \|\psi\|^2_0.
\end{align}
One obtains from \eqref{eq:normsmCS} for $p\geq2$ that $\|\partial_t\psi\|^2_{p\vsa}\lesssim \|\psi\|^2_{p-2,\vsa}$ and hence the bound on $\partial_t$.

Consider now the operator of multiplication by $\phi(\xi)$ in one dimension $d=1$. We wish to show that 
\begin{align}
    \|\phi\psi\|_{p\vsa}^2 =  \|(\fa_{\xi\vsa}^*)^p (\phi\psi)\|^2_0 \lesssim  \|(\fa_{\xi\vsa}^*)^p \psi\|^2_0 = \|\psi\|_{p\vsa}^2.
\end{align}
This holds when $p=0$. Assume it holds for $p-1\geq0$. Then
\begin{align}
    \|\phi\psi\|_{p\vsa}^2 =  \|(\fa_{\xi\vsa}^*)^{p-1} (\fa_{\xi\vsa}^*)(\phi\psi)\|^2_0 =\|(\fa_{\xi\vsa}^*)^{p-1} (\phi\psi_1)\|^2_0 \lesssim \|\psi_1\|_{p-1,\vsa}^2
\end{align}
by induction hypothesis, where $\psi_1(\xi)=(\fa_{\xi\vsa}^*-i\vsa \nu \xi)\psi$. By construction of the functional spaces and the above result for the operator $\vsa^{-1}\xi$ knowing that $\vsa |\nu|\lesssim \vsa^{-1}$, we find $\|\psi_1\|_{p-1,\vsa}^2\lesssim \|\psi\|_{p,\vsa}^2$ and the result is proved. The same proof applies in two dimensions $d=2$ as well using the norm for $\mS_p(\Rm^2)$ in \eqref{eq:normsmS} and that the commutator $[\fa_\zeta^*,\phi]=0$.

The proof in the time-dependent setting uses that
\begin{align}
    \partial_t^r (\phi\psi)= \phi \big(\partial_t+i\frac12 \nu'(t)\xi^2\big)^r \psi = \phi \psi_r,\qquad \psi_r=\big(\dsum_{j=0}^r \nu_j(t) \xi^{2(r-j)} \partial_t^j \big) \psi
\end{align}
for smooth and bounded functions $\nu_j(t)$ independent of $\anut$ by assumption on $\nu(t)$. Assume dimension $d=1$ as $d=2$ is treated similarly.  Using that $\xi^{2(r-j)} \partial_t^j$ maps $\mCS_{p\vsa}$ to $\mCS_{p-2r,\vsa}$ for each $0\leq j\leq r$ and the bounds proved above in the time-independent setting, we find
\begin{align}
    \|(\fa_{\xi\vsa}^*)^{p-2r} \partial_t^r \phi\psi\|_0 = \|(\fa_{\xi\vsa}^*)^{p-2r} \phi \psi_r\|_0 \lesssim \|\psi_r\|_{p-2r,\vsa} \lesssim \|\psi\|_{p\vsa}.
\end{align}
The explicit expression of the norms in \eqref{eq:normsmCS} and the above estimate conclude the proof of the lemma.
\end{proof}

\subsection{Stability of the transport and other operators}
\label{sec:invtrans}

We then have the following stability result for the transport equation $\mT f=g$. We recall that $\vsa=\anut^{-\frac12}$.

\begin{lemm}\label{lem:mT}
 Let $p\in \Nm$. The solution $f_0(t,\xione)$ of $\mT f_0=g$ on $[0,\rT]\times\Rm$ with initial condition $f_0(0,\xione)=\fhat(\xione)\in \mS_{p\vsa}(\Rm,\Cm)$ and source term $g\in \mCS_{p\vsa}(\Rm,\Cm)$ satisfies the estimate
  \begin{equation}
    \|f_0\|_{p\vsa} \leq C_{p} \big(\|\fhat\|_{p\vsa} + \aver{\rT} \|g\|_{p\vsa}\big)
 \end{equation}
where $\aver{\rT}=1+\rT$ and $C_{p}$ is independent of $\rT$. 
\end{lemm}

\begin{proof}
We first adapt Lemma \ref{lem-1g} to handle volume sources and define
\[
  \lambda_{s,t} = \int_s^t \frac{k_\tau}{2\rho_\tau}d\tau,\qquad \tilde\nu_{s,t} = 2 \dint_s^t \frac{c_t^2}{c_\tau^2}\big( \dot\theta_\tau\gamma_\tau + k_\tau \gamma_\tau^2)d\tau,\qquad e^{\mu_{s,t}} = \frac{c_t}{c_s},
\]
to obtain that the solution to $\mT f=0$ with initial condition $f(s,\cdot)$ is given by
\begin{equation}\label{eq:mTstot}
    f(t,\xione) =  \exp\left( i \lambda_{s,t} + i\frac{\tilde \nu_{s,t}}{2} \xione^2 \right)  e^{\mu_{s,t}/2} f\left(s, e^{\mu_{s,t}}\xione \right).
\end{equation}
By an application of the Duhamel principle, the solution to $\mT f=g$ for $0\leq t\leq T$ with $f(0,\xione)=0$ is thus given explicitly by
\begin{equation}\label{eq:solmTsource}
  f(t,\xione)=i\dint_0^t \exp\Big(i\lambda_{s,t}+i\frac12 \tilde \nu_{s,t} \xione^2\Big) e^{\mu_{s,t}/2} g(s, e^{\mu_{s,t}}\xione) ds.
\end{equation}

Consider the solution in \eqref{eq:mTstot} at a fixed time $s$ with the explicit time-dependence $\lambda(t)=\lambda_{s,t}$, $e^{\mu(t)}=e^{\mu_{s,t}}$ and $\nu(t)=\tilde \nu_{s,t}$ to simplify notation. We prove the lemma for the operator
\begin{align}
    f(t,\xi)  \to  e^{i\lambda(t)} e^{\frac{\mu(t)}2} e^{i\frac12 \nu(t)\xi^2} f(t,e^{\mu(t)}\xi).
\end{align}

The term $e^{i\lambda(t)}$ generates a smooth in time modulation. Since $\mu(t)$ is smooth and bounded above and below by positive constants independent of $t$, $e^{\mu(t)/2}$ is also smooth.  The operator $f(t,\xi)\to e^{\frac{\mu(t)}2}e^{i\lambda(t)} f(t,\xi)$ is therefore bounded in the norms $\|\cdot\|_{p\vsa}$.

All other time dependent coefficients are smooth and uniformly bounded in time independent of $\rT$ except for $\nu(t)$ that may grow linearly with $\nu'(t)$ uniformly bounded. Consider the operator $f(t,\xi)\mapsto h(t,\xi)=f(t,e^{\mu(t)}\xi)$ with $\fhat\in \mCS_{p\vsa}$. Plugging $h(t,\xi)$ into the definition of the norms \eqref{eq:normsmCS}, we directly obtain that $\|h\|_{p\vsa}\lesssim \|f\|_{p\vsa}$.

The final transformation (with $\nu(t)$  replaced by $\nu(t)e^{-2\mu(t)}$)
\begin{align}
     f(t,\xi) \mapsto h(t,\xi) = f(t,\xi) e^{i\frac12 \nu(t) \xi^2}
\end{align}
was analyzed in  Lemma \ref{lem:pvsa}. This concludes the analysis of the map from $\fhat(\xi)$ to $f_0(t,\xi)$. 

The volume source term $g(t,\xi)$ is treated similarly by the Duhamel principle, with an additional possible integration in time that provides the extra multiplication by $\aver{\rT}$. This concludes the proof of the lemma.
\end{proof}

We conclude this section with a summary of the operators we introduced to construct approximations of the Dirac equation and some estimates they satisfy.

The terms in $T_0$ and $T_1$ that contribute to the construction of the leading term $a_0(t,\Xitot)$ in the formal expansion $a=\sum_{j\geq0} \eps^{\frac j2} a_j$ were given in Lemma \ref{lem-1d}. Constructing higher-order terms $a_j$ and proving convergence results require estimates on the operators $T_j$, which are constructed as in the proof of Lemma \ref{lem-1d} and given explicitly by
\begin{equation}\label{eq:T1explicit}
  T_1 = D_t + \mA + \tilde T_1,\qquad T_j=\tilde T_j,\quad j\geq2,
\end{equation}
with $\fU_t^*D_t \fU_t=D_t+\mA$ where
\begin{equation} \label{eq:mA}
  \mA:=- \frac 12(\dot\varphi_t\sigma_2+\dot\theta_t(-s_t\sigma_1+c_t\sigma_3)-\dot\theta_t (-D_\xione D_\zetatwo -\zetatwo\xione+\gamma_t(\xione^2-D_\zetatwo^2))) + \dot\gamma_t \xione D_\zetatwo,
\end{equation}
and for $j\geq1$,
\begin{equation}\label{eq:Tjexplicit}
  \tilde T_j =\dsum_{|\alpha|=j+1}\mV_t^* z^\alpha \mV_t\nu_\alpha\cdot\sigma, \qquad \Cm^3\ni 
  \nu_\alpha = \frac{1}{\alpha!} \tilde R_{2,\varphi_t}\tilde R_{3,\theta_t} (R_{-\theta_t}\nabla)^\alpha h(y_t).
\end{equation}
We verify for completeness that $\mV_t^* z^\alpha \mV_t =(\zetatwo-\gamma_t\xione)^{\alpha_2}(-1)^{\alpha_1} (D_\zetatwo+\gamma_tD_\xione)^{\alpha_1}$.
 
We recall that $(j_t,k_t)=(\nu_{111},\nu_{021})$ while $\nu_{201}=0$ was important to obtain an explicit expression for the transport solution since the transport operator involves no term of the form $D^2_{\xione}$. We now summarize the estimates we will be needing.



\begin{lemm}\label{lem:boundsT}
  We may write $T_0^{-1}=T_{01}^{-1}+T_{02}^{-1}$ and $T_1=T_{11}+T_{12}$ such that the following operators are bounded with $\anut-$independent bounds:
  \begin{align}
    \vsa^{-1}T_{01}^{-1}:& \qquad N_t^\perp\cap\mCS_{p+1,\vsa} \to \mCS_{p\vsa},  & 
    T_{02}^{-1}:& \qquad N_t^\perp\cap \mCS_{p\vsa} \to \mCS_{p\vsa}, \\ 
    \vsa T_{11}:& \qquad  \mCS_{p+1,\vsa} \to \mCS_{p\vsa},  & 
    T_{12}:& \qquad  \mCS_{p+2,\vsa} \to \mCS_{p\vsa}, \\
    \vsa^{j+1} T_{j}:& \qquad  \mCS_{p+j+1,\vsa} \to \mCS_{p\vsa} ,\quad j\geq2,  & 
    \aver{T}^{-1} \mT^{-1} :& \qquad  \mCS_{p\vsa} \to \mCS_{p\vsa}.
  \end{align}
\end{lemm}
\begin{proof}
From Lemma \ref{lem:T0} and the decomposition
\begin{align}
    T_{01}^{-1} = Q \matrice{ 0 & 0 \\ 0 & - 2\xi c_t (\rho_t\zeta+\partial_\zeta)^{-1}(\rho_t\zeta-\partial_\zeta)^{-1} } Q ,\quad T_{02}^{-1}= T_0^{-1}-T_{01}^{-1},
\end{align}
and Lemma \ref{lem:pvsa}, we obtain the above first two bounds.

We define $T_1=T_{11}+T_{12}$ with $T_{11}$ the contribution that is linear in $D_\xi$ while $T_{12}$ accounts for the rest (no contribution in $T_1$ is quadratic in $D_\xi$). The above corresponding bounds then follow from Lemma \ref{lem:pvsa} for quadratic expressions in $D_\xi$, $D_\zeta$, $\xi$, and $\zeta$. The same lemma is used to bound $T_j$ for $j\geq2$. The final estimate is a repeat of Lemma \ref{lem:mT}.
\end{proof}

%
%
%
%
%
%
%
%
%
%
%

%
\section{Asymptotic expansion and error estimates} \label{sec:error}
%

We recall that our objective is to construct approximations of solutions $\Psi(t,x)$ of the Dirac equation \eqref{eq:D1}. Our first step was to perform a gauge transformation $\tilde\Psi(t,x)=e^{-i\chi(t,x)/\eps}\Psi(t,x)$ replacing the Dirac operator $\Di$ by $\tilde\Di+R$ in \eqref{eq:D2}. Since $R$ in \eqref{eq:R} is a negligible perturbation to arbitrary order in $\eps$, our second step was to look for wavepackets $\tilde\Psi(t,x)$ in the kernel of $\eps D_t+\tilde\Di$. 

Wavepackets in natural coordinates $\psi(t,z)=\eps^{\frac12} S\tilde\Psi (t,z)$ with the scaling $S$ defined in \eqref{eq:S} then solve $L\psi=0$ in \eqref{eq:L}. Further transformations resulted in the definition of $a(\Xitot)=\fU_t^* \psi(\Xitot)$, where $\fU_t=\bU_t\mV_t$ with $\bU_t$ defined in \eqref{eq:bUt} and $\mV_t$ in \eqref{eq:FS}. The problem $L\psi=0$ is then equivalent to $Ta=0$ for $T=\fU_t^* L \fU_t$. Writing $L=\sum_{j\geq0} \eps^{\frac j2}L_j$ with $L_j$ presented in \eqref{eq:L0}-\eqref{eq:Lj}, we have a corresponding expansion $T=\sum_{j\geq0} \eps^{\frac j2}T_j$ with $T_j=\fU_t^* L_j \fU_t$.

\subsection{Construction of the asymptotic wavepacket.} \label{sec:construction}
Using the notation recalled above, we now construct approximations $a^J=\sum_{j=0}^J \eps^{\frac j2} a_j$ of $a(t,\Xitot)$ solution of $Ta=0$. Plugging the expansion for $a^J$ in the equation $Ta^J=0$, using $T=\sum_{j\geq0}\eps^{\frac j2} T_j$, and equating like powers of $\eps$ gives the sequence of equations \eqref{eq:seqT}, which we recall here:
\begin{equation}\label{eq:seqT2}
  \sum_{k=0}^j T_k a_{j-k}=0,\qquad 0\leq j\leq J.
\end{equation}
We solve these equations in turn. 

The leading order equation $T_0a_0=0$ combined with Lemma \ref{lem:T0} shows that $a_0(t,\Xitot)=f_0(t,\xione)\phi_t(\zetatwo)$ with $\phi_t$ defined in \eqref{eq:phi}. The compatibility condition for the next equation $T_0a_1=-T_1a_0$ implies $\mT f_0=0$ with the transport operator $\mT$ defined in Lemma \ref{eq:transport}. For a fixed initial condition $f_0(0,\xione)=\fhat(\xione)$, Lemma \ref{lem:mT} provides a unique solution $f_0(t,\xione)$ of $\mT f_0=0$ and a leading term $a_0(t,\Xitot)=f_0(t,\xione)\phi_t(\zetatwo)$.

Consider next the construction of $a_1$, solution to $T_0a_1+T_1a_0=0$. It is given according to Lemma \ref{lem:T0} by
\[
  a_1= -T_0^{-1} T_1 a_0 + f_1 \phi_t
\]
with $f_1$ arbitrary at this stage. We will consider regularity properties in a lemma below.


We now extend the construction to higher-order approximations.
%
%
%
%
Let $j\geq2$ and assume $a_{k}$ for $0\leq k\leq j-1$ constructed except for $f_{j-1}$.  In order to define $a_j$, we impose the compatibility condition
\[
   \Big( \sum_{k=1}^j T_k a_{j-k},\phi_t\Big)_2=0\quad\mbox{ or equivalently }\quad \mT f_{j-1} = g_{j-1},\quad g_{j-1}=- \Big(\sum_{k=2}^j T_k a_{j-k},\phi_t\Big)_2.
\]
We recall that $(\cdot,\cdot)_2$ is the standard inner product in $L^2(\Rm,\Cm^2)$ in the $\zetatwo$ variable. By lemma \ref{lem:mT}, this is a well-posed transport equation for $f_{j-1}$. Then, by lemma \ref{lem:T0},
\[
  a_j= T_0^{-1} \big( -\sum_{k=1}^j T_k a_{j-k} \big)+ f_j \phi_t. 
\]  
This constructs $a_j$ for $0\leq j\leq J$ iteratively and we set $f_J=0$ for concreteness. 
This concludes the construction of the approximation
\begin{align}
    a^J=\sum_{j=0}^J\eps^{\frac j2}a_j
\end{align}
of formal order $\eps^{\frac{J+1}2}$ of $a$ solution of $Ta=0$. 

Note that the initial conditions for $a_j(0,\Xitot)$ for $j\geq1$ are defined implicitly by the above construction. Only $\fhat(\xione)$ in the initial condition $a_0(0,\Xitot)=\fhat(\xione)\phi_t(\zetatwo)$ is prescribed. Our construction aims to propagate wavepackets that belong to an appropriate (non-dispersive) branch of continuous spectrum of $\tilde\Di$. The initial condition for $a^J$ ensures that the latter belongs to that branch with sufficient accuracy.

The terms of the above expansion satisfy the following estimate: 
\begin{lemm}\label{lem:regaj}
Let $a_j$ be constructed as above for $0\leq j\leq J$ (with $f_J=0$) and for $p\in\Nm$, let $\fhat\in \mS_{p+3J}$. Then we have the estimates for $0\leq j\leq J-1$,
\begin{align}\label{boundaj}
     \|a_j\|_p  \leq C_{pJ} \big(\aver{\rT} \anut^3\big)^{j}  \anut^{p}  \|\fhat\|_{p+3j},\qquad  \|a_J\|_p \leq C_{pJ}  \big(\aver{\rT} \anut^3\big)^{J-1}  \anut^{p+1}  \|\fhat\|_{p+3J},
\end{align}
with constants $C_{pJ}$ independent of $\rT$ and $\anut$.
\end{lemm}
  %
\begin{proof}
Consider $0<j<J$. The estimate for $a_J$ is a bit different since $f_J=0$.
  
We first observe that by construction of the terms $a_j$ for $1\leq j\leq J-1$,
\begin{align}\label{eq:expaj}
    a_j=-T_0^{-1}\sum_{k=0}^{j-1} T_{j-k} a_{k} - \phi_t(\zeta) \mT^{-1} \sum_{k=0}^{j-1} (T_{j+1-k}a_{k},\phi_t)_2.
\end{align}
We wish to prove by induction that 
\begin{align}\label{eq:ajvsa}
    \|a_j\|_{p\vsa} \lesssim (\aver{T} \vsa^{-3})^j \|a_0\|_{p+3j,\vsa}, \qquad \|a_J\|_{p\vsa} \lesssim (\aver{T} \vsa^{-3})^{J-1} ( \vsa\|a_0\|_{p+3J,\vsa} + \|a_0\|_{p+3J-1,\vsa}).
\end{align}
Using that $\|a_0\|_{p\vsa}= \|f_0\phi_t\|_{p\vsa}\lesssim \|\fhat\|_{p\vsa}$ thanks to Lemma \ref{lem:mT}, then $\vsa^{p}\|f\|_p \lesssim  \|f\|_{p\vsa} \lesssim  \vsa^{-p} \|f\|_{p}$ and $\vsa=\anut^{-\frac12}$ thus provide the results stated in the lemma. It remains to verify \eqref{eq:ajvsa}.

From \eqref{eq:expaj} for $j=1$, we obtain $a_1=-T_0^{-1}T_1 a_0 - \phi_t \mT^{-1} (T_2a_0,\phi_t)_2$ so that using the result of Lemma \ref{lem:boundsT}, we find the second term to be the least regular and 
\begin{align}
    \|a_1\|_{p\vsa} \lesssim \aver{T} \vsa^{-3} \|a_0\|_{p\vsa+3}
\end{align}
so that \eqref{eq:ajvsa} holds when $j=1$. Assume it holds for $j-1\geq1$. Then using \eqref{eq:expaj} and Lemma \ref{lem:boundsT}, we find
\begin{align}
     \|a_j\|_{p\vsa} \lesssim \aver{T} \dsum_{k=0}^{j-1}
 \vsa^{-(j+2-k)} \|a_k\|_{p+j+2-k,\vsa}
 \lesssim \aver{T} \dsum_{k=0}^{j-1}
 \vsa^{-(j+2-k)} (\aver{T}\vsa^{-3})^k \|a_0\|_{p+2k+j+2,\vsa},
\end{align}
which is largest when $k=j-1$ and provides the sought estimate when $0\leq j<J$. It remains to consider the term
\begin{align}
    a_J= - T_0^{-1} \sum_{k=0}^{J-1} T_{J-k} a_k = - T_0^{-1} \Pi_t \sum_{k=0}^{J-1} T_{J-k} a_k = -T_0^{-1}  \Pi_t T_1 a_{J-1}  - T_0^{-1} \Pi_t \sum_{k=0}^{J-2} T_{J-k} a_k, 
\end{align}
by construction of the wavepackets, where $\Pi_t$ projects onto $N_t^\perp$. This operator has a smooth Schwartz kernel in $(t,\zeta,\zeta')$ that is independent of $\xi$. Therefore $\Pi_t$ is bounded from $\mCS_{p\vsa}$ to itself with $\vsa-$independent bound. 

Since $T_0^{-1}$ is bounded from $N^\perp\cap \mCS_{p+1,\vsa}$ to $\mCS_{p\vsa}$, we find for $0\leq k\leq J-2$
\begin{align}
    \|T_0^{-1} \Pi_t T_{J-k}a_k \|_{p\vsa} \lesssim \|a_k \|_{p+2+J-k,\vsa} \lesssim (\aver{\rT}\vsa^{-3})^k \|a_0\|_{p+2+J+2k,\vsa} \lesssim (\aver{\rT}\vsa^{-3})^{J-1} \|a_0\|_{p+3J-1,\vsa}
\end{align}
since $2+J+2k\leq 3J-1$.

It remains to consider $T_0^{-1} \Pi_t T_1 a_{J-1}$.
We decompose $T_0^{-1}\Pi_tT_1=A+B$ with $A=T_{01}^{-1}\Pi_tT_{11}+T_{02}^{-1}\Pi_tT_{12}$ with bound $\aver{T}$ from $\mCS_{p+2,\vsa} \to  \mCS_{p\vsa}$ and $B=T_{01}^{-1} \Pi_tT_{12} + T_{02}^{-1}\Pi_tT_{11}$ with bound $\aver{T} \vsa$ from $\mCS_{p+3,\vsa} \to \mCS_{p\vsa}$ as per Lemma \ref{lem:boundsT}. This implies
\begin{align}
    \|a_J\|_{p\vsa} \lesssim (\aver{T} \vsa^{-3})^{J-1} ( \vsa\|a_0\|_{p+3J,\vsa} + \|a_0\|_{p+3J-1,\vsa})
\end{align}
as was to be shown. Thus \eqref{eq:ajvsa} holds and this concludes the proof of the lemma.
\end{proof}

\subsection{Main approximation result} \label{sec:erroranalysis}
%
We can now state our main result. 

\begin{theo}\label{thm:main}
  Let $\Psi$ be the solution on $[0,\rT]\times\Rm^2$ of $(\eps D_t+\Di)\Psi=0$ with initial condition $\Psi(0,x)=\Psi^J(0,x)$ where $\Psi^J(t,x)= e^{\frac i\eps \chi(t,x)} (\eps^{-\frac12}S^{-1} \fU_t a^J) (t,x)$ with $a^J$ constructed in section \ref{sec:construction} based on an initial condition $\fhat(\xi)\in \mS_{3J+2}(\Rm,\Cm)$. 
  
  Then there exists $C_J> 0$ independent of $\rT$ such that for every $\epsi \in (0,1]$,
  \begin{align} \label{eq:mainerror}
      \sup_{t\in [0,\rT]} \|\Psi-\Psi^J\|_{L^2(\Rm^2,\Cm^2)} \leq C_J \left(\aver{\rT} \anut^3\right)^J \eps^{\frac{J}2}.
  \end{align}
\end{theo}


We recall that $\anut$ is defined in \eqref{eq:anut} and $\aver{\rT}=1+\rT$.  In particular, the approximation error is controlled so long as $\eps^{\frac12} \aver{T} \anut^3 \ll1$. When $\anut$ is uniformly bounded, errors are controlled up to times $\rT\ll\eps^{-\frac12}$ as in \cite{bal2021edge}. However, when $\anut$ is of order $\rT$, then errors are controlled up to times $\rT\ll\eps^{-\frac18}$.

Note that such results are qualitatively reasonable: in the presence of dispersion, it becomes more difficult to control spatial moments (necessitated by the Taylor expansion of the coefficients $h$) of the solution.  
 
The rest of the section is devoted to the proof of this theorem. It is based on a similar approximation result in the local variables $z=(z_1,z_2)$, which we state as a result of independent interest.
\begin{prop}\label{prop:error}
Let $\psi$ be the solution on $[0,\rT]\times\Rm^2$ of $L\psi=0$ with initial condition $\psi(0,z)=\psi^J(0,z)$ where $\psi^J(t,z)=  \fU_t a^J (t,z)$ with $a^J$ constructed in section \ref{sec:construction} based on an initial condition $\fhat(\xi)\in \mS_{3J+2}(\Rm,\Cm)$. 
  
Then there exists $C_J> 0$ independent of $\rT$ such that for every $\epsi \in (0,1]$,
\begin{align} \label{eq:mainerror}
      \sup_{t\in [0,\rT]} \|\psi-\psi^J\|_{L^2(\Rm^2,\Cm^2)} \leq C_J \left(\aver{\rT} \anut^3\right)^J \eps^{\frac{J}2}.
\end{align}
\end{prop}

\begin{proof}
Solutions of the local problems $L\psi=0$ and $Ta=0$ are equivalent via the relation $\psi=\fU_t a$.
We now show that $\psi^J=\sum_{j=0}^J\eps^{\frac j2}\psi_j$ with $\psi_j=\fU_t a_j$ approximately solves $L\psi=0$ when $J\geq1$. Define for $J\geq1$ the remainder operator 
\begin{align}\label{eq:Lgeq}
    L_{\geq J} = \eps^{-\frac J2}\Big(L-\dsum_{j=0}^{J-1}\eps^{\frac j2}L_j\Big)
\end{align}
and the reduced operator $\tilde L_{\geq J}=L_{\geq J}-\delta_{1J}D_t$ involving only the coefficients $h$. 

We deduce from \eqref{eq:seqT2} that $\sum_{k=0}^j L_k \psi_{j-k}=0$ for $0\leq j\leq J$ and hence 
\[
 L\psi^J = \dsum_{j=0}^J \eps^{\frac j2} L\psi_j=\dsum_{j=0}^J \Big(\dsum_{k=0}^{J-j}\eps^{\frac{j+k}2} L_k \psi_j+\eps^{\frac{J+1}2} L_{\geq J+1-j}\psi_j \Big) = \eps^{\frac{J+1}2}\dsum_{j=0}^J L_{\geq J+1-j}\psi_j,
\]
since
\[
  \dsum_{j=0}^J\dsum_{k=0}^{J-j}\eps^{\frac{j+k}2} L_k\psi_j = \dsum_{j=0}^J\dsum_{l=j}^J \eps^{\frac l2}L_{l-j}\psi_j = \dsum_{l=0}^J\eps^{\frac l2}\dsum_{j=0}^l L_{l-j}\psi_j=0.
\]
We thus obtain
\begin{equation}\label{eq:rJ}
  L\psi^J = \eps^{\frac{J+1}2} \Big( D_t\psi_J + \dsum_{j=0}^J \tilde L_{\geq J+1-j} \psi_k\Big) =: r_J.
\end{equation}
We now derive a uniform bound in time of order $\eps^{\frac{J+1}2}$ in $L^2(\Rm^2,\Cm^2)$ for $r_J$.  

We observe that $\mV_t \mS_p(\Rm^d;\Cm^q) \mV_t^* \cong  \mS_p(\Rm^d;\Cm^q)$ as spaces of functions (of $z$ for $d=2$). In other words, $\|g\|_p \cong \|\mV_t^* g\|_p$ define equivalent norms for each $p$. The reason is that the above spaces based on Hermite functions are invariant by conjugation by Fourier transforms as well as invariant under invertible linear transforms $(\xione,z_2)\mapsto (\xione,z_2+\gamma_t\xione)$ of the base variables (uniformly in time since $\gamma_t$ is bounded). Since $\gamma_t$ is smooth, we also obtain the equivalence of the norms of $\mV_t \mCS_{p}(\Rm^d;\Cm^q) \mV_t^* \cong  \mCS_{p}(\Rm^d;\Cm^q)$. Note that there is no meaningful notion of anisotropic space $\mCS_{p\vsa}$ in the variable $z$.

The regularity results of lemma \ref{lem:regaj} therefore apply to $\psi_j=\fU_t a_j$ since conjugation by $\bU_t$ also preserves norms as the rotation angles $\theta_t$ and $\varphi_t$ are smooth in $t$. As a consequence,
\begin{align}\label{eq:boundpsij}
    \|\psi_J\|_{2} &\lesssim \aver{T}^{J-1}\anut^{3J} \|\fhat\|_{3J+2},\\
    \|\psi_j\|_{J-j+2} &\lesssim  (\aver{T}\anut^{3})^j 
    \anut^{J-j+2} \|\fhat\|_{J+2j+2} \lesssim \aver{T}^{J-1}\anut^{3J} \|\fhat\|_{3J+2}
\end{align}
for $0\leq j\leq J-1$. 

Since $D_t$ is bounded from $\mCS_{p+2}$ to $\mCS_{p}$, we obtain that $\|D_t\psi_J\|_{0}\lesssim \aver{T}^{J-1} \|\fhat\|_{3J+2}$.


Using Taylor expansions of the coefficients for $1\leq j\leq J$,
\begin{align}
    h(y_t+\sqrt\eps z)=\sum_{|\alpha|\leq j} \frac{1}{\alpha!} \partial^\alpha h(y_t) \eps^{\frac{|\alpha|}2} z^\alpha + \eps^{\frac{j+1}2}\sum_{|\beta|=j+1}z^\beta R_\beta(y_t+\sqrt \eps z),
\end{align}
with smooth functions $R_\beta$ by regularity assumptions on $\kappa$ and $A$, we observe that $\tilde L_{\geq j}$ has the same regularity properties as $\tilde L_j$. In particular, it maps $\mCS_{p+j+1}$ to $\mCS_{p}$. Therefore,  for $0\leq j\leq J$, $\|\tilde L_{\geq J+1-j} \psi_j\|_0\lesssim \|\psi_j\|_{J+2-j}\lesssim \aver{T}^{J-1}\anut^{3J} \|\fhat\|_{3J+2}$ thanks to \eqref{eq:boundpsij}.

By definition \eqref{eq:rJ}, this shows that
\begin{align}\label{eq:estrJ}
    \|r_J\|_0\lesssim \eps^{\frac{J+1}2} \aver{T}^{J-1} \anut^{3J} \|\fhat\|_{3J+2}.
\end{align}


Define $\tilde L=L-\eps^{\frac 12} D_t$, which we verify is self adjoint \cite{TH}. We thus observe that $L(\psi^J-\psi)=r^J$ is equivalent to 
\[
  (D_t + \eps^{-\frac 12} \tilde L)(\psi^J-\psi) = \eps^{-\frac 12} r_J.
\]
By unitarity of the above Dirac operator \cite{TH}, we obtain an error on $\|\psi^J-\psi\|_0$ of order $\aver{T}\eps^{-\frac 12} \|r_J\|_0\lesssim \eps^{\frac J2} \aver{T}^J \anut^{3J}$. This concludes the proof of the proposition.
\end{proof}

\begin{proof} (Theorem \ref{thm:main}) 
Let $\tilde\Psi$ be the solution of $(\eps D_t+\tilde\Di)\tilde\Psi=0$ with initial conditions $\tilde\Psi(0,x)=\tilde\Psi^J(t,x)=(\eps^{-\frac12}S^{-1} \fU_t a^J)(t,x)$. Since $\tilde\Psi=\eps^{-\frac12}S^{-1}\psi$ and $\tilde\Psi^J=\eps^{-\frac12}S^{-1}\psi^J$ for $\eps^{-\frac12}S^{-1}$ an $L^2-$isometry, we directly deduce from Proposition \ref{prop:error} that 
\begin{align} \label{eq:mainerror1}
      \sup_{t\in [0,\rT]} \|\tilde \Psi-\tilde \Psi^J\|_{L^2(\Rm^2,\Cm^2)} \leq C_J \left(\aver{\rT} \anut^3\right)^J \eps^{\frac{J}2}.
\end{align}



By construction, $\Psi^J(t,x)=e^{\frac i\eps\chi(t,x)}\tilde\Psi^J(t,x) =e^{\frac i\eps\chi(t,x)} \eps^{-\frac12}S^{-1}\psi^J(t,x)$ so that 
\begin{equation}
    (\eps D_t + \Di)\Psi^J = e^{\frac i\eps\chi(t,x)} (\eps D_t + \tilde \Di + R) \tilde \Psi^J = e^{\frac i\eps\chi(t,x)} \eps^{-\frac12} S^{-1} [ \eps^{\frac12} r_J + (\eps^{\frac12}SR) \psi^J],
\end{equation}
where we used $\eps D_t + \tilde\Di=S^{-1}L \eps^{\frac12} S$ and where $r_J$ defined in \eqref{eq:rJ} and estimated in \eqref{eq:estrJ}.

By unitarity for $(\eps D_t + \Di)$, the statement of the theorem follows if we show that $SR \psi^J(t,z)=R(t,y_t+\sqrt\eps z)\psi^J(t,z)$ satisfy a bound similar to $r_J$ as in \eqref{eq:estrJ}.

Since $R$ vanishes on $\Omega_\delta$ and is uniformly bounded on $\Rm\times\Rm^2$ by construction, we find that for every $p\geq0$,
\begin{align}
    \|(1+|z|^2)^\frac p2 R(t,y_t+\sqrt\eps z)\|_\infty \lesssim \eps^{\frac p2}.
\end{align}
Therefore, thanks to \eqref{eq:boundpsij},
\begin{align}
    \|R(t,y_t+\sqrt\eps z) \psi^J\|_0 & \lesssim \sum_{j=0}^J \eps^{\frac j2} \|R(t,y_t+\sqrt\eps z) \psi_j\|_0 \lesssim \sum_{j=0}^J \eps^{\frac{J+1}2} \|(1+|z|^2)^\frac {J-j+1}2 \psi_j\|_0 
   \\
 &  \lesssim \  \eps^{\frac{J+1}2} \sum_{j=0}^J \|\psi_j\|_{J-j+1} \lesssim \eps^{\frac {J+1}2} \aver{T}^{J-1} \anut^{3J} \|\fhat\|_{3J+2}.
\end{align}
Thus, $SR\psi^J$ satisfies the same estimate \eqref{eq:estrJ} as $r_J$. Since $\Di$ is self-adjoint, we conclude the proof of Theorem \ref{thm:main} by the same unitarity principle as in Proposition \ref{prop:error}.
\end{proof}

\begin{proof} (Theorem \ref{thm:5})
The above result with $J=1$ also provides a proof of Theorem \ref{thm:5} when the gauge transformation is based on $\chi(t,x)$. Indeed, in this simplified setting with $B$ constant and $\Delta\kappa=0$, we observe that $a_0(t,\Xitot)$ is given by \eqref{eq-00m} with $\lambda_t=0$, $e^{\mu_t}=1$, and $\tilde \nu_t=\nu_t=2\gamma(\theta_t-\theta_0)$ with $\gamma=\frac{B}{1+B^2}$. This provides the expression for the kernel $g_t(z)$ in \eqref{eq:00j}, which quantifies the dispersive effects while $\psi_{\theta_t,B}(z)$ captures all other effects in $\fU_t a_0(z)$ as one readily verifies. 

An additional error in $\Psi^1(t,x)=e^{\frac i\eps \chi(t,x)}\eps^{-\frac12}S^{-1} \psi^1 (t,x)$ comes from replacing $\chi(t,x)$ by its quadratic expansion $\chi_2(t,x)$ (called $\chi(t,x)$ in Theorem \ref{thm:5}) in \eqref{eq:00ii}. In the $z$ variables, we find
\[
  e^{i\frac 1\eps \chi(y_t+\sqrt\eps z)} = e^{i\frac1\eps \chi_2(y_t+\sqrt\eps z)} + O(\eps^{\frac 12}) |z|^3.
\]
This term is multiplied by $\psi^1=\psi_0+\sqrt\eps\psi_1$ and yields an error of order $\eps^{\frac12}$ in the $L^2-$sense for $(e^{\frac i\eps \chi(t,x)} -e^{\frac i\eps \chi_2(t,x)})\eps^{-\frac12}S^{-1} \psi^1 (t,x)$ when $\psi_1\in \mCS_0$ and when $|z|^3\psi_0\in \mCS_0$ as well. Both bounds hold as soon as $\psi_0\in \mCS_3$, and hence when $\fhat\in \mS_5$ as required for $J=1$ in Theorem \ref{thm:main} with a bound
\[
  \| [e^{i\frac 1\eps \chi(y_t+\sqrt\eps z)} - e^{i\frac1\eps \chi_2(y_t+\sqrt\eps z)}]\psi^1  \|_{\mCS_0}\leq C \eps^{\frac12}.
\]
We therefore obtain an overall bound on $\Psi-e^{\frac i\eps \chi_2(t,x)}\eps^{-\frac12}S^{-1} \psi^1 (t,x)$ growing as $C\eps^{\frac12}$ on an interval $[0,\rT]$ with $\rT$ fixed. Since $\sqrt\eps \psi_1$ is also of order $\eps^{\frac12}$, then so is  $\Psi-e^{\frac i\eps \chi_2(t,x)}\eps^{-\frac12}S^{-1} \psi_0 (t,x)$.
The construction performed on $[0,\rT]$ clearly holds on $[-\rT,0]$ as well. This completes the proof of Theorem \ref{thm:5}.
\end{proof}
\section{Numerical simulations}\label{sec:num}



In this section, we illustrate numerically the main effects of the magnetic field and the curved interface on the propagating wavepackets: (i) slowdown; (ii) Aharonov--Bohm phase-shift;  (iii) dispersion.

The leading term in the asymptotic expansion provided by Theorem \ref{thm:main} is given by
\begin{equation}\label{eq:Psi0}
  \Psi_0(t,x) = \eps^{-\frac12} e^{\frac i\eps \chi(t,x)} (\fU_t a_0) \Big(t,\frac{x-y_t}{\sqrt\eps}\Big), \qquad \text{where:}
\end{equation}
\begin{itemize}
    \item $\fU_t a_0(t,z) = \bU_t \mV_t a_0 (t,z)$ with $\bU_t$ a spatial and spinorial rotation that does not quantitatively affect the amplitude landscape of the wavepacket; 
    \item $\mV_t$ is a Fourier-like transform given explicitly in \eqref{eq:FS}.  
    \item $a_0(t,\Xitot) = f_0(t,\xi) \phi_t(\zeta)$ with $f_0(t,\xi)$ given respectively by \eqref{eq-1n} and \eqref{eq:phi}.
\end{itemize}

For concreteness, we work below with the Gaussian initial condition $f_0(0,\xi)=e^{-\frac12\sigma \xione^2}$ for $\sigma>0$. 
Thanks to \eqref{eq-1n} and \eqref{eq:Qt}, we obtain
\begin{equation}\label{eq:gaussprofile}
  \mV_t a_0(t,z) = \Big( \frac{\rho_t}{4\pi}\Big) ^{\frac14} e^{\frac12\mu_t}  (Q_\sigma)^{-\frac{1}2} e^{i\lambda_t} e^{-\frac12 \rho_t z_2^2} e^{-\frac12 Q_\sigma^{-1} (z_1+is_t z_2)^2}  \matrice{1\\-1},
\end{equation}
with $Q_\sigma = e^{2\mu_t} \sigma + s_t\gamma_t -ie^{2\mu_t}\nu_t$. 
The various constants that appear in these formulas are collected as follows
\begin{align}
   & \quad r_t=|\nabla\kappa(y_t)|,\quad \rho_t=\sqrt{r_t^2+B_t^2},\quad c_t=\frac{r_t}{\rho_t},\quad s_t=\frac{B_t}{\rho_t},\quad \gamma_t=\frac{B_t}{\rho_t^2},\quad j_t\gamma_t=\frac{d \ln c_t}{dt}, \\
   & k_t= \frac{c_t}2 \Big(\partial_nB(y_t)-B_t \frac{\Delta \kappa(y_t)}{r_t}\Big), \quad  \lambda_t=\dint_0^t \frac{k_s}{2\rho_s} ds, \quad e^{\mu_t}=\frac{c_t}{c_0}, \quad \nu_t = 2 \int_0^t \frac{c_0^2}{c_s^2} (\dot\theta_s\gamma_s + k_s \gamma_s^2) ds,
\end{align}
and $n(y_t)$ and $\tau(y_t)$ are the normal and tangent vectors to $\Gamma$ at $y_t$. 

The rate of change of dispersion $\dot\nu_t$ is, up to the multiplicative constant $2\gamma_te^{-2\mu_t}$ given by the two contributions $\dot\theta_t+\gamma_tk_t$, which we write explicitly as
\begin{equation} \label{eq:ratedisp}
  \dot\theta_t + \gamma_t k_t = c_t \Big( \epsilon_tK_t + \frac12 \frac{B\partial_nB-B^2r_t^{-1}\Delta\kappa}{B^2+r_t^2}(y_t) \Big),
\end{equation}
with $\epsilon_t=\pm1$ when $\{\mp\kappa>0\}$ is convex in a neighborhood of $y_t$. This provides an expression to assess how the geometry of $\kappa$ and that of $B$ combine to amplify or suppress dispersion.

\begin{figure}[htbp]
    \centering
    \includegraphics[width=7cm]{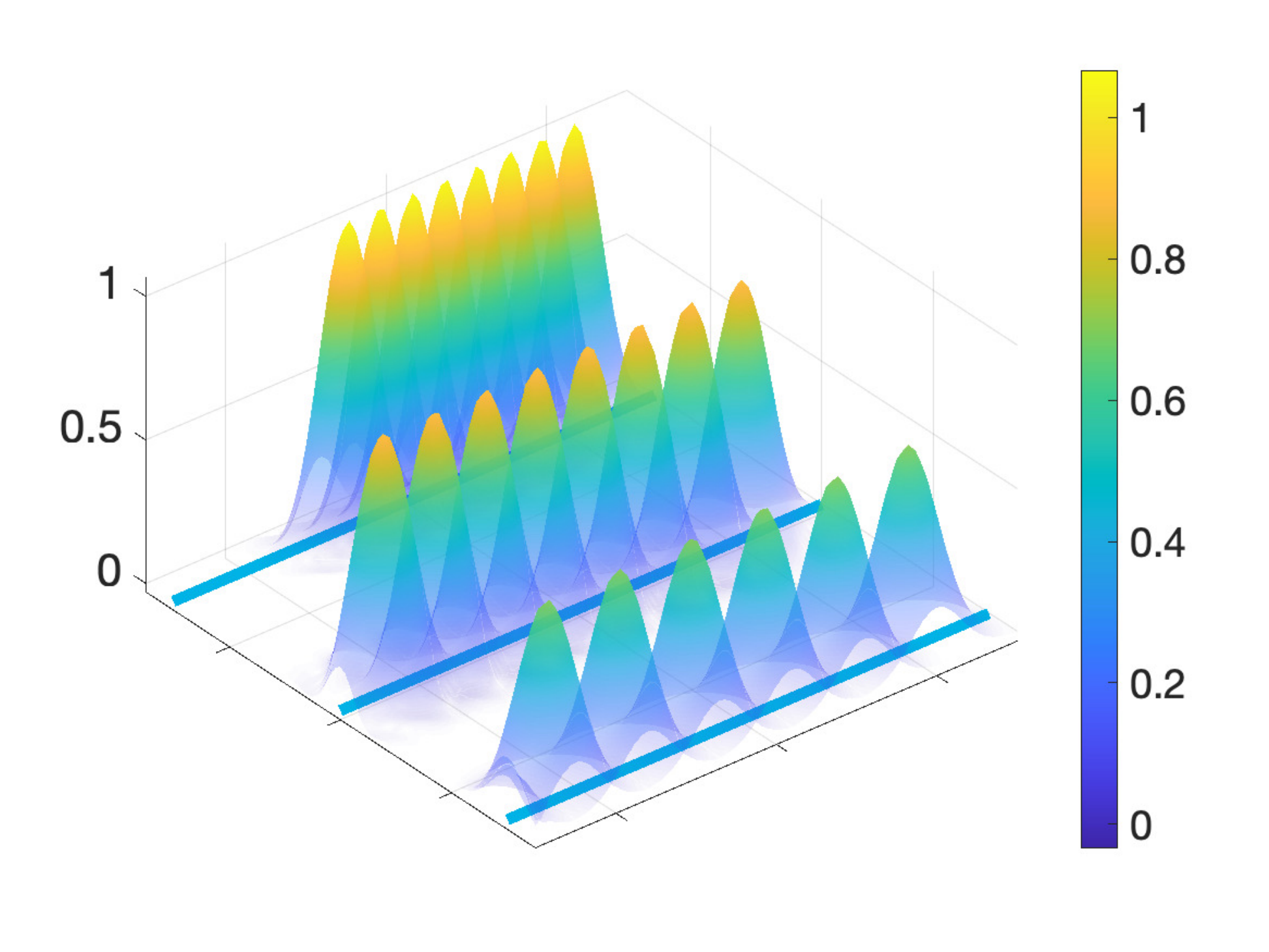} 
    \includegraphics[width=7cm]{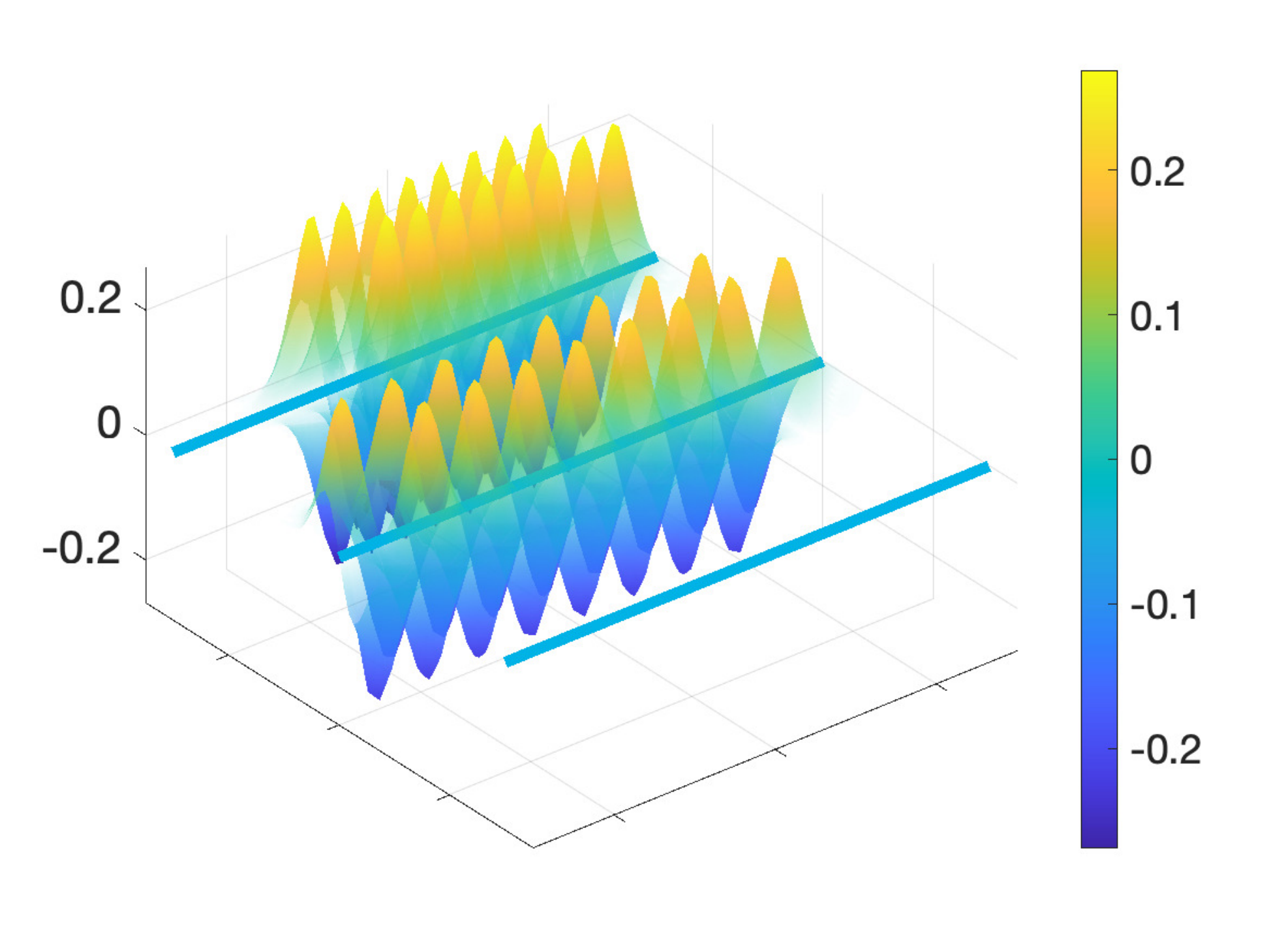}\\
    \includegraphics[width=7cm]{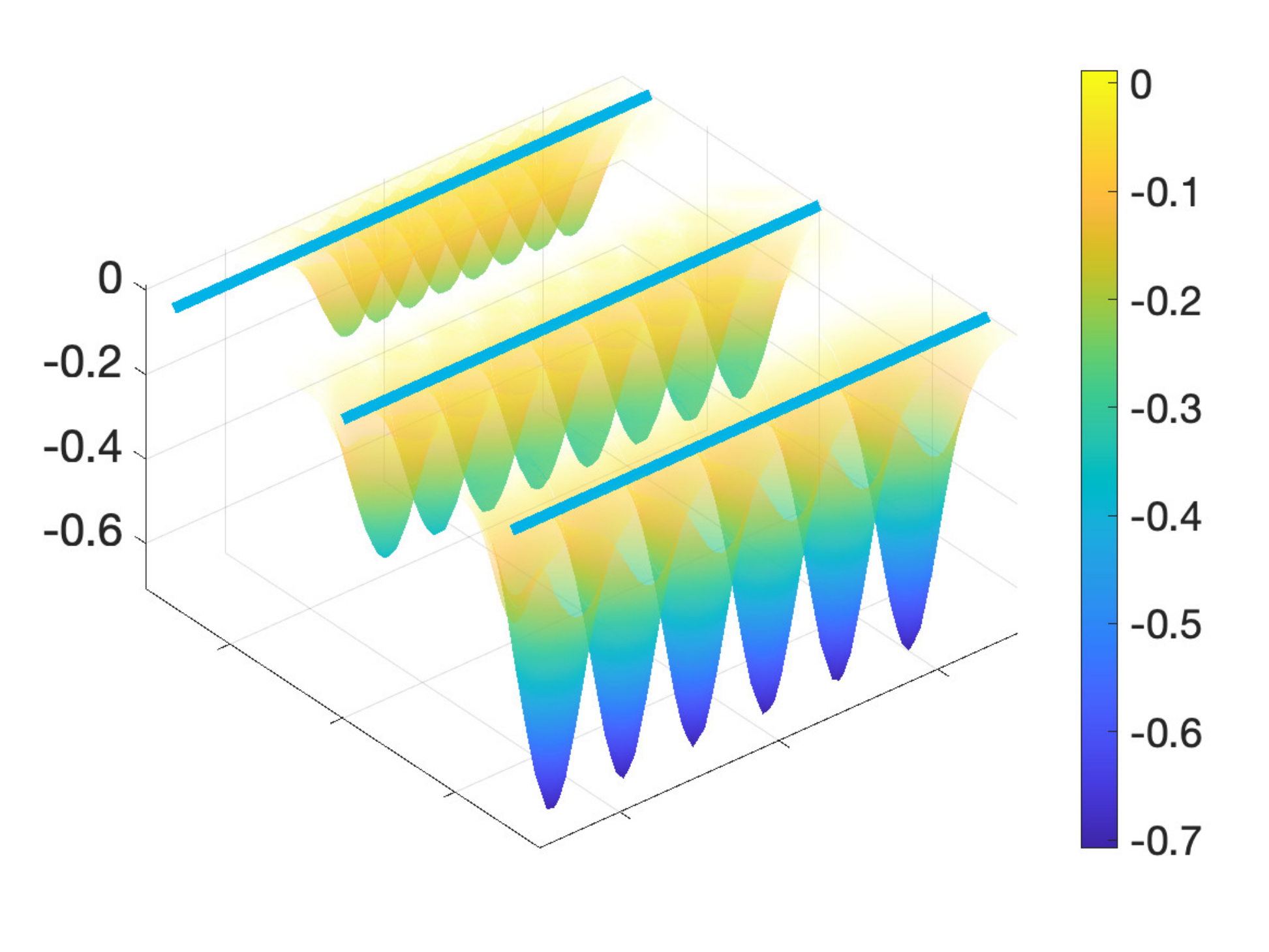} 
    \includegraphics[width=7cm]{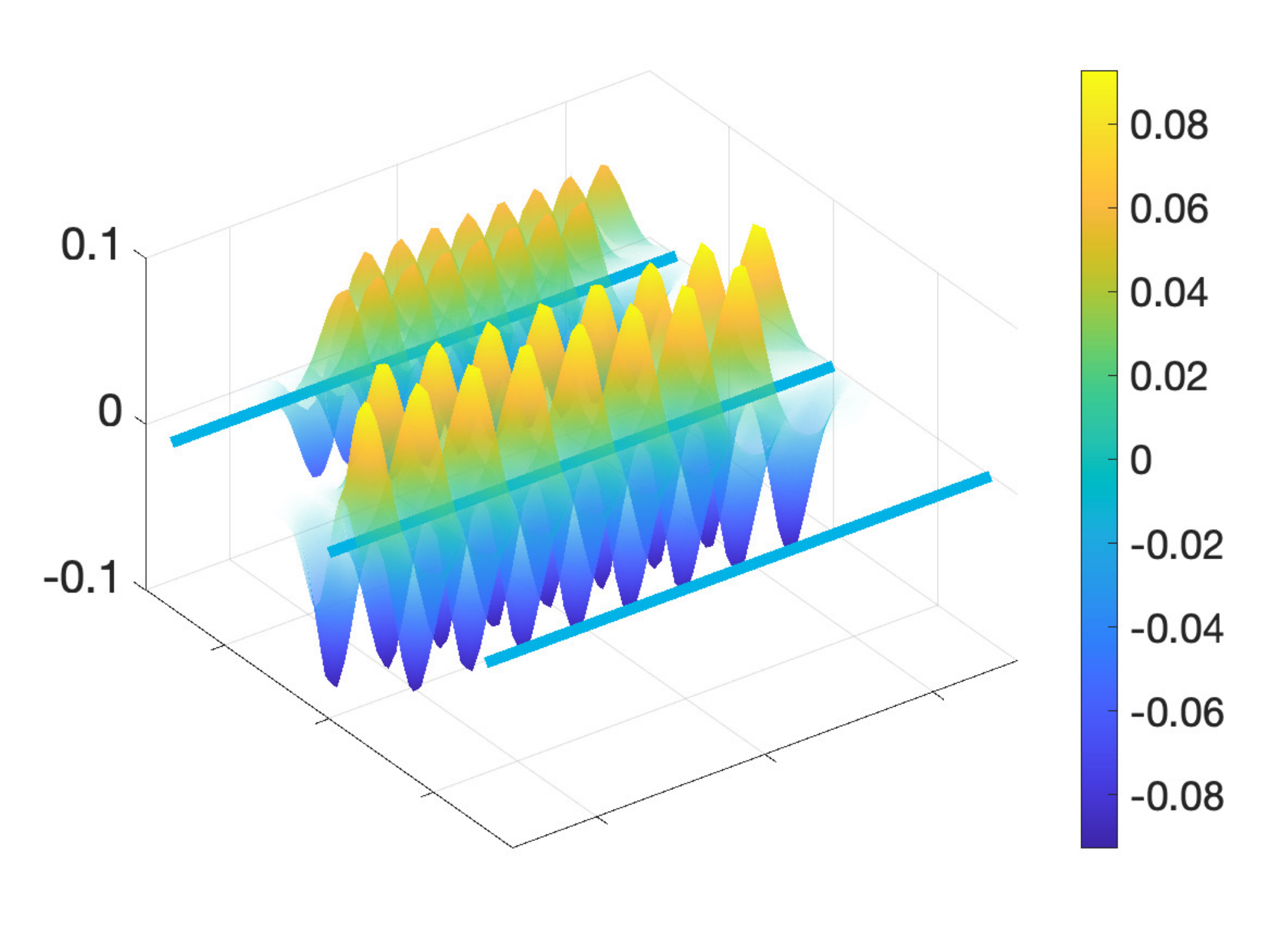}
        \caption{$\epsi = 0.05$. Snapshots of real part (left figures) and imaginary part (right figures) of first (top figures) and second (bottom figures) component of the wavepacket with straight interface $\kappa(x)=x_2$. The wavepacket is propagating towards the reader, for constant magnetic fields $B=1.5,0.75,0$ (from left to right in each individual figure). 
        } \label{fig:slowdown}
\end{figure}

\subsection{Magnetic slowdown in constant B-field}\label{sec:slowdown} The most visible impact of the magnetic field is the slowdown of the wavepacket: it propagates at speed $c_t = (1+B_t^2)^{-\frac12}$, which is strictly less than $1$ whenever $B_t \neq 0$. This is confirmed by the results shown in Figure \ref{fig:slowdown}.

In this simulation, a constant magnetic field for a straight interface given by $\kappa(x)=x_2$ is modeled by $\tilde A=A=-Bx_2e_1=Bx_2\tau$ (with therefore $\chi=0$ and $\beta=B$). This implies $j_t=k_t=0$ so that $e^{\mu_t}=1$ and $\nu_t=0$. The terms $\rho=\sqrt{1+B^2}$ and $\gamma=\frac{B}{1+B^2}$ are constant. The wavepacket velocity is constant and given by
\[
  \dot y_t = -\frac{1}{\sqrt{1+B^2}} e_1 =- ce_1.
\]
The wavepacket in \eqref{eq:gaussprofile} takes the simpler form
\[
  \mV_t a_0(t,z) = 
  \Big(\frac{\rho}{4\pi}\Big)^{\frac14} \Big(\frac{1}{\sigma+\rho\gamma^2}\Big)^{\frac12} e^{-\frac12 \left(\rho z_2^2 + \frac{1}{\sigma+\rho\gamma^2}(z_1+isz_2)^2\right)} \matrice{1\\-1}.
\]
The oscillations of the spinor components are clearly visible in Figure \ref{fig:slowdown}.

\subsection{Aharonov--Bohm effect for a circular interface} This effect emerges in the phase $\chi(t,x)/\eps$ in \eqref{eq:Psi0} when $\Gamma$ is a loop. 
Once $y_t$ completes an exact full rotation around $\Gamma$, the wavepacket acquires the phase shift
\begin{align}\label{eq-3c}
    \dfrac{\chi(t,y_t)}{\epsi} = \dfrac{1}{\epsi}\int_0^t \dot{y_s} A(y_s) ds = \dfrac{1}{\epsi}\int_\Gamma A,
\end{align}
see \eqref{eq-2g}. This is a gauge-independent quantity involving the magnetic flux $\Phi = \int_\Gamma A$ enclosed by $\Gamma$, and of order $1/\epsi$.

\begin{figure}[t]
    \centering
     \includegraphics[width=6cm]{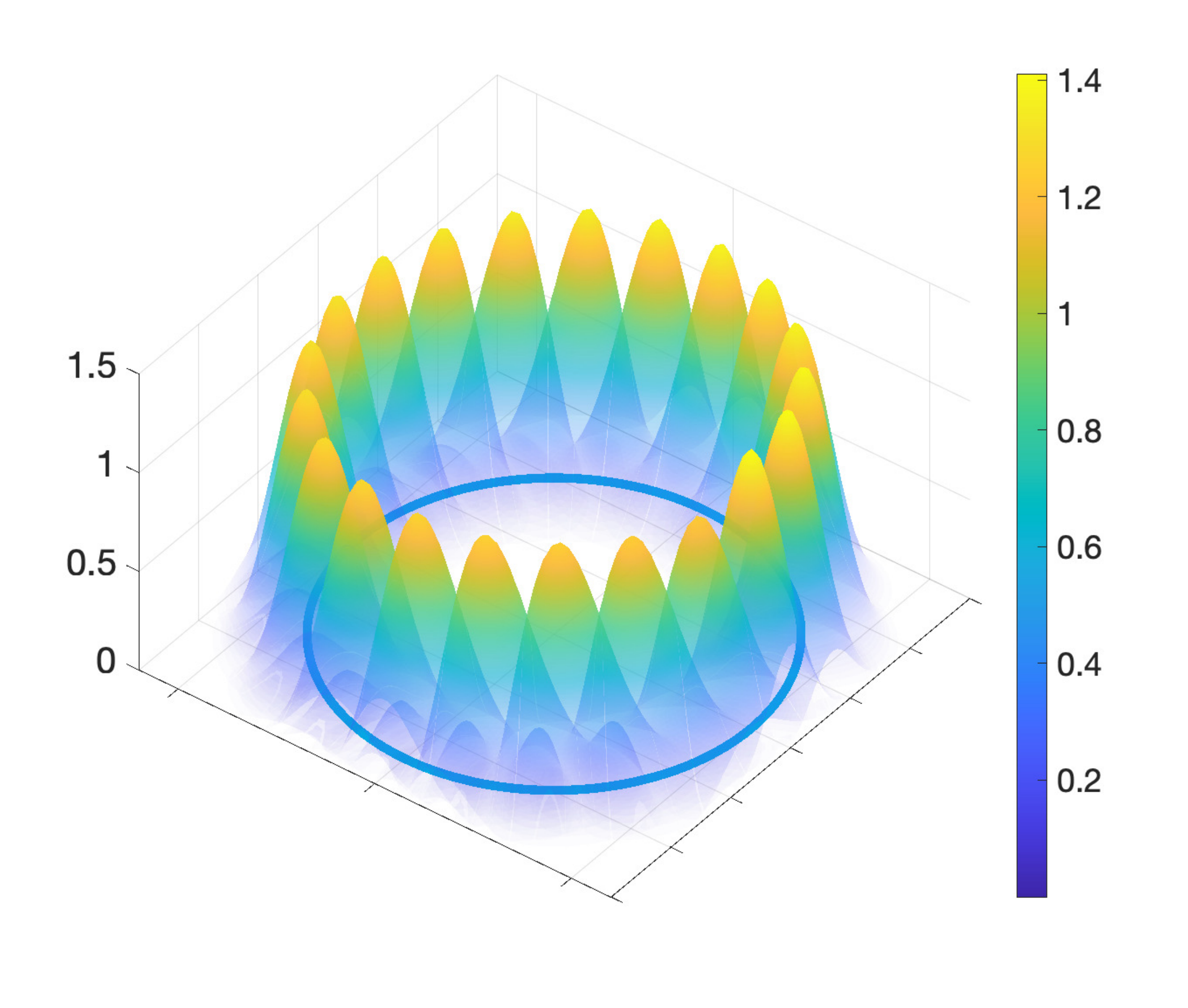}\includegraphics[width=6cm]{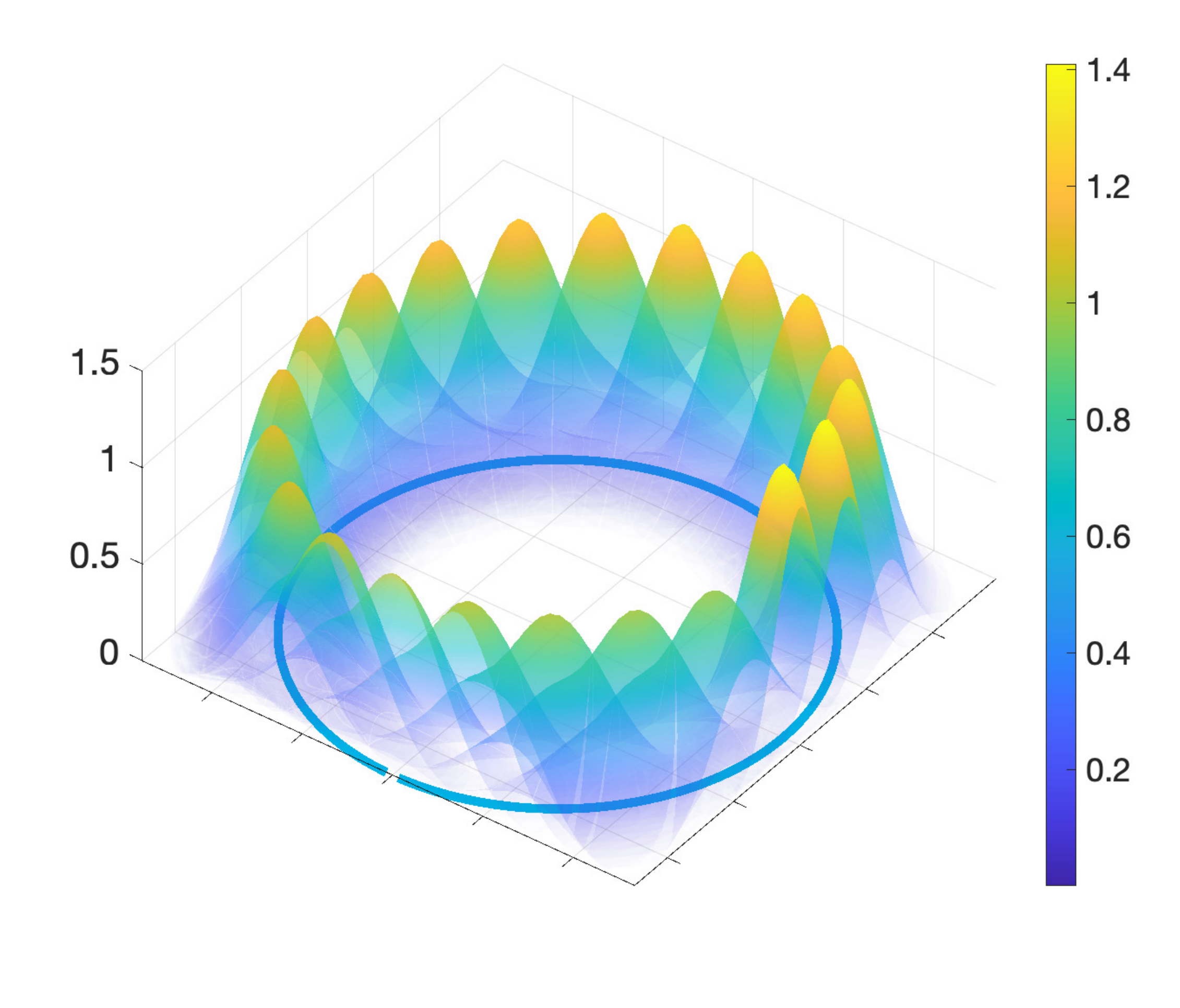}\\
    \includegraphics[width=6cm]{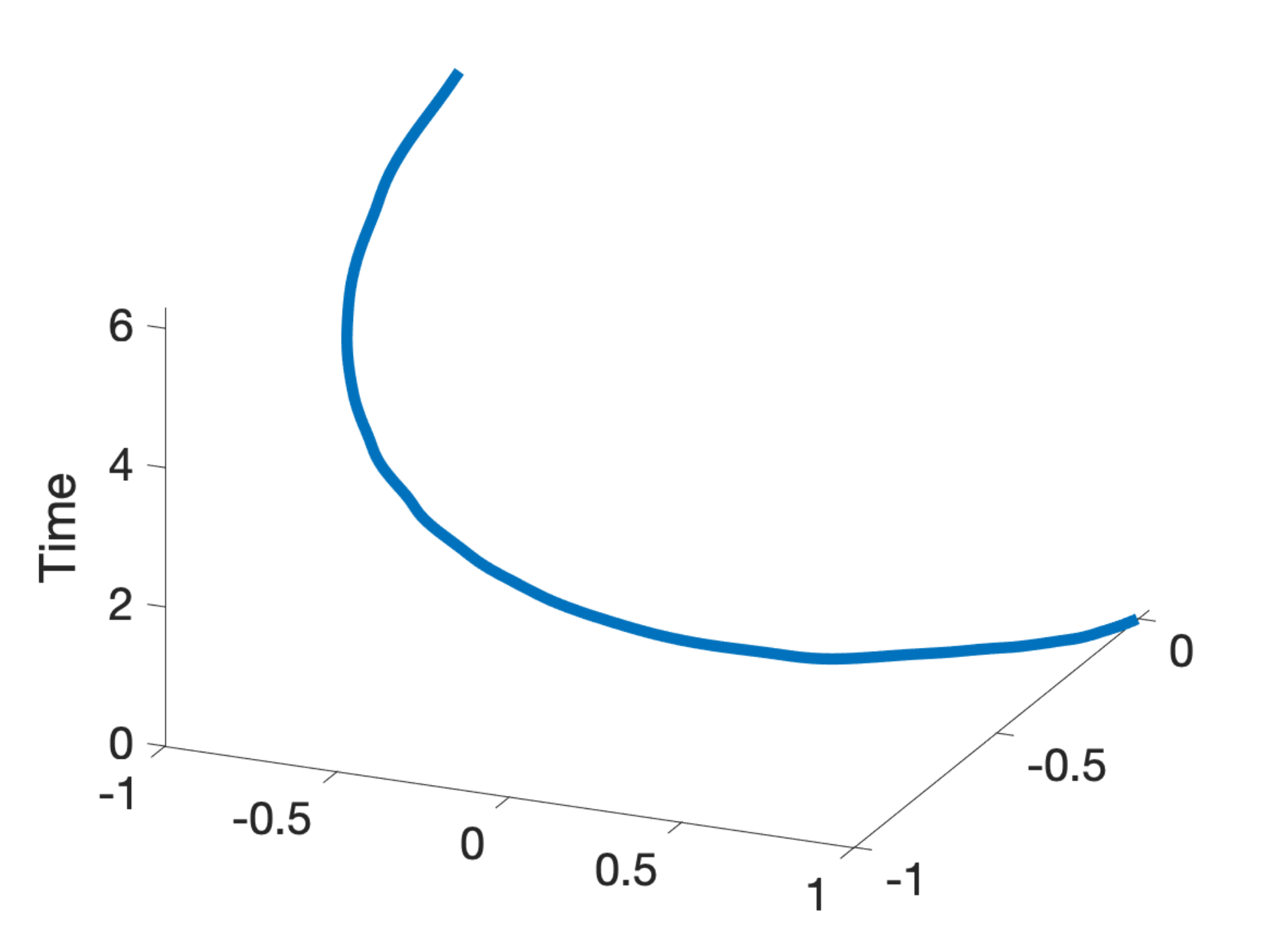}\includegraphics[width=6cm]{phase_AB_alexis.pdf}
    \caption{Aharonov--Bohm effect with $\varepsilon=3/40$, $\kappa(x) = \frac{\vert x \vert^2-1}{2}$ and $A(x) = \frac{\Phi}{2\pi |x|}$. The left and right panels correspond to $\Phi=0$ and $\Phi=2\pi$, respectively. The bottom panels are plots of $t \mapsto (\cos \varphi_t,\sin \varphi_t,t)$, where $\varphi_t$ is the phase of the top spinor component. The case $\Phi=2\pi$ induces $1/\epsi \simeq 13$ revolutions of the phase as the wavepacket travels once around the circle.
    }
    \label{fig:AB-phase}
\end{figure}

Consider a  magnetic vector potential with flux $\Phi>0$ given in polar coordinates by 
\begin{equation}
\label{eq:vec_pot}
    A(r,\theta) = \frac{\Phi}{2\pi r} {e}_{\theta}
\end{equation}
and a circular interface $\Gamma$ given by $|x|=R>0$.
Note that $B=\nabla\times A=\Phi\delta_0$ vanishes away from the origin and in particular in the vicinity of $\Gamma$. However, the wavepacket still feels a magnetic effect:  after a full revolution around $\Gamma$, it acquires according to \eqref{eq-3c} a measurable phase-shift $\Phi/\epsi$ that cannot be gauged away. This is the Aharonov--Bohm effect.

\iffalse________________________ \alexis{I am not sure why we need this paragraph?}
We now explicitly construct the gauge $\chi$ of \S\ref{}. From $\partial_\tau \chi=A_\tau$ along $\Gamma$, we obtain in polar coordinates $\chi(R,\theta)=\frac {\Phi}{2\pi} \theta$; then using $\partial_n\chi=A_r=0$, we find
\[
    \chi(r,\theta) = \frac {\Phi}{2\pi} \theta.
\] 
Note that this gauge is not globally defined along $\Gamma$: it jumps by $\Phi$ (the flux enclosed by $\Gamma$) after a full revolution. This motivated the construction of the time-dependent gauge $\chi(t,x)$ that follows the wavepacket in \S\ref{sec:}.
___________________________________ \fi


\begin{figure}[t]
    \centering
    \includegraphics[width=7.5cm]{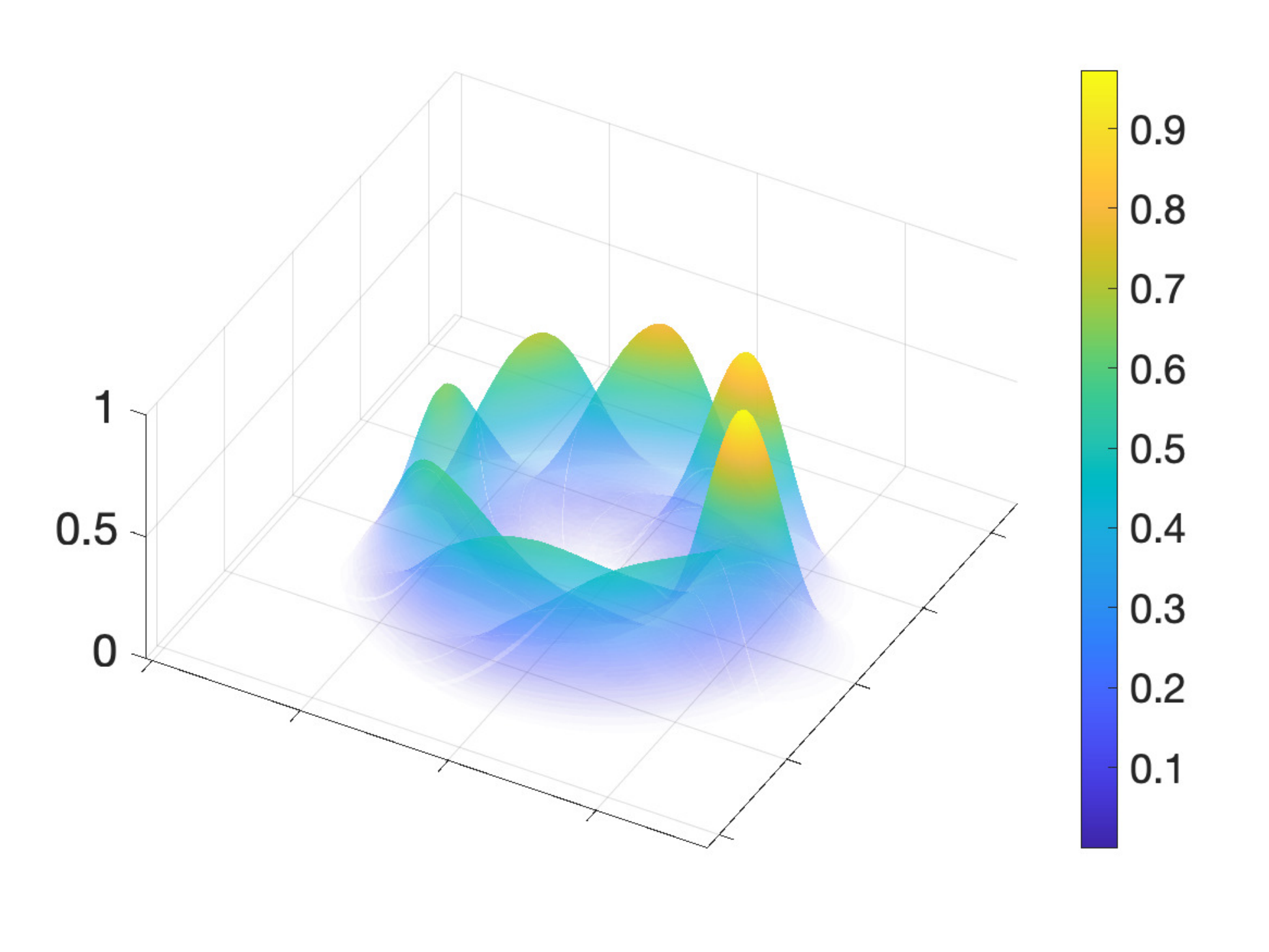} 
    \includegraphics[width=7.5cm]{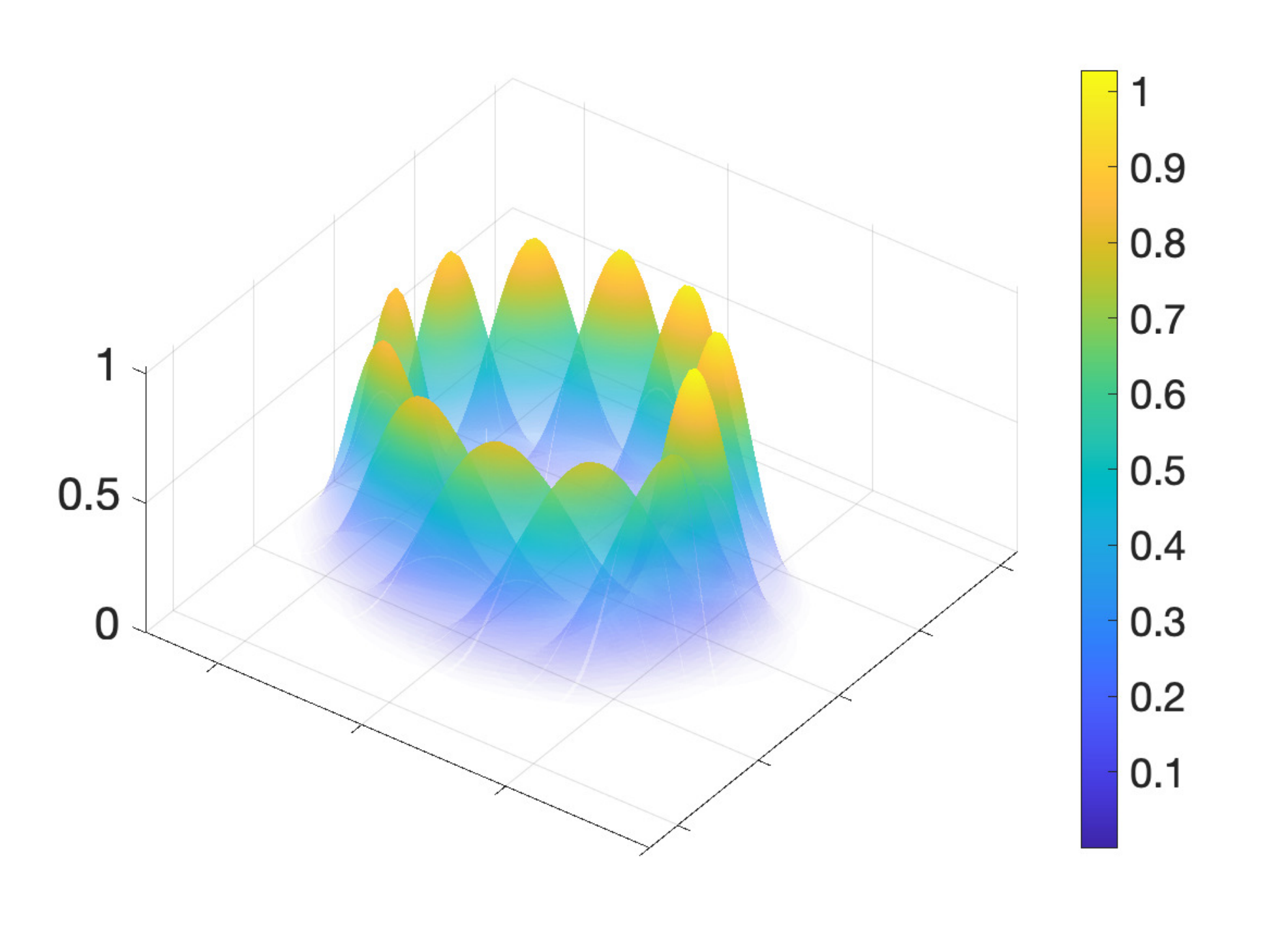}
    \caption{$\epsi = 0.05$. Snapshots showing one revolution of wavepacket, starting at 4 o'clock, on circular edge in different constant magnetic fields $B=1/\sqrt{2}$ (left, enhanced spreading) and $B=3/2$ (right, reduced spreading) with common interface $\kappa(x) =\frac{\vert x\vert^2-1}{2}.$ 
        }
    \label{fig:dispcurve}
\end{figure}

\subsection{Dispersive and phase effects in closed interfaces}\label{sec-7.3}

The coefficient $Q_\sigma$ in \eqref{eq:gaussprofile} controls the dispersion. The only term there that can grow with $t$ is 
\begin{equation}
    \nu_t e^{2\mu_t} = 2 \int_0^t \dfrac{c_t^2}{c_s^2} \left( \dot{\te_s} \gamma_s + k_s \gamma_s^2 \right) ds.
\end{equation}

Consider now a circle of radius $R$ and a choice of domain wall $\kappa(x)=\frac{|x|^m-R^m}{mR^{m-1}}$ for $m>0$ with $|\nabla\kappa|=1=r$ and $\Delta \kappa =m|x|^{-1}$ equal to $mR^{-1}$ on $\Gamma$.  Assume $B$ constant so that all coefficients are independent of time and given by
\[
  c=\frac{1}{\sqrt{1+B^2}}, \quad k=-\frac{cBm}{2R}, \quad  \dot\theta=\frac{c}{R}, \quad \dot\theta+\gamma k = \frac{c}{R}\left(1-\frac{mB^2}{2(1+B^2)}\right).
\]
We thus observe that $\nu_t=2t\gamma(\dot\theta+\gamma k)$ grows linearly in time provided that $\dot\theta+\gamma k \neq 0$. In this case, the resulting wavepacket decreases like $t^{-1/2}$, in a way depending on $B$ (higher magnetic fields, however, do not necessarily enhance dispersion). This was predicted in Lemma \ref{lem-1i} and  \eqref{eq:gaussprofile} and is numerically confirmed in Figure  \ref{fig:dispcurve}. 

In the other hand, when $m=\frac{2(1+B^2)}{B^2}$, the resulting domain wall prevents dispersion; see Figure \ref{fig:nodisp}. This can be of interest in application where one wants to slow down propagation without losing on coherence.

\begin{figure}[htbp]
    \centering
    \includegraphics[height=5cm]{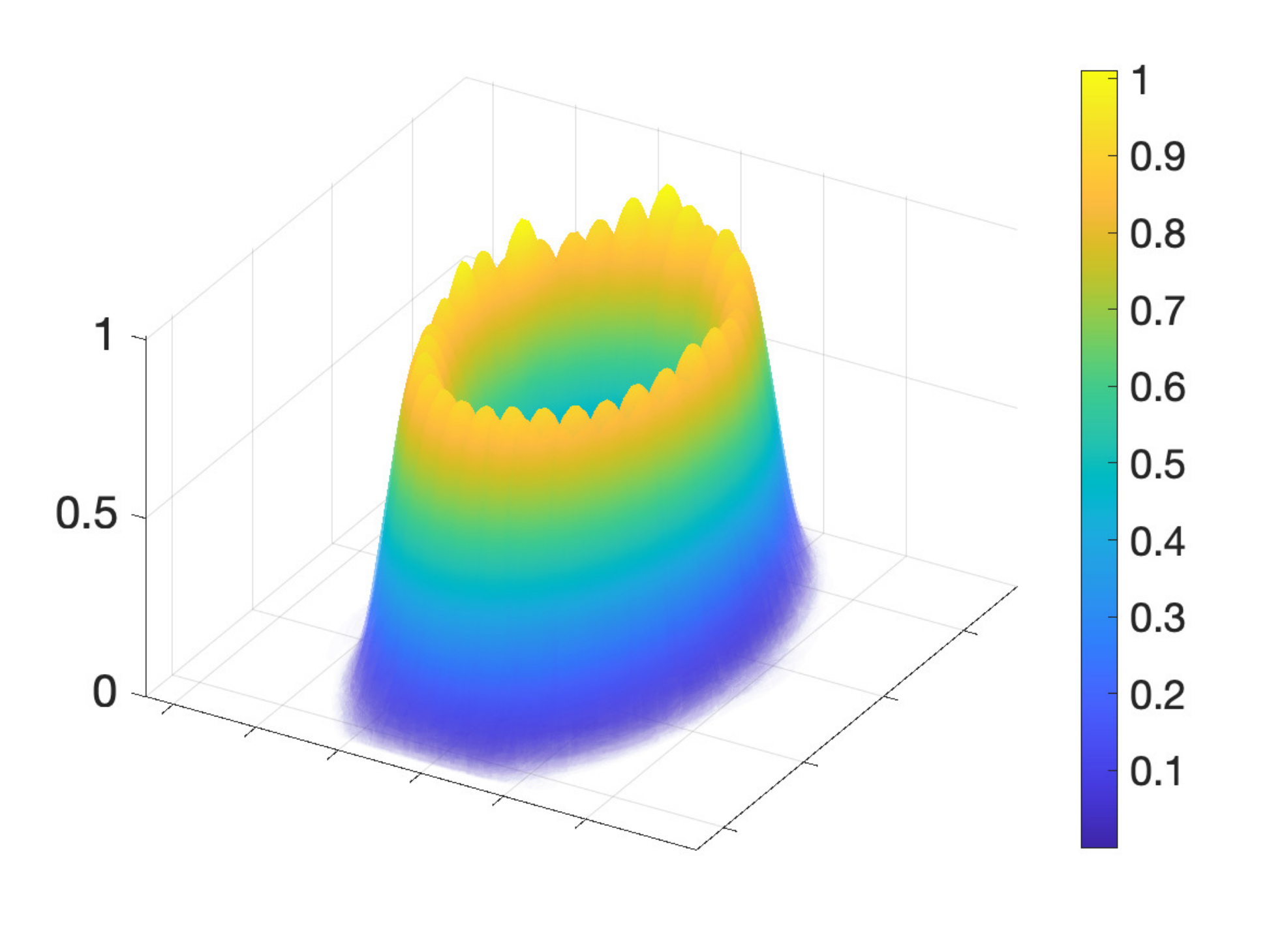}
    \includegraphics[height=5cm]{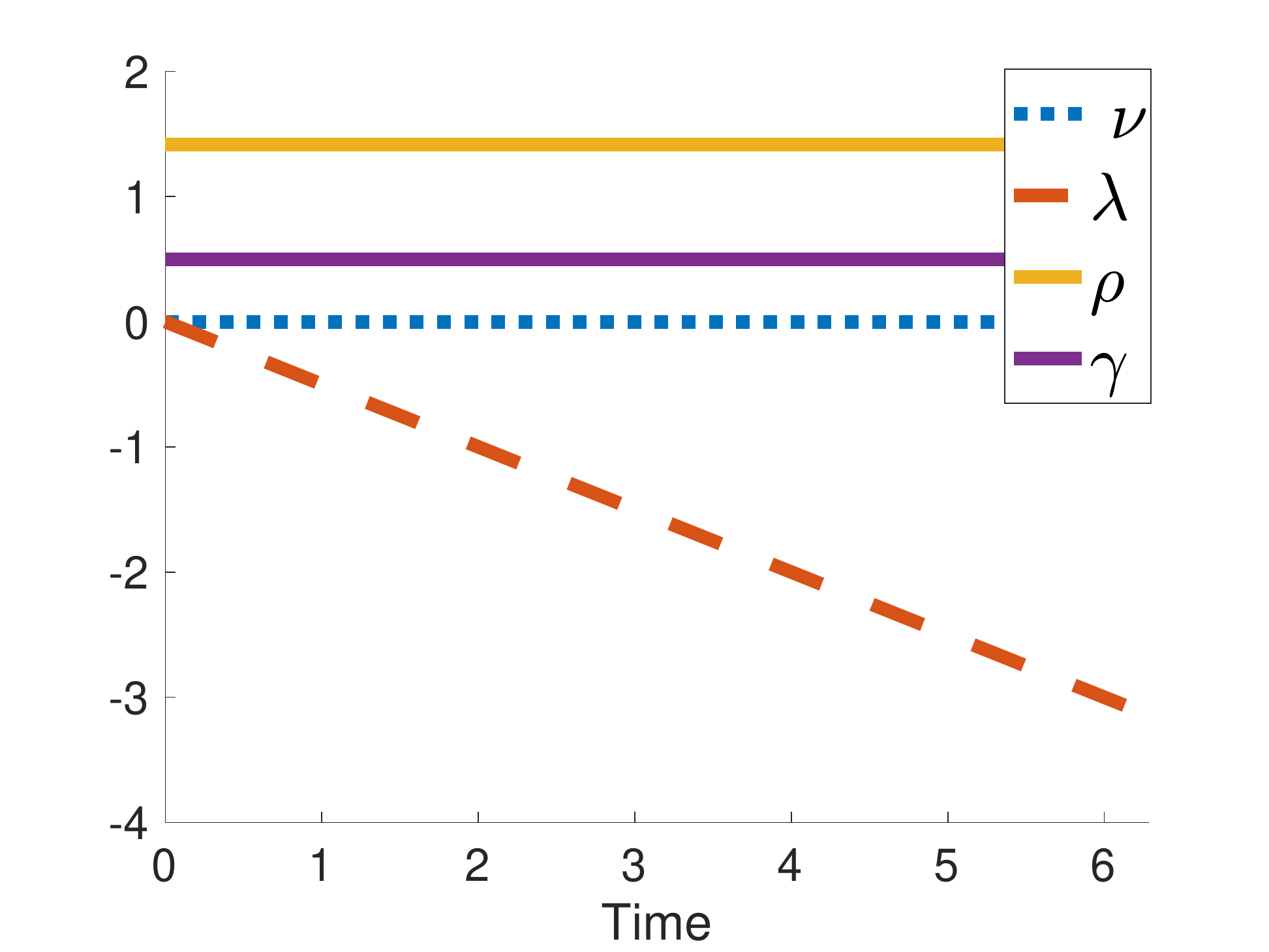}
    \caption{Homogeneous magnetic field strength $B=1$ with $\kappa(x) =\frac{\vert x\vert^4-1}{4}$. Snapshots showing an almost dispersion-free propagation around the circle.}
    \label{fig:nodisp}
\end{figure}


In this scenario we can explicitly construct the local gauge $\chi$. We have $n=e_r$ and $\tau=-e_\theta$, by using the defining relations $\partial_\tau\chi=A_\tau$ on $\Gamma$ followed by $\partial_n\chi=A_n$ across it. We have
\begin{equation}
    A=Bx_1e_2 = \frac {Br}2\big( (1+\cos 2\theta) e_\theta + \sin 2\theta e_r\big) 
\end{equation}
so that integrating along $\Gamma$, then across $\Gamma$, we obtain \begin{equation}
    \chi(R,\theta) = \frac{BR^2}2 \left(\theta+\frac{\sin 2\theta}2\right), \qquad \chi(r,\theta) = BR^2\frac {\theta} 2 + Br^2 \frac{\sin 2\theta}4.
\end{equation}
This shows that $\chi(r,\theta)$ is not globally defined as a continuous function on $\Rm^2$: the term $BR^2\frac {\theta} 2$ jumps after each revolution. The increment $\pi B R^2$ is the magnetic flux: we retrieve a Aharonov--Bohm effect. 

\subsection{Dispersive and (limited) compression effect in varying magnetic fields} 

We consider here the setting of a flat interface (with $\kappa(x)=x_2$) but with a varying magnetic field. We then observe compression and dispersive effects consistent with \eqref{eq:gaussprofile}.

\begin{figure}[htbp]
    \centering
    \includegraphics[width=8cm,height=6cm]{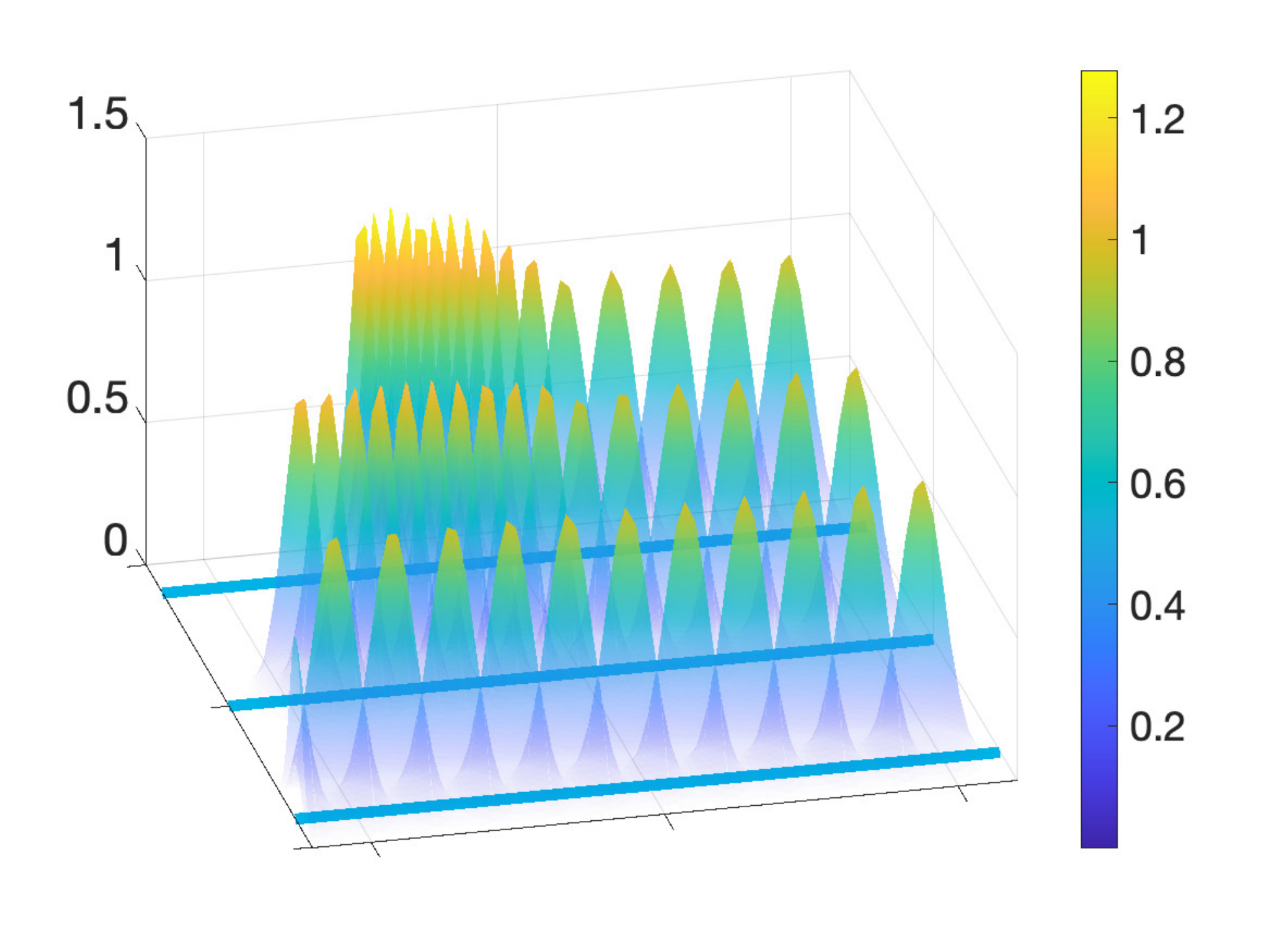} \includegraphics[width=7cm]{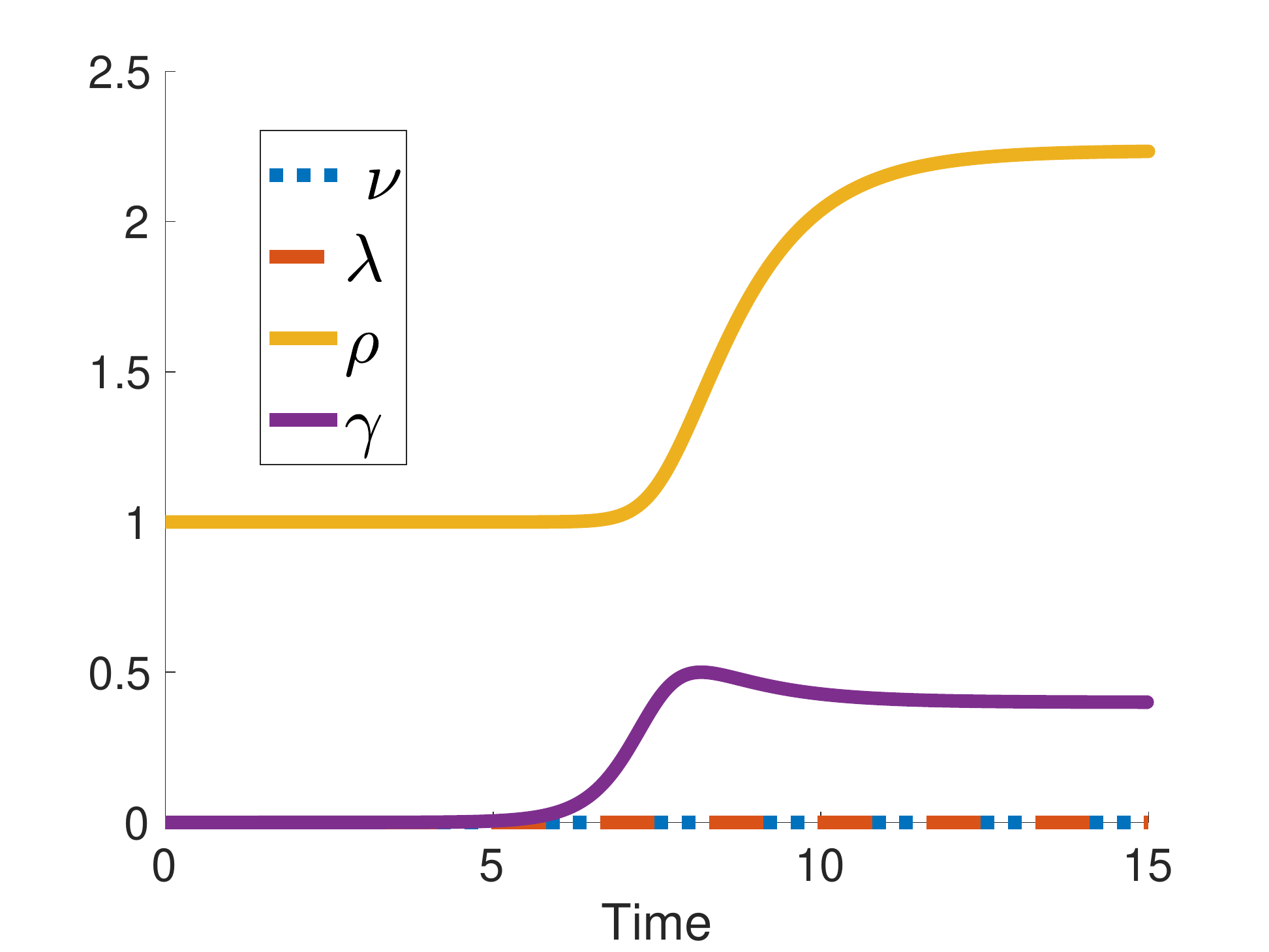}
    \caption{Snapshots of wavepacket, starting on the right, for a potential $A_1(x) = -B_0x_2 (1-\tanh(x_1-2))$ with $B_0 =0,1,2$ from bottom to top, which corresponds to a magnetic field $B(x) = B_0 (1-\tanh(x_1-2))$, on a straight interface $\kappa(x)=x_2$. Right figure with $B_0=1.$}
    \label{fig:varyB2}
\end{figure} 
In Figure \ref{fig:varyB2}, the intensity of the magnetic field increases as the wavepacket propagates along $\Gamma$. We thus expect an increase in $\rho_t$ and as a result a compression of the wavepacket. This is confirmed by the numerical simulations of Figure \ref{fig:varyB2}.


\begin{figure}[htbp]
    \centering
    \includegraphics[width=8cm,height=6cm]{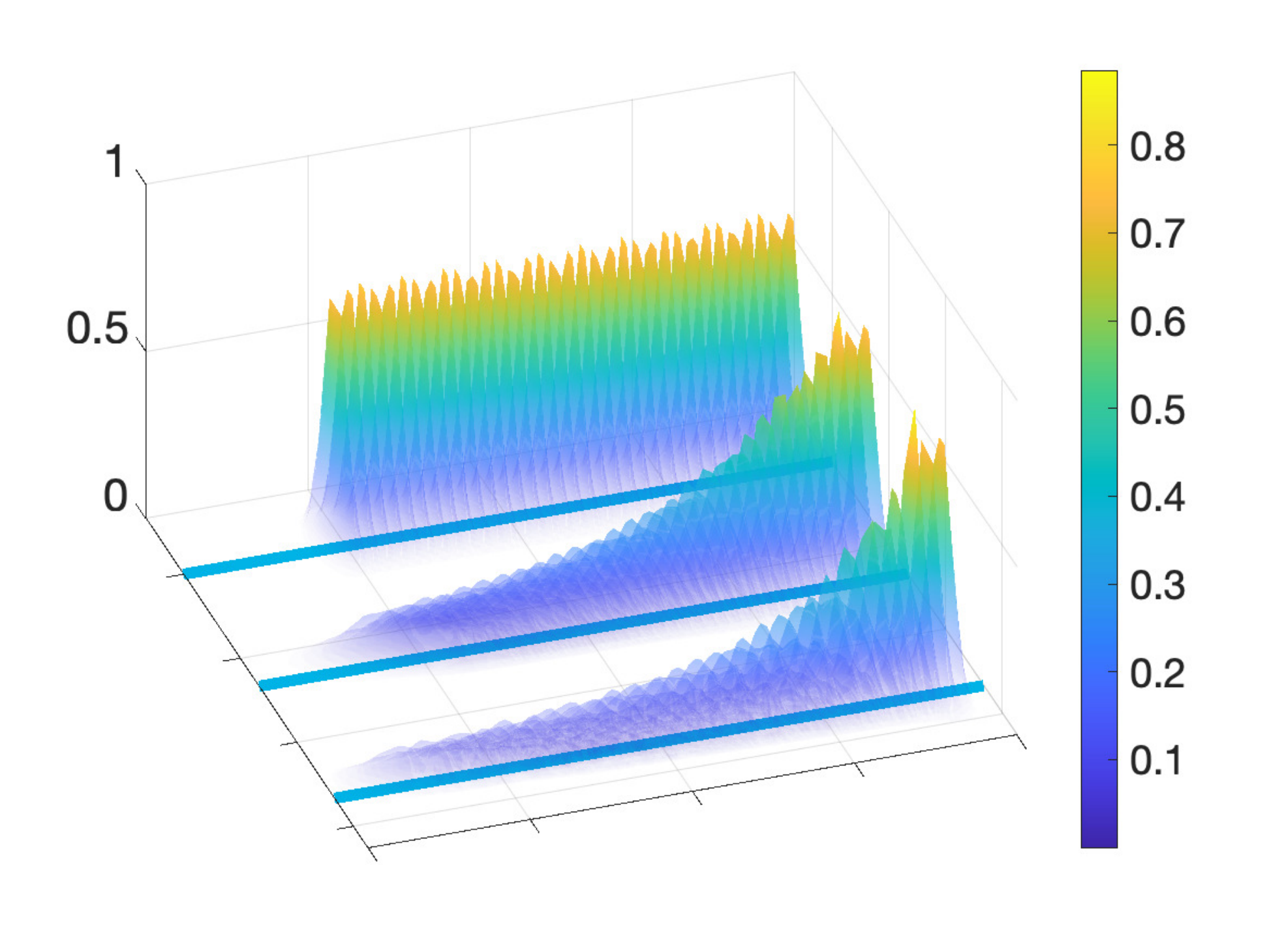} 
    \includegraphics[width=8cm,height=6cm]{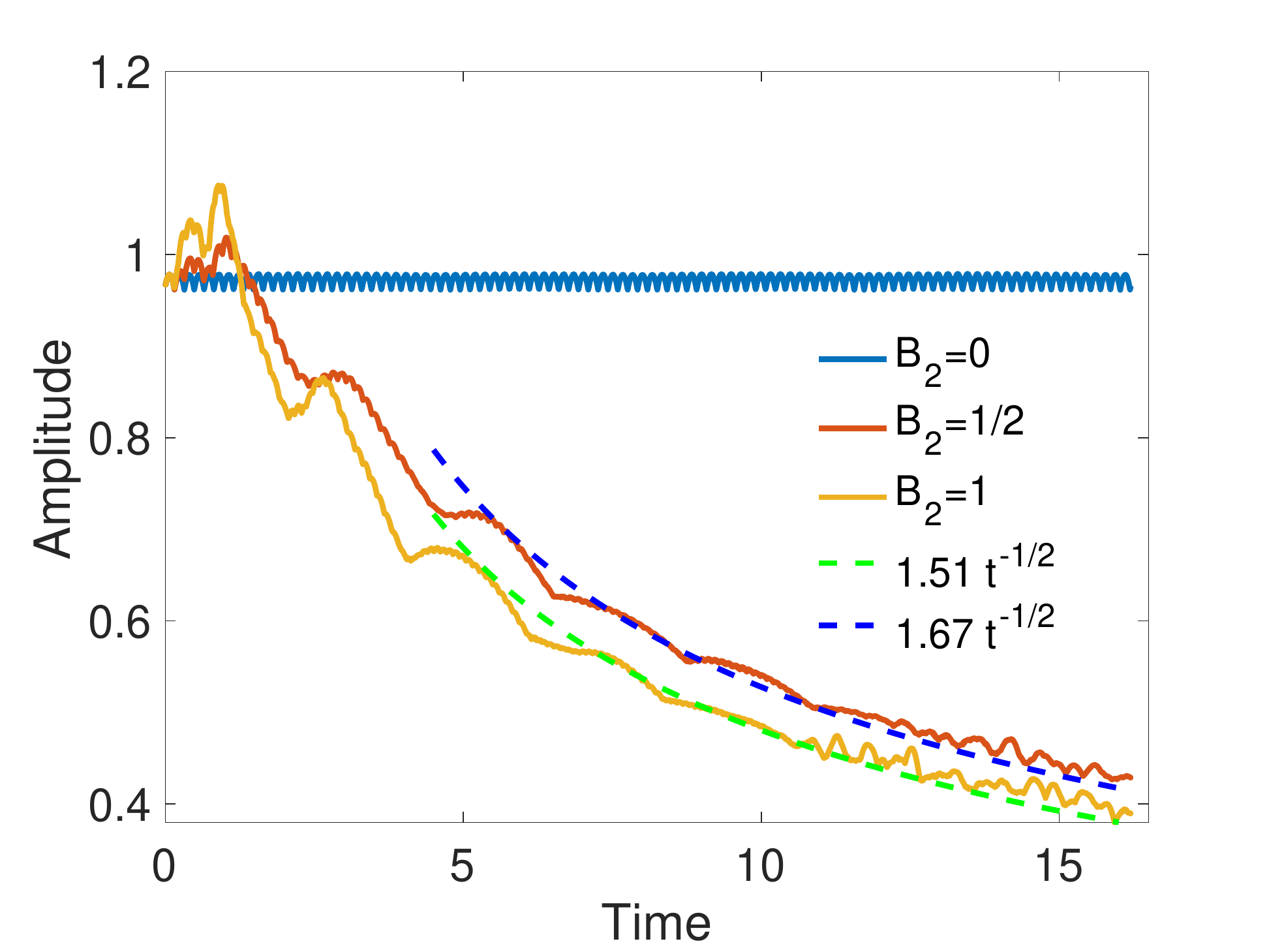} \\
    \caption{$\varepsilon=0.075$, $B=1+4B_2x_2$ with $B_2 =0,0.5,1$ and $\kappa(x)=x_2$. The left figure shows the snapshots of the three wavepackets, starting on the right, and on the right, we see the power-law decay in the non-constant magnetic field. The right figure shows the $L^{\infty}$-decay of the amplitudes. 
   } \label{fig:varyingB}
\end{figure} 

We next consider the setting of a magnetic field that increases transversely to $\Gamma$: $B(x)=1+4B_2x_2$. Set $A=\tilde A = -(1+2B_2x_2) x_2 e_1$, resulting in $\beta=1+2B_2x_2$; in particular $\partial_n\beta=2B_2$ and $k=k_t=cB_2$ is constant. This shows that $\nu_t$ grows linearly with time. The wavepacket decays like $t^{-1/2}$, as confirmed numerically in Figure \ref{fig:varyingB}.

The amplitude drop generated by dispersion is, however, reversible. For instance, the magnetic field $B(x)=1+4\cos(2\pi x_1 / 15)x_2$  generates time-dependent oscillations in $\nu_t$: this coefficient is proportional to $\sin(2\pi t/15)$. There is no dispersion for $t \in 15\Z$, as shown in Figure \ref{fig:decom_com}. 

\begin{figure}[htbp]
    \centering
    \includegraphics[height=4.5cm]{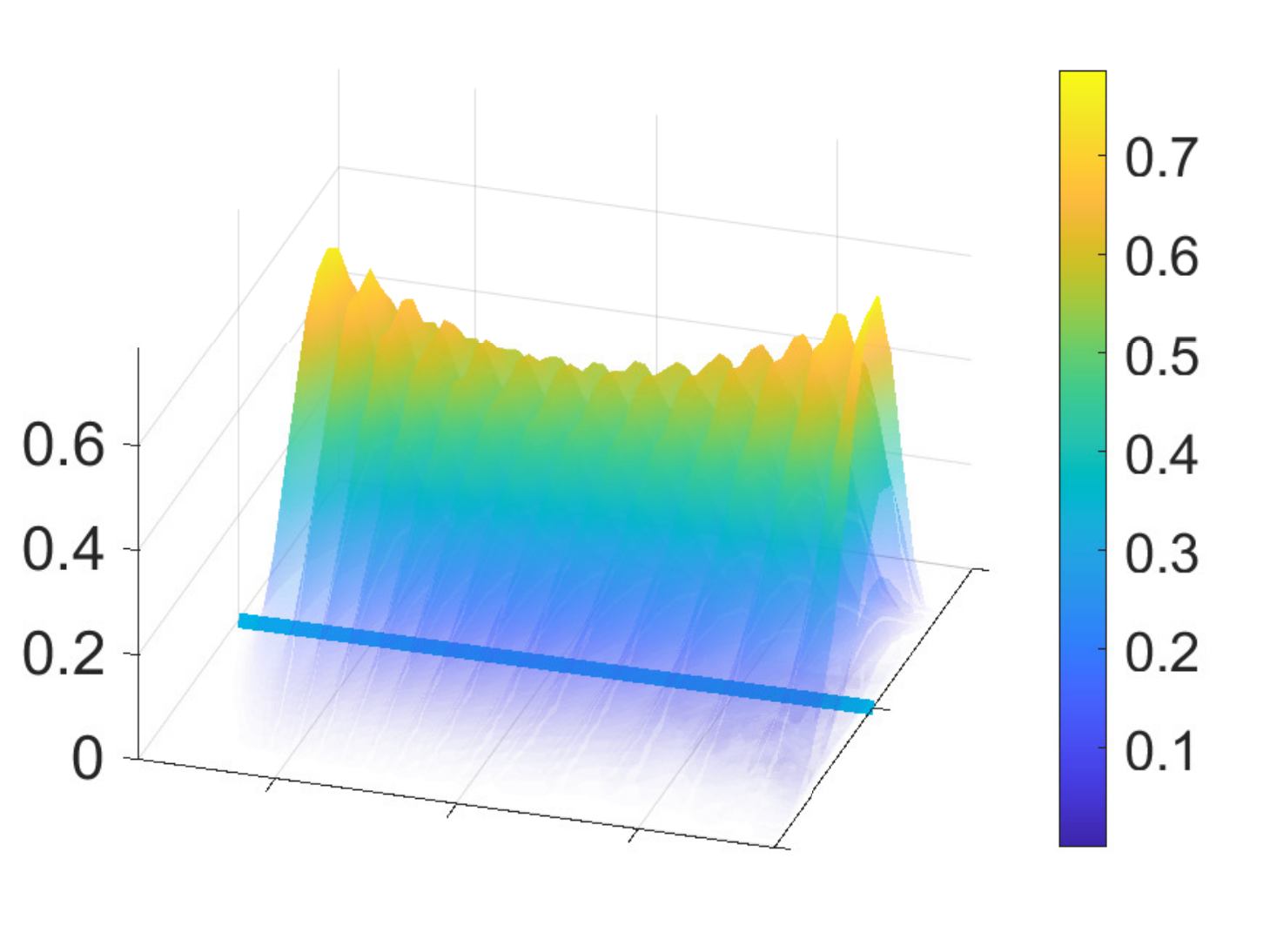} 
    \includegraphics[height=4.5cm]{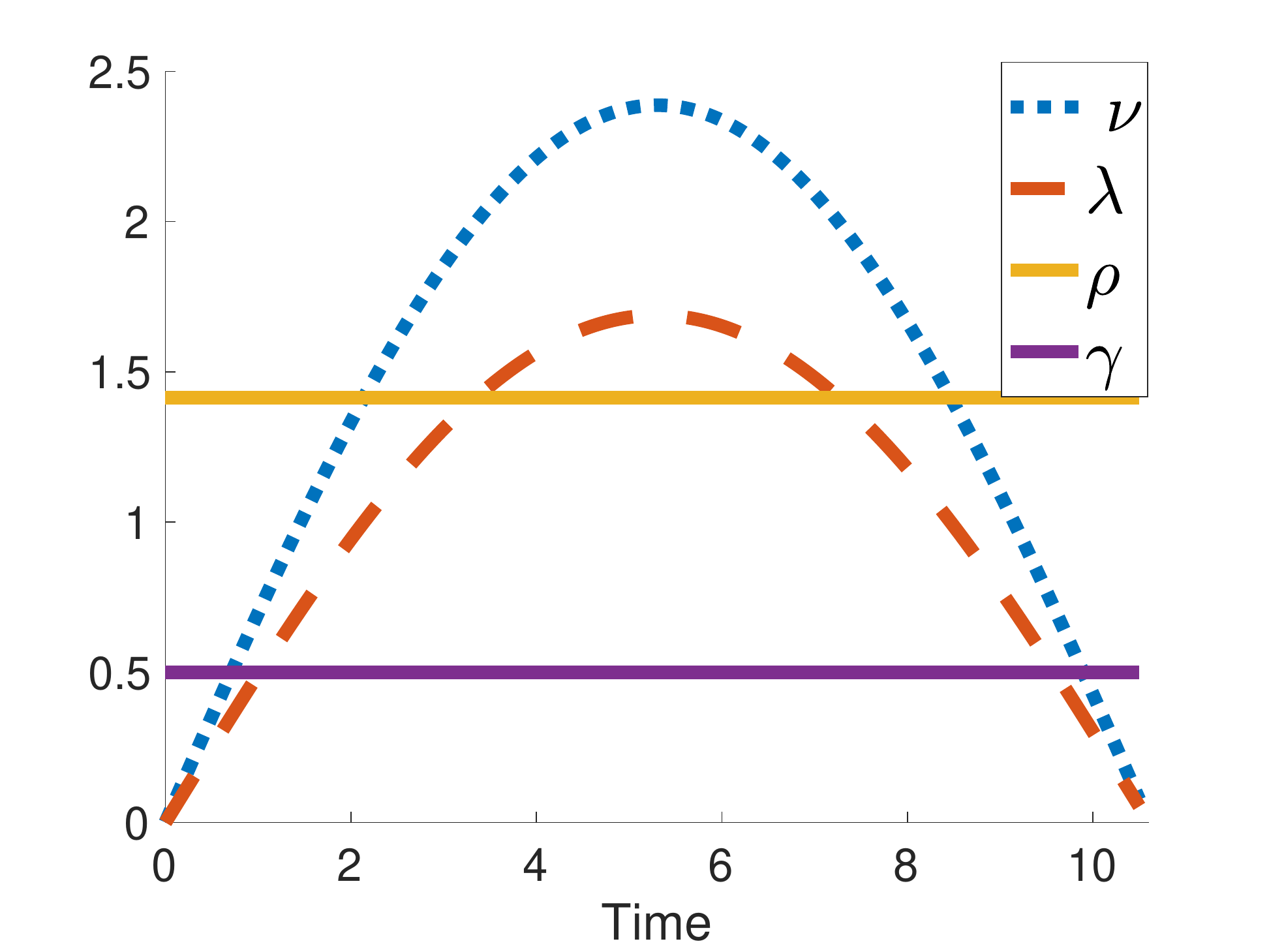}
    \caption{$\varepsilon=0.1$, snapshots of a wavepacket, starting on the right, for a straight interface $\kappa(x)=x_2$ with $B(x)=1+4\cos(\tfrac{2\pi x_1}{15})x_2$, a periodic modulation of the field in Figure \ref{fig:varyingB}, sampled over half a period of the cosine. The reversible amplitude drop caused by dispersion is clearly visible.}
    \label{fig:decom_com}
\end{figure}

\FloatBarrier



\bibliographystyle{amsxport} 
\bibliography{edgemode}

\begin{bibdiv}
\begin{biblist}

\bib{avron1994}{article}{
      author={Avron, Joseph~E.},
      author={Seiler, Ruedi},
      author={Simon, Barry},
       title={Charge deficiency, charge transport and comparison of
  dimensions},
        date={1994},
     journal={Comm. Math. Phys.},
      volume={159},
      number={2},
       pages={399\ndash 422},
         url={http://projecteuclid.org/euclid.cmp/1104254604},
}

\bib{B19b}{article}{
      author={Bal, Guillaume},
       title={Continuous bulk and interface description of topological
  insulators},
        date={2019},
     journal={Journal of Mathematical Physics},
      volume={60},
      number={8},
       pages={081506},
}

\bib{B19a}{article}{
      author={Bal, Guillaume},
       title={Topological protection of perturbed edge states},
        date={2019},
     journal={Communications in Mathematical Sciences},
      volume={17},
      number={1},
       pages={193\ndash 225},
}

\bib{B20}{article}{
      author={Bal, Guillaume},
       title={Topological invariants for interface modes},
        date={2020},
      eprint={arXiv:1906.08345},
}

\bib{bal2021topological}{article}{
      author={Bal, Guillaume},
       title={Topological charge conservation for continuous insulators},
        date={2021},
      eprint={arXiv:2106.08480},
}

\bib{bal2021edge}{article}{
      author={Bal, Guillaume},
      author={Becker, Simon},
      author={Drouot, Alexis},
      author={Kammerer, Clotilde~Fermanian},
      author={Lu, Jianfeng},
      author={Watson, Alexander},
       title={Edge state dynamics along curved interfaces},
        date={2021},
      eprint={arXiv:2106.00729},
}

\bib{Be13}{book}{
      author={Bernevig, Bogdan~Andrei},
       title={Topological insulators and topological superconductors},
   publisher={Princeton University Press},
        date={2013},
        ISBN={9781400846733},
         url={https://doi.org/10.1515/9781400846733},
}

\bib{bal2021}{article}{
      author={Bal, Guillaume},
      author={Massatt, Daniel},
       title={Multiscale invariants of floquet topological insulators},
        date={2022},
     journal={To appear in SIAM Multiscale Modeling and Simulation},
      eprint={arXiv:2101.06330},
}

\bib{BES94}{article}{
      author={Bellissard, Jean},
      author={van Elst, Andreas},
      author={Schulz-Baldes, Hermann},
       title={The noncommutative geometry of the quantum {H}all effect},
        date={1994},
     journal={Journal of Mathematical Physics},
      volume={35},
      number={10},
       pages={5373\ndash 5451},
}

\bib{Col04}{article}{
      author={Colin~de Verdi\`ere, Yves},
       title={The level crossing problem in semi-classical analysis. {II}.
  {T}he {H}ermitian case},
        date={2004},
     journal={Annales de l'Institut Fourier},
      volume={54},
      number={5},
       pages={1423\ndash 1441},
}

\bib{DH72}{article}{
      author={Duistermaat, J.J.},
      author={Hörmander, L.},
       title={Fourier integral operators. {II}.},
        date={1972},
     journal={Acta Math.},
      volume={128},
      number={3-4},
       pages={183–269},
}

\bib{Drouot:19b}{article}{
      author={Drouot, Alexis},
       title={The bulk-edge correspondence for continuous honeycomb lattices},
        date={2019},
     journal={Communication in Partial Differential Equations},
      volume={44},
      number={12},
       pages={1406–1430},
}

\bib{Drouot:19}{article}{
      author={Drouot, Alexis},
       title={Characterization of edge states in perturbed honeycomb
  structures},
        date={2019},
     journal={Pure and Applied Analysis},
      volume={1},
      number={3},
       pages={385–445},
}

\bib{Drouot2020microlocal}{article}{
      author={Drouot, Alexis},
       title={Microlocal analysis of the bulk-edge correspondence},
        date={2021},
     journal={Communications in Mathematical Physics},
      volume={383},
       pages={2069–2112},
}

\bib{Drouot:21}{article}{
      author={Drouot, Alexis},
       title={Ubiquity of conical points in topological insulators},
        date={2021},
     journal={Journal de l'Ecole Polytechnique},
      volume={8},
       pages={507\ndash 532},
}

\bib{DW20}{article}{
      author={Drouot, Alexis},
      author={Weinstein, MI},
       title={Edge states and the valley {H}all effect},
        date={2020},
     journal={Advances in Mathematics},
      volume={368},
       pages={107142},
}

\bib{EG02}{article}{
      author={Elbau, Peter},
      author={Graf, Gian-Michele},
       title={Equality of bulk and edge {H}all conductance revisited},
        date={2002},
     journal={Communications in mathematical physics},
      volume={229},
      number={3},
       pages={415\ndash 432},
}

\bib{FG2}{article}{
      author={Fermanian~Kammerer, Clotilde},
      author={G\'erard, Patrick},
       title={A {L}andau-{Z}ener formula for non-degenerated involutive
  codimension three crossings},
        date={2003},
     journal={Ann. Henri Poincar\'e},
      volume={4},
       pages={123\ndash 168},
}

\bib{FLW16}{article}{
      author={Fefferman, Charles~L},
      author={Lee-Thorp, James~P},
      author={Weinstein, Michael~I},
       title={Edge states in honeycomb structures},
        date={2016},
     journal={Annals of PDE},
      volume={2},
      number={2},
       pages={12},
}

\bib{GP}{article}{
      author={{Graf}, Gian-Michele},
      author={{Porta}, Marcello},
       title={Bulk-edge correspondence for two-dimensional topological
  insulators},
        date={2013},
     journal={Communications in Mathematical Physics},
      volume={324},
      number={3},
       pages={851\ndash 895},
}

\bib{GRW21}{article}{
      author={Guglielmon, J.},
      author={Rechtsman, M.~C.},
      author={Weinstein, M.~I.},
       title={Landau levels in strained two-dimensional photonic crystals},
        date={2021Jan},
     journal={Phys. Rev. A},
      volume={103},
       pages={013505},
         url={https://link.aps.org/doi/10.1103/PhysRevA.103.013505},
}

\bib{Hag94}{article}{
      author={Hagedorn, George~A.},
       title={Molecular propagation through electron energy level crossings},
        date={1994},
     journal={Memoirs of the AMS},
      volume={111},
      number={536},
}

\bib{HJ98}{article}{
      author={Hagedorn, George~A},
      author={Joye, Alain (F-CNRS-T)},
       title={{Landau--Zener transitions through small electronic eigenvalue
  gaps in the Born--Oppenheimer approximation.}},
        date={1998},
     journal={Ann. Inst. H. Poincar\'e Phys. Th\'eor.},
      volume={68},
      number={1},
       pages={85\ndash 134},
}

\bib{moessner2021topological}{book}{
      author={Moessner, Roderich},
      author={Moore, Joel~E},
       title={Topological phases of matter},
   publisher={Cambridge University Press},
        date={2021},
}

\bib{PSB16}{book}{
      author={Prodan, Emil},
      author={Schulz-Baldes, Hermann},
       title={Bulk and boundary invariants for complex topological insulators},
   publisher={Springer verlag, Berlin},
        date={2016},
}

\bib{TH}{book}{
      author={Thaller, Bernd},
       title={{T}he {D}irac equation},
      series={Texts and monographs in physics},
   publisher={Springer-Verlag},
        date={1992},
        ISBN={9783540548836},
         url={https://books.google.ch/books?id=X7XvAAAAMAAJ},
}

\bib{TKNN}{article}{
      author={Thouless, D.~J.},
      author={Kohmoto, M.},
      author={Nightingale, M.~P.},
      author={den Nijs, M.},
       title={{Quantized Hall Conductance in a Two-Dimensional Periodic
  Potential}},
        date={1982Aug},
     journal={Phys. Rev. Lett.},
      volume={49},
       pages={405\ndash 408},
         url={http://link.aps.org/doi/10.1103/PhysRevLett.49.405},
}

\bib{VO}{book}{
      author={{Volovik}, Grigory~E.},
       title={Nonlinear phenomena in condensed matter: {U}niverse in a {H}elium
  droplet},
        date={1989},
}

\bib{WI}{article}{
      author={{Witten}, Edward},
       title={{Three lectures on topological phases of matter}},
        date={2016},
     journal={Nuovo Cimento Rivista Serie},
      volume={39},
      number={7},
       pages={313\ndash 370},
}

\end{biblist}
\end{bibdiv}


\newpage

\begin{thebibliography}{0}
\bibitem[AHV13]{AHV} Avila,J.C.,Schulz-Baldes,H.  and Villegas-Blas,C. \emph{Topological invariants of edge states for periodic two-dimensional models.} Math. Phys.,
Anal. Geom., 16:136–170, 2013.
\bibitem[B19]{B19} G. Bal, Continuous bulk and interface description of topological insulators. 
J. Math. Phys. 60 (2019), no. 8, 081506, 20 pp.
\bibitem[BKR17]{BKR}Bourne, C., Kellendonk, J. and Rennie, A. The K-Theoretic Bulk–Edge Correspondence for Topological Insulators. Ann. Henri Poincaré 18, 1833–1866 (2017)
\bibitem[FLW16]{FLT} Fefferman, C.L., Lee-Thorp, J.P. and Weinstein, M.I. \emph{Edge States in Honeycomb Structures.} Ann. PDE 2, 12 (2016).
\bibitem[GP13]{GP} Graf, G.M., Porta, M. Bulk-Edge Correspondence for Two-Dimensional Topological Insulators. Commun. Math. Phys. 324, 851–895 (2013).
\bibitem[L16]{L16}Li, X. et al. .\emph{Experimental Observation of Topological Edge States at the Surface Step Edge of the Topological Insulator ${\mathrm{ZrTe}}_{5}$},Phys. Rev. Lett.,116,17,176803, 2016.
\end{thebibliography}

\end{document}